\documentclass[12pt]{amsart}  %document style grand article journal
\usepackage{graphicx}%g�re les graphiques
\usepackage{here} %place les photos o� il faut
\usepackage{array} %faire des tableaux plus compliqu�s

\usepackage{vmargin}
\usepackage{amssymb}
\usepackage{mathrsfs}
\usepackage[all]{xy}
\usepackage[usenames,dvipsnames]{color}%64 couleurs
\usepackage{soul}
\RequirePackage[colorlinks,linkcolor=blue,citecolor=LimeGreen,urlcolor=red]{hyperref} %ref actives en couleur au lieu de carres
\usepackage{amsmath}
\newtheorem{theorem}{Theorem}[section]

\newtheorem{lemma}[theorem]{Lemma}

\newtheorem{prop}[theorem]{Proposition}
\newtheorem{remark}[theorem]{Remark}

\numberwithin{equation}{section}

\newcommand{\R}{\mathbb{R}}

\newcommand{\N}{\mathbb{N}}

\newcommand{\abs}[1]{\left|#1\right|}
\newcommand{\eps}{\varepsilon}
\newcommand{\norm}[1]{\left\|#1\right\|}

\renewcommand{\leq}{\leqslant}
\renewcommand{\geq}{\geqslant}

\renewcommand{\bar}{\overline}

\renewcommand{\tilde}{\widetilde}

\def\signmb{\bigskip \begin{center} {\sc
Marc Briant\par\vspace{3mm}
Brown University\par
Division of Applied Mathematics\par
182 George Street, Box F,
Providence, RI 02192, USA\par
\vspace{3mm}   
e-mail:} \tt{briant.maths@gmail.com} \end{center}}

\def\signae{\bigskip \begin{center} {\sc
Amit Einav\par\vspace{3mm}
University of Cambridge\par
DPMMS, Centre for Mathematical Sciences\par
Wilberforce Road,
Cambridge CB3 0WA,
UK\par\vspace{3mm}
e-mail:} \tt{a.einav@dpmms.cam.ac.uk} \end{center}}

\begin{document} 

\title[CAUCHY PROBLEM FOR HOMOGENEOUS BOSONIC BOLTZMANN-NORDHEIM]{ON THE CAUCHY PROBLEM FOR THE HOMOGENEOUS BOLTZMANN-NORDHEIM EQUATION FOR BOSONS}
\author{Marc Briant}
\author{Amit Einav}
\thanks{The first author was supported by EPSRC grant EP/H023348/1 for the Cambridge Centre for Analysis, and by the $150^{th}$ Anniversary Postdoctoral Mobility Grant of the London Mathematical Society.}
\thanks{The second author was supported by EPSRC grant EP/L002302/1}
\maketitle

\begin{abstract}
The Boltzmann-Nordheim equation is a modification of the Boltzmann equation, based on physical considerations, that describes the dynamics of the distribution of particles in a quantum gas composed of bosons or fermions. In this work we investigate the Cauchy theory of the spatially homogeneous Boltzmann-Nordheim equation for bosons, in dimension $d\geq 3$. We show existence and uniqueness locally in time for any initial data in $L^\infty\left(1+\abs{v}^s\right)$ with finite mass and energy, for a suitable $s$, as well as the instantaneous creation of moments of all order.
\end{abstract}

\vspace*{10mm}
%\smallskip

\\

\textbf{Acknowledgements:} The first author would like to thank Cl\'{e}ment Mouhot for suggesting the Boltzmann-Nordheim for bosons as a possible venue of research, and Miguel Escobedo for many fruitful discussions they had. 
The authors would also like to thank the anonymous reviewers for their comments and helpful suggestions.

\tableofcontents

\section{Introduction} \label{sec:intro}

This work considers the dynamics of a distribution function of particles in a dilute homogeneous quantum bosonic gas in $\R^d$, $f(t,v)$.\\
In general, the evolution equation of particles of dilute quantum gas that undergo binary collisions is given by the so-called Boltzmann-Nordheim equation:
\begin{eqnarray*}
\partial_tf &=& Q(f) 
\\ &=& \int_{\R^d\times \mathbb{S}^{d-1}}B\left(v,v_*,\theta\right)\left[f'(1+\alpha f)f'_*(1+\alpha f_*) - f(1+\alpha f')f_*(1+\alpha f'_*)\right]\:dv_*d\sigma,
\end{eqnarray*}
with $(t,v)\in \R^{+}\times \R^d$, $\alpha\in \left\lbrace -1,1 \right\rbrace$ and where $f'$, $f_*$, $f'_*$ and $f$ are the values taken by $f$ at $v'$, $v_*$, $v'_*$ and $v$ respectively. $B$ is the collision kernel that encodes the physical properties of the collision process, and
$$\left\{ \begin{array}{rl}&\displaystyle{v' = \frac{v+v_*}{2} +  \frac{|v-v_*|}{2}\sigma} \vspace{2mm} \\ \vspace{2mm} &\displaystyle{v' _*= \frac{v+v_*}{2}  -  \frac{|v-v_*|}{2}\sigma} \end{array}\right., \: \quad \quad \mbox{cos}\:\theta = \left\langle \frac{v-v_*}{\abs{v-v_*}},\sigma\right\rangle .$$

This equation has been derived by Nordheim (see \cite{Nor}) using quantum statistical considerations. One notices that when $\alpha=0$ one recovers the Boltzmann equation, which rules the dynamics of particles in a dilute gas in classical mechanics when only elastic binary collisions are taken into account. The quantum effects manifest themselves in the fact that the probability of collision between two particles depends not only on the the number of particles undergoing the collision, but also the number of particles already occupying the final collision state. This appears in the Boltzmann-Nordheim equation in the form of the added multiplicative term where $\alpha=-1$ corresponds to fermions and $\alpha=1$ corresponds to bosons.

\bigskip
The collision kernel $B$ contains all the information about the interaction between two particles and is determined by physics. We mention, at this point, that one can derive this type of equations from Newtonian mechanics (coupled with quantum effects in the case of the Boltzmann-Nordheim equation), at least formally (see \cite{Ce} or \cite{Ce1} for the classical case and \cite{Nor} or \cite{ChapCow} for the quantum case). However, while the validity of the Boltzmann equation from Newtonian laws is known for short times (Landford's theorem, see \cite{La} or more recently \cite{GST,PSS}), we do not have, at the moment, the same kind of proof for the Boltzmann-Nordheim equation.

\bigskip

%%%%%%%%%%%%%%%%%%%%%%%%%%%%%%%%%%%%%%%%%%%%%%%%%%%%%%%%%%%%%%%%%%%%%%%%%%%%%%%%%%%%%%%%%%%%%%%%%%%%%%%%%%%%%%%%%%%%%%%%%%%%%%%
%%%%%%%%%%%%%%%%%%%%%%%%%%%%%%%%%%%%%%%%%%%%%%%%%%%%%%%%%%%%%%%%%%%%%%%%%%%%%%%%%%%%%%%%%%%%%%%%%%%%%%%%%%%%%%%%%%%%%%%%%%%%%%%
%%%%%%%%%%%%%%%%%%%%%%%%%%%%%%%%%%%%%%%%%%%%%%%%%%%%%%%%%%%%%%%%%%%%%%%%%%%%%%%%%%%%%%%%%%%%%%%%%%%%%%%%%%%%%%%%%%%%%%%%%%%%%%%

\subsection{The problem and its motivations} \label{subsec:previously}

Throughout this paper we will assume that the collision kernel $B$ can be written as
$$B(v,v_*,\theta) = \Phi\left(|v - v_*|\right)b\left( \mbox{cos}\:\theta\right),$$
which covers a wide range of physical situations (see for instance \cite{Vi2} Chapter $1$).
\par Moreover, we will consider only kernels with hard potentials, that is 
\begin{equation}\label{hardpot}
\Phi(z) = C_\Phi z^\gamma \:,\:\: \gamma \in [0,1],
\end{equation}
where $C_\Phi>0$ is a given constant. Of special note is the case $\gamma=0$ which is usually known as Maxwellian potentials.
We will assume that the angular kernel $b\circ \mbox{cos}$ is positive and continuous on $(0,\pi)$, and that it satisfies a strong form of Grad's angular cut-off:
\begin{equation}\label{binfty}
b_\infty=\norm{b}_{L^\infty_{[-1,1]}}<\infty
\end{equation}
The latter property implies the usual Grad's cut-off \cite{Gr1}:
\begin{equation}\label{lb}
l_b = \int_{\mathbb{S}^{d-1}}b\left(\mbox{cos}\:\theta\right)d\sigma = \left|\mathbb{S}^{d-2}\right|\int_0^\pi b\left(\mbox{cos}\:\theta\right) \mbox{sin}^{d-2}\theta \:d\theta < \infty.
\end{equation}
Such requirements are satisfied by many physically relevant cases. The hard spheres case ($b=\gamma=1$) is a prime example.

\bigskip
With the above assumption we can rewrite the Boltzmann-Nordheim equation for bosonic gas as
\begin{equation}\label{BN}
\partial_t f = C_\Phi\int_{\R^d\times \mathbb{S}^{d-1}}|v-v_*|^\gamma b\left(\mbox{cos}\:\theta\right)\left[f'f'_*(1+f+ f_*) - ff_*(1+f'+f'_*)\right]dv_*d\sigma.
\end{equation}
and break it into obvious gain and loss terms
$$\partial_t f = Q^+(f)- fQ^-(f)$$
where 
\begin{eqnarray}
Q^+(f) &=& C_\Phi\int_{\R^d\times \mathbb{S}^{d-1}}|v-v_*|^\gamma b\left(\mbox{cos}\:\theta\right)f'f'_*(1+f+ f_*)\:dv_*d\sigma, \label{Q+}
\\ Q^-(f) &=&C_\Phi\int_{\R^d\times \mathbb{S}^{d-1}}|v-v_*|^\gamma b\left(\mbox{cos}\:\theta\right)f_*(1+f'+f'_*)\:dv_*d\sigma. \label{Q-}
\end{eqnarray}

\bigskip
The goal of this work is to show local in time existence and uniqueness of solutions to the Boltzmann-Nordheim equation for bosonic gas. The main difficulty with the problem is the possible appearance of a Bose-Einstein condensation, i.e. a concentration of mass in the mean velocity, in finite time. In mathematical terms, this can be seen as the appearance of a Dirac function in the solution of the equation $\eqref{BN}$, noticeable by a blow-up in finite time.
\par Such concentration is physically expected, based on various experiments and numerical simulations, as long as the temperature $T$ of the gas is below a critical temperature $T_c (M_0)$ which depends on the mass $M_0$ of the bosonic gas. We refer the interested reader to \cite{EscVel} for an overview of these results.
\bigskip

%%%%%%%%%%%%%%%%%%%%%%%%%%%%%%%%%%%%%%%%%%%%%%%%%%%%%%%%%%%%%%%%%%%%%%%%%%%%%%%%%%%%%%%%%%%%%%%%%%%%%%%%%%%%%%%%%%%%%%%%%%%%%%
%%%%%%%%%%%%%%%%%%%%%%%%%%%%%%%%%%%%%%%%%%%%%%%%%%%%%%%%%%%%%%%%%%%%%%%%%%%%%%%%%%%%%%%%%%%%%%%%%%%%%%%%%%%%%%%%%%%%%%%%%%%%%%
%%%%%%%%%%%%%%%%%%%%%%%%%%%%%%%%%%%%%%%%%%%%%%%%%%%%%%%%%%%%%%%%%%%%%%%%%%%%%%%%%%%%%%%%%%%%%%%%%%%%%%%%%%%%%%%%%%%%%%%%%%%%%%

\subsection{\textit{A priori} expectations for the creation of a Bose-Einstein condensation} \label{subsec:explainations}

In this subsection, we explore some properties of the Boltzmann-Nordheim equation bosonic gas in order to motivate why a concentration phenomenon is expected. We emphasize that everything is stated \textit{a priori} and should not be considered a rigorous proof.

\bigskip
We start by noticing the symmetry property of the Boltzmann-Nordheim operator.

\bigskip
\begin{lemma}\label{lem:integralQ}
Let $f$ be such that $Q(f)$ is well-defined. Then for all $\Psi(v)$ we have
$$\int_{\R^d}Q(f)\Psi\:dv = \frac{C_\Phi}{2}\int_{\R^d\times\R^d\times\mathbb{S}^{d-1}}q(f)(v,v_*)\left[\Psi'_* + \Psi' - \Psi_* - \Psi\right]\:d\sigma dvdv_*,$$
with
$$q(f)(v,v_*) = |v-v_*|^\gamma b\left(\mbox{cos}\:\theta\right)ff_*\left(1+f'+f'_*\right).$$
\end{lemma}
\bigskip

This result is well-known for the Boltzmann equation and is a simple manipulation of the integrand using changes of variables $(v,v_*)\to (v_*,v)$ and $(v,v_*)\to (v',v'_*)$, as well as using the symmetries of the operator $q(f)$. A straightforward consequence of the above is the \textit{a priori} conservation of mass, momentum and energy for a solution $f$ of $\eqref{BN}$ associated to an initial data $f_0$. That is

\begin{equation}\label{conservationlaws}
\int_{\R^d}\left(\begin{array}{c} 1 \\ v \\ \abs{v}^2\end{array}\right)f(v) \:dv = \int_{\R^d}\left(\begin{array}{c} 1 \\ v \\\abs{v}^2\end{array}\right)f_0(v)\:dv = \left(\begin{array}{c} M_0 \\ u \\M_2\end{array}\right).
\end{equation}

\bigskip
The entropy associated to $\eqref{BN}$ is the following functional
$$S(f) = \int_{\R^d} \left[(1+f)\mbox{log}(1+f) - f\mbox{log}(f) \right]\:dv$$
which is, \textit{a priori}, always increasing in time. It has been proved in \cite{Hua} that for given mass $M_0$, momentum $u$ and energy $M_2$, there exists a unique maximizer of $S$ with these prescribed values which is of the form
\begin{equation}\label{FBE}
F_{BE}(v) = m_0 \delta (v-u) + \frac{1}{e^{\frac{\beta}{2}\left(\abs{v-u}^2-\mu\right)}-1},
\end{equation}
where
\begin{itemize}
\item $m_0 \geq 0$.
\item $\beta\in(0,+\infty]$ is the inverse of the equilibrium temperature.
\item $-\infty < \mu \leq 0$ is the chemical potential.
\item $\mu \cdot m_0 =0$.
\end{itemize}
This suggests that for a given initial data $f_0$, the solution of the Boltzmann-Nordheim equation $\eqref{BN}$ should converge, in some sense, to a function of the form $F_{BE}$ with constants that are associated to the physical quantities of $f_0$. Hence, we can expect the appearance of a Dirac function at $u$ if $m_0 \neq 0$.
\par One can show (see \cite{Lu1} or \cite{EscVel2}) that for a given $(M_0, u, M_2)$ we have that if $d=3$, $m_0 =0$ if and only if
\begin{equation}\label{subcritical}
M_0 \leq \frac{\zeta (3/2)}{\left(\zeta (5/2)\right)^{3/5}}\left(\frac{4\pi}{3}\right)^{3/5} M_2^{3/5},
\end{equation}
where $\zeta$ denotes the Riemann Zeta function. Equivalent formulas can be obtained in a similar way for any higher dimension.

\bigskip
According to \cite{ChapCow} Chapter $2$, the kinetic temperature of a bosonic gas is given by
$$T = \frac{m}{3k_B}\frac{M_2}{M_0},$$
where $k_B$ is the physical Boltzmann constant.
This implies, using $\eqref{subcritical}$, that $m_0 =0$ if and only if $T \geq T_c(M_0)$ where

$$T_c(M_0) = \frac{m\zeta (5/2)}{2\pi k_B\zeta (3/2)}\left(\frac{M_0}{\zeta (3/2)}\right)^{2/3}.$$
\par Initial data satisfying $\eqref{subcritical}$ is called subcritical (or critical in case of equality).

\bigskip
From the above discussion, we expect that for low temperatures, $T<T_c(M_0)$, our solution to the Boltzmann-Nerdheim equation will split into a regular part and a highly concentrated part around $u$ as it approaches its equilibrium $F_{BE}$. In \cite{Spo}, Spohn used this idea of a splitting to derive a physical quantitative study of the Bose-Einstein condensation and its interactions with normal fluid, in the case of radially symmetric (isotropic) solutions.
\bigskip

%%%%%%%%%%%%%%%%%%%%%%%%%%%%%%%%%%%%%%%%%%%%%%%%%%%%%%%%%%%%%%%%%%%%%%%%%%%%%%%%%%%%%%%%%%%%%%%%%%%%%%%%%%%%%%%%%%%%%%%%%%%%%
%%%%%%%%%%%%%%%%%%%%%%%%%%%%%%%%%%%%%%%%%%%%%%%%%%%%%%%%%%%%%%%%%%%%%%%%%%%%%%%%%%%%%%%%%%%%%%%%%%%%%%%%%%%%%%%%%%%%%%%%%%%%%
%%%%%%%%%%%%%%%%%%%%%%%%%%%%%%%%%%%%%%%%%%%%%%%%%%%%%%%%%%%%%%%%%%%%%%%%%%%%%%%%%%%%%%%%%%%%%%%%%%%%%%%%%%%%%%%%%%%%%%%%%%%%%

\subsection{Previous studies} \label{subsec:previous}

The issue of existence and uniqueness for the homogeneous bosonic Boltzmann-Nordheim equation has been studied recently in the setting of hard potentials with angular cut-off, especially by X. Lu \cite{Lu1,Lu3,Lu2,Lu5}, and M. Escobedo and J. J. L. Vel\'azquez \cite{EscVel2,EscVel}. It is important to note, however, that these developments have been made under the isotropic setting assumption. We present a short review of what have been done in these works.
\par In his papers \cite{Lu1} and \cite{Lu3}, X. Lu managed to develop a global-in-time Cauchy theory for isotropic initial data with bounded mass and energy, and extended the concept of solutions for isotropic distributions. Under these assumptions, Lu proved existence and uniqueness of radially symmetric solutions that preserve mass and energy. Moreover, he showed that if the initial data has a bounded moment of order $s>2$, then this property will propagate with the equation. Additionally, Lu showed moment production for all isotropic initial data in $L^1_2$.
\par More recently, M. Escobedo and J. J. L. Vel\'azquez used an idea developed by Carleman for the Boltzmann equation \cite{Ca} in order to obtain uniqueness and existence locally in time for radially symmetric solutions in the space $L^\infty(1+\abs{v}^{6+0})$ (see  \cite{EscVel}). As a condensation effect can occur, we can't expect more than local-in-time results in $L^\infty$ spaces in the general setting.

\bigskip
The issue of the creation of a Bose-Einstein condensation has been extensively studied experimentally and numerically in physics \cite{SemTka1}\cite{SemTka2}\cite{JoPoRi}\cite{LaLaPoRi}. Mathematically, a formal derivation of some properties of this condensation, as well as its interactions with the regular part of the solution, has been studied in \cite{Spo} in the isotropic framework.
In the series of papers, \cite{Lu1,Lu3,Lu2}, X. Lu managed to show a condensation phenomenon, under appropriate initial data and in the isotropic setting, as the time goes to infinity. He has shown that at the low temperature case, the isotropic solutions to (\eqref{BN}) converge to the regular part of $F_{BE}$, which has a smaller mass than the initial data. This loss of mass is attributed to the creation of a singular part in the limit, i.e. the desired condensation. It is interesting to notice, as was mentioned in \cite{Lu2}, that this argument does not require the solution to be isotropic and that this created condensation neither proves, or disproves, creation of a Bose-Einstein condensation in finite time.
\par In a recent breakthroughs, \cite{EscVel2,EscVel}, the appearance of Bose-Einstein condensation in finite time has finally been shown. In \cite{EscVel} the authors showed that if the initial data is isotropic in $L^\infty(1+\abs{v}^{6+0})$ with some particular conditions for its distribution of mass near $\abs{v}^2=0$, then the associated isotropic solution exists only in finite time, and its $L^\infty$-norm blows up. This was done by a thorough study of the concentration phenomenon occurring in a bosonic gas. In \cite{EscVel2}, the authors showed that supercritical initial data indeed satisfy the blow-up assumptions in the case of the isotropic setting.
\par More precisely, in \cite{EscVel} the authors showed that there exist  $R_{blowup}$, $\Gamma_{blowup} >0$ such that if the isotropic initial satisfies
$$\int_{\abs{v}\leq R_{blowup}} f_0(\abs{v}^2)\:dv \geq \Gamma_{blowup},$$
measuring concentration around $\abs{v}=0$, then there will be a blow-up in the $L^\infty$-norm in finite time. This should be compared with the very recent proof of Lu \cite{Lu5} showing global existence of solutions in the isotropic setting when
$$\int_{\R^d}\frac{f_0(\abs{v}^2)}{\abs{v}}\:dv \leq \Gamma_{global}$$
for a known $\Gamma_{global} >0$ and $d=3$. The above gives us a measure of lack of concentration near the origin at $t=0$. 
\par At this point we would like to mention that the problem of finite time condensation, intimately connected to the Boltzmann-Nordheim equations for bosons, is far from being fully resolved, and the aforementioned results by Lu, Escobedo and Vel\'azquez are a paramount beginning of the investigation of this problem.
\bigskip

%%%%%%%%%%%%%%%%%%%%%%%%%%%%%%%%%%%%%%%%%%%%%%%%%%%%%%%%%%%%%%%%%%%%%%%%%%%%%%%%%%%%%%%%%%%%%%%%%%%%%%%%%%%%%%%%%%%%%%%%%%%%%%
%%%%%%%%%%%%%%%%%%%%%%%%%%%%%%%%%%%%%%%%%%%%%%%%%%%%%%%%%%%%%%%%%%%%%%%%%%%%%%%%%%%%%%%%%%%%%%%%%%%%%%%%%%%%%%%%%%%%%%%%%%%%%%
%%%%%%%%%%%%%%%%%%%%%%%%%%%%%%%%%%%%%%%%%%%%%%%%%%%%%%%%%%%%%%%%%%%%%%%%%%%%%%%%%%%%%%%%%%%%%%%%%%%%%%%%%%%%%%%%%%%%%%%%%%%%%%

\subsection{Our goals and strategy} \label{subsec:strategy}

The \textit{a priori} conservation of mass, momentum and energy seems to suggest that a natural space to tackle the Cauchy problem is $L^1_2$, the space of positive functions with bounded mass and energy. While this is indeed the right space for the regular homogeneous Boltzmann equation (see \cite{Lu4,MisWen}), the possibility of sharp concentration implies that the $L^\infty$-norm is an important part of the mix as it can measure the condensation blow up.
Additionally, one can see that for short times, when no condensation is created, the boundedness of the $L^\infty$-norm implies a strong connection between the trilinear gain term in the homogeneous Boltzmann-Nordheim equation and the quadratic gain term in the homogeneous Boltzmann equation. Thus, it seems that the right space to look at, when one investigates the Boltzmann-Nordheim equation, is in fact $L^1_2 \cap L^\infty$, or the intersection of $L^1_2$ with some weighted $L^\infty$ space.

\bigskip
The main goal of the present work is to prove that the above intuition is valid by showing a local-in-time existence and uniqueness result for the Boltzmann-Nordheim equation when initial data in $L^1_2 \cap L^\infty(1+\abs{v}^s)$ for a suitable $s$, \emph{without any isotropic assumption}. One of the main novelty of the present paper is the highlighting of the role played by the $L^\infty$-norm not only on the control of possible blow-ups, but also on the gain of regularity of the solutions. This $L^\infty$ investigation is an adaptation of the work of Arkeryd \cite{Ark2} for the classical Boltzmann operator. A core difference between Arkeryd's work and ours lies in the control of the loss term, $Q^{-}$, which can no longer be controlled above zero using the entropy, as well as more complexities arising from dealing with a trilinear term.
\par We tackle the issue of the existence of solutions with an explicit Euler scheme for a family of truncated Boltzmann-Nordheim operators, a natural approach when one wants to propagate boundedness. The sequence of functions we obtain is then shown to converge to a solution of $\eqref{BN}$. The key ingredients we use are a new control on the gain term, $Q^+$, for large and small relative velocities $v-v_*$, estimations of 'gain of regularitiy at infinity' due to having the initial data in $L^\infty \left( 1+\abs{v}^s \right)$, and a refinement and an extension to higher dimensions of a Povzner-type inequality for the evolution of convex and concave functions under a collision.
\par The issue of uniqueness is being dealt by an adaptation of the strategy developed by Mischler and Wennberg in \cite{MisWen} for the homogeneous Boltzmann equation. The main difficulty in this case is the control of terms of the form $\abs{v-v_*}^{2+\gamma}$ that appear when one studies the evolution of the energy of solutions. \\
Besides our local theorems, we also show the appearance of moments of all orders to the solution of $\eqref{BN}$.\\
As can seen from the above discussion, as well as the proofs to follow, we treat the Cauchy theory, and the creation of moments, for the Boltzmann-Nordheim equation as an 'extension' of known results and methods for the Boltzmann equation - though the technicalities involved are far from trivial.
\bigskip

%%%%%%%%%%%%%%%%%%%%%%%%%%%%%%%%%%%%%%%%%%%%%%%%%%%%%%%%%%%%%%%%%%%%%%%%%%%%%%%%%%%%%%%%%%%%%%%%%%%%%%%%%%%%%%%%%%%%%%%%%%%%%%%
%%%%%%%%%%%%%%%%%%%%%%%%%%%%%%%%%%%%%%%%%%%%%%%%%%%%%%%%%%%%%%%%%%%%%%%%%%%%%%%%%%%%%%%%%%%%%%%%%%%%%%%%%%%%%%%%%%%%%%%%%%%%%%%
%%%%%%%%%%%%%%%%%%%%%%%%%%%%%%%%%%%%%%%%%%%%%%%%%%%%%%%%%%%%%%%%%%%%%%%%%%%%%%%%%%%%%%%%%%%%%%%%%%%%%%%%%%%%%%%%%%%%%%%%%%%%%%%

\subsection{Organisation of the article}\label{subsec:organization}

Section $\ref{sec:mainresults}$ is dedicated to the statements and the descriptions of the main results proved in this paper.
\par In Section $\ref{sec:apriori}$ we derive some key properties of the gain and loss operators $Q^+$ and $Q^-$, and show several \textit{a priori} estimates on solutions to $\eqref{BN}$. We end up by proving a gain of regularity at infinity for solutions to the homogeneous Boltzmann-Nordheim equation.
\par As moments of solutions to $\eqref{BN}$ are central in the proof of uniqueness, Section $\ref{sec:moments}$ is dedicated to their investigation. We show an extension of a Povzner-type inequality and use it to prove the instantaneous  appearance of bounded moments of all order. Lastly, we quantify the blow-up near $t=0$ for the moment of order $2+\gamma$.
\par In Section $\ref{sec:uniqueness}$ we show the uniqueness of bounded solutions that preserve mass and energy and then we turn our attention to the proof of local-in-time existence of such bounded, mass and energy preserving solutions in Section $\ref{sec:existence}$. 
\bigskip

\section{Main results} \label{sec:mainresults}

We begin by introducing a few notation that will be used throughout this work.
As we will be considering spaces in the variables $v$ and $t$ separately at times, we will index by $v$ or $t$ the spaces we are working on. The subscript $v$ will always refer to $\R^d$. For instance $L^1_v$ refers to $L^1(\R^d)$ and  $L^\infty_{[0,T],v}$ refers to $L^\infty([0,T]\times \R^d)$. 

We define the following spaces, when $p\in\left\lbrace 1,\infty \right\rbrace$ and $s\in\mathbb{N}$:
$$L^p_{s,v} =\left\{ f \in L^p_v, \quad \norm{(1+\abs{v}^s)f}_{L^p_v} < +\infty\right\}$$
Lastly, we denote the moment of order $\alpha$, where $\alpha\geq0$, of a function $f$ of $t$ and $v$ by
\begin{equation}\label{moments}
 M_\alpha(t) = \int_{\R^d}\abs{v}^\alpha f(t,v)\:dv.
\end{equation}
Note that when $f\geq 0$ the case $\alpha=0$ corresponds to the mass of $f$ while the case $\alpha=2$ corresponds to its energy.

\bigskip
The main result of the work presented here is summed up in the next theorems:

\bigskip
\begin{theorem}\label{theo:uniqeandproperties}
Let $f_0\geq 0$ be in $L^1_{2,v} \cap L_{s,v}^\infty$ when $d\geq 3$ and $d-1<s$.
Then if a non-negative solution to the Boltzmann-Nordheim equation on $[0,T_0)\times \R^d$, $f\in L^{\infty}_{\mbox{\scriptsize{loc}}}\left([0,T_0),L^1_{2,v} \cap L_v^\infty \right)$, that preserves mass and energy exists it must be unique.\\
Moreover, this solution satisfies
\begin{itemize}
\item For any $0\leq s' <\bar{s}$, where $\bar{s}=\min\left\{s\:,\:\frac{d}{1+\gamma}\left(s-d+1+\gamma+\frac{2(1+\gamma)}{d}\right)\right\}$, we have that
$$f \in L^{\infty}_{\mbox{\scriptsize{loc}}}\left([0,T_0),L^1_{2,v} \cap L_{s',v}^\infty \right),$$
\item if $\gamma>0$ then for all $\alpha>0$ and for all $0<T<T_0$, $$M_\alpha(t) \in L^{\infty}_{\mbox{\scriptsize{loc}}}\left([T,T_0)\right).$$ 
\end{itemize}
\end{theorem}
\begin{theorem}\label{theo:existence}
Let $f_0\geq 0$ be in $L^1_{2,v} \cap L_{s,v}^\infty$ when $d\geq 3$ and $d-1<s$.
\\Then, if
\begin{itemize}
\item[(i)] $\gamma=0$ and $s>d$, or
\item[(ii)] $0<\gamma\leq 1$ and $s>d+2+\gamma$,
\end{itemize}
 there exists $T_0 >0$, $d$, $s$, $C_\Phi$, $b_\infty$, $l_b$, $\gamma$,  $\norm{f_0}_{L^1_{2,v}}$ and $\norm{f_0}_{L^\infty_{s,v}}$,  such that there exists a non-negative solution to the Boltzmann-Nordheim equation on $[0,T_0)\times \R^d$, $f\in L^{\infty}_{\mbox{\scriptsize{loc}}}\left([0,T_0),L^1_{2,v} \cap L_v^\infty \right)$, that preserves mass and energy. Moreover
 $$T_0 = +\infty \quad \mbox{or}\quad \limsup\limits_{T \to T_0^-}\norm{f}_{L^\infty_{[0,T] \times \R^d}} = + \infty.$$
\end{theorem}
\bigskip

\begin{remark}
We mention a few remarks in regards to the above theorem:
\begin{enumerate}
\item It is easy to show that $\bar{s}=s$ if $d=3$ or $s\geq d$.
\item The difference between conditions $(i)$ and $(ii)$ in Theorem \ref{theo:existence} arises from the explicit Euler scheme we employ. To show the existence, we start by solving an appropriate truncated equation. However, any 'regularity at infinity' that may be gained due to the term $\abs{v-v_*}^\gamma$ for $\gamma>0$ is lost due to this truncation. Thus, an additional assumption on the weighted $L^\infty$ norm is required. Note that the case $d=3,\gamma=1$ gives the same condition as that of \cite{EscVel}.
\item Of great importance is the observation that the above theorems identifies an appropriate norm in the general non-isotropic setting, the $L^\infty$ norm, under which a study of the appearance of a blow up in finite time is possible - giving rise to a proof of local existence and uniqueness. We would like to mention that this blow up may not be the Bose-Einstein condensation itself and additional assumptions, such as the ones presented in \cite{EscVel}\cite{EscVel2}, may be needed to fully characterise the condensation phenomena.
\item Much like the classical Boltzmann equation, higher order moments are created immediately, but unlike it, $M_\alpha(t)$ are only locally bounded. We also emphasize here that this creation of moments only requires $f_0$ to be in $L^1_{2,v} \cap L^\infty_v$ as we shall see in Section $\ref{sec:moments}$.
\item Lastly, let us mention that our proofs still hold in $d=2$, but only in the special case $\gamma=0$. This is due to the use of the Carleman representation for $Q^{+}$.
\end{enumerate}
\end{remark}
\bigskip

%At this point we would like to mention that Theorem $\ref{theo:cauchypb}$ implies an asymptotic Bose-Einstein condensation phenomenon for a globally defined solution with subcritical initial data when the time goes to infinity. This follows from the work of Lu \cite{Lu2}.
%\bigskip

%\bigskip

\section{A priori estimate: control of the regularity by the $L^\infty_v$-norm}\label{sec:apriori}

This section is dedicated to proving an \textit{a priori} estimate in the $L^\infty_v$ space for solutions to $\eqref{BN}$, locally in time. As was mentioned before, we cannot expect more than this as we know from \cite{EscVel} that even for radially symmetric solutions there are solutions with a blow-up in finite time.
\par Many results in this section are an appropriate adaptation of the work of Arkeryd \cite{Ark2}. Nonetheless, we include full proofs to our main claims for the sake of completion.
\par The main theorem of the section, presented shortly, identifies the importance of the $L^\infty_v$ requirement as an indicator for blow-ups. Indeed, as we shall see, the boundedness of the solution, along with appropriate initial conditions, immediately implies higher regularity at infinity.
\bigskip
\begin{theorem}\label{theo:apriori}
Let $f_0 \geq 0$ in $L^1_{2,v} \cap L^\infty_{s,v}$ when $d\geq 3$ and $d-1<s$.
\\ Let $f$ be a non-negative solution of $\eqref{BN}$ in $L^{\infty}_{\mbox{\scriptsize{loc}}}\left([0,T_0),L^1_{2,v}\cap L^\infty_v\right)$, with initial value $f_0$, satisfying the conservation of mass and energy.
\\Define
\begin{equation}\label{sbar}
\bar{s}=\min\left\{s\:;\:\frac{d}{1+\gamma}\left(s-d+1+\gamma+\frac{2(1+\gamma)}{d}\right)\right\}.
\end{equation}
\\Then for all $0\leq T <T_0$ and all $s' <\bar{s}$ there exists an explicit $C_T>0$ such that following holds
$$\forall t \in [0,T], \quad \norm{f(t,\cdot)}_{L^\infty_{s',v}} \leq C_T.$$
The constant $C_T$ depends only on $T$, $d$, the collision kernel, $\norm{f}_{L^\infty_{[0,T],v}}$, $\norm{f_0}_{L^1_{2,v}\cap L^\infty_{s,v}}$, $s$ and $s'$.
\end{theorem}
\bigskip

The entire section is devoted to the proof of this result. 
\par We start by stating a technical lemma that will be used throughout the entire section, whose proof we leave to the Appendix.
\bigskip
\begin{lemma}\label{lem:integration lemma part II}
Let $s_1,s_2\geq 0$ be such that $s_2-s_1<d$ and let $f\in L^1_{s_1,v}\cap L^\infty_{s_2,v}$.
\\Then, for any $0\leq \alpha <d$ we have
$$\int_{\R^d}f(v_*)\abs{v-v_*}^{-\alpha}\:dv_* \leq C_{d,\alpha}\left(\norm{f}_{L^1_{s_1,v}}+\norm{f}_{L^\infty_{s_2,v}}\right)\left(1+\abs{v}\right)^{-b}$$
where 
$$b=\min\left(\alpha,s_1+\frac{\alpha(s_2-s_1)}{d}\right)$$
and $C_{d,\alpha}>0$ depends only on $d$ and $\alpha$.
\end{lemma}
\bigskip

%%%%%%%%%%%%%%%%%%%%%%%%%%%%%%%%%%%%%%%%%%%%%%%%%%%%%%%%%%%%%%%%%%%%%%%%%%%%%%%%%%%%%%%%%%%%%%%%%%%%%%%%%%%%%%%%%%%%%%%%%%%%%%%%%%%%%%%%%%%%%%%%%%%%%%%%%%%%%%%%%%%
%%%%%%%%%%%%%%%%%%%%%%%%%%%%%%%%%%%%%%%%%%%%%%%%%%%%%%%%%%%%%%%%%%%%%%%%%%%%%%%%%%%%%%%%%%%%%%%%%%%%%%%%%%%%%%%%%%%%%%%%%%%%%%%%%%%%%%%%%%%%%%%%%%%%%%%%%%%%%%%%%%%
%%%%%%%%%%%%%%%%%%%%%%%%%%%%%%%%%%%%%%%%%%%%%%%%%%%%%%%%%%%%%%%%%%%%%%%%%%%%%%%%%%%%%%%%%%%%%%%%%%%%%%%%%%%%%%%%%%%%%%%%%%%%%%%%%%%%%%%%%%%%%%%%%%%%%%%%%%%%%%%%%%%

\subsection{Key properties of the gain and loss operators}\label{subsec:keypropertiesQ+Q-}

In this subsection we gather and prove some useful properties of the gain and loss operators $Q^-$ and $Q^+$ that will be used in what is to follow.
\par First, we have the following control on the loss operator.

\bigskip
\begin{lemma}\label{lem:Q^- control}
Let $f \geq 0$ be in $L^1_{2,v}$. Then 
\begin{equation}\label{eq: Q^- control}
\forall v \in \R^d, \quad Q^{-}(f)(v) \geq C_{\Phi}l_b\left(1+\abs{v}^\gamma\right)\norm{f}_{L^1_v}-C_{\Phi}C_\gamma l_b\norm{f}_{L^1_{2,v}},
\end{equation}
where
\begin{equation}\label{Cgamma}
C_\gamma=\sup\limits_{x\geq 0}\frac{1+x^\gamma}{1+x^2}.
\end{equation}
\end{lemma}
\bigskip

\begin{proof}[Proof of Lemma $\ref{lem:Q^- control}$]
Using the fact that for any $x,y>0$ and $0\leq \gamma\leq 1$ we have
$$\abs{x}^\gamma - \abs{y}^\gamma \leq \abs{x-y}^{\gamma}$$
we find that for any $v\in \R^d$
\begin{equation}\nonumber
\begin{split}
Q^{-}(f)(v) &\geq C_{\Phi}\int_{\R^d\times \mathbb{S}^{d-1}}\left[\left(1+\abs{v}^\gamma\right)-\left(1+\abs{v_*}^{\gamma}\right)\right]b\left(\cos \theta\right)f_* \:dv_* d\sigma
\\&\geq C_{\Phi}l_b\left(1+\abs{v}^\gamma\right)\norm{f}_{L^1_v}-C_{\Phi} C_\gamma l_b\norm{f}_{L^1_{2,v}}.
\end{split}
\end{equation}
\end{proof}
\bigskip

\begin{remark}
Had we had a uniform in time control over the entropy, $\int_{\R^d}f\log dv$, we would have been able to find a strictly positive lower bound for the loss operator, much like in the case of the Boltzmann equation. However, for the Boltzmann-Nordheim equation the appropriate decreasing entropy is given by
$$\int_{\R^d}\left((1+f)\log(1+f)-f\log f\right)dv,$$
which is not as helpful.
\end{remark}

\bigskip
An essential tool in the investigation of the $L^\infty$ properties of solutions to the Boltzmann equation is the so-called Carleman representation. This representation of the gain operator has been introduced by Carleman in \cite{Ca2} and consisted of changing the integration variables in the expression for it from $dv_\ast d\sigma$ to $dv^\prime dv_\ast^\prime$ on $\R^d$ and appropriate hyperplanes. As shown in \cite{GPV}, the representation reads as:
\begin{equation}\label{eq: Carleman rep}
\int_{\R^d\times \mathbb{S}^{d-1}}  B(v-v_*,\sigma)f'f'_*\:dv_*d\sigma = 2^{d-1}\int_{\R^d}\frac{dv'}{\abs{v-v'}}\int_{E_{vv'}}\frac{B\left(2v-v'_*-v',\frac{v'_*-v'}{\abs{v'_*-v'}}\right)}{\abs{v'_*-v'}^{d-2}}f'f'_*\:dE(v'_*)
\end{equation} 
where $E_{vv'}$ is the hyperplane that passes through $v$ and is orthogonal to $v-v'$, and $dE(v'_*)$ is the Lebesgue measure of it. The above suggests that controlling the integration on $E_{vv'}$ may be the key to a good control of the gain operator. This was indeed the successful strategy undertaken by Arkeryd (see \cite{Ark2}), and is the strategy we will follow as well.

\bigskip
\begin{lemma}\label{lem:controlQ+infty}
Let $f\geq 0$ be in $L^1_{2,v} \cap L^\infty_v$. If $\gamma\in[0,d-2]$, then
\begin{equation*}
\norm{Q^+(f)}_{L^\infty_v}\leq C_+ \left(1 + 2\norm{f}_{L^\infty_v}\right)\sup\limits_{v,v'\in \R^{d}}\left[\int_{E_{vv'}}f'_*\:dE(v'_*)\right] \int_{\R^d} \frac{f'}{\abs{v-v'}^{d-1-\gamma}}\:dv',
\end{equation*}
where $C_+ = 2^{d-1}C_\Phi b_\infty$.
\end{lemma}
\bigskip

\begin{remark}
Note that the requirement of having $\gamma$ in $[0,d-2]$ prevents our method from working in $d=2$ unless $\gamma=0$. 
\end{remark}
\bigskip

\begin{proof}[Proof of Lemma $\ref{lem:controlQ+infty}$]
As was the noted before the statement of the lemma, the key ingredient to the proof is the Carleman representation $\eqref{eq: Carleman rep}$.

We start by noticing that our collision kernel satisfies
\begin{equation}\nonumber
\frac{B\left(2v-v_*'-v',\frac{v_*'-v'}{\abs{v_*'-v'}}\right)}{\abs{v_*'-v'}^{d-2}} \leq C_\Phi b_\infty\frac{\abs{v-v_*}^\gamma}{\abs{v'-v_*'}^{d-2}}
=\frac{C_\Phi b_\infty}{\abs{v-v_*}^{d-2-\gamma}}=\frac{C_\Phi b_\infty}{\abs{2v-v'-v_*'}^{d-2-\gamma}}
\end{equation}
Since we are on $E_{vv'}$ we have that $\abs{2v-v'-v_*'}=\sqrt{\abs{v-v'}^2+\abs{v-v_*'}^2}$ and we conclude that
\begin{equation}\nonumber
\frac{B\left(2v-v_*'-v',\frac{v_*'-v'}{\abs{v_*'-v'}}\right)}{\abs{v_*'-v'}^{d-2}} \leq \frac{C_\Phi b_\infty}{\abs{v-v'}^{d-2-\gamma}}
\end{equation}
as $\gamma \leq d-2$. Thus, bounding $f$ and $f_*$ by their $L^\infty_v$-norms and then combining the above with the new representation $\eqref{eq: Carleman rep}$ we find that
$$Q^{+}\left(f\right)(v) \leq C_{\Phi}b_\infty\left(1+2\norm{f}_{L^\infty_v}\right)\norm{\int_{E_{vv'}}f'_*\:dE(v_*')}_{L^{\infty}}\norm{\int_{\R^d}f'\abs{v-v'}^{-d+1+\gamma}\:dv'}_{L^\infty_v}$$
which is the desired result.
\end{proof}
\bigskip

The following two lemmas give us control over the integration of the gain operator over Carleman's hyperplanes, which is essential to the proof of the main theorem for this section.

\bigskip
\begin{lemma}\label{lem:integration on hyperplanes of Q^+}
Let $f\geq 0$ be in $L^1_v\cap L^\infty_v$. For any given $v\in \R^d$ we have that almost everywhere in the direction of $v-v'$
\begin{equation}\label{eq:integration on hyperplanes of Q^{+}}
\int_{E_{vv'}}Q^{+}(f)(v_*')\:dE(v_*') \leq C_{+E}\left(1+2\norm{f}_{L^\infty_v}\right)\norm{f}_{L^1_v}\sup\limits_{v_1 \in \R^d}\left[\int_{\R^d}\frac{f(v)}{\abs{v-v_1}^{1-\gamma}}\:dv\right],
\end{equation}
where $C_{+E}>0$ depends only on $d$, $C_\Phi$ and $b_\infty$.
\end{lemma}
\bigskip

\begin{proof}[Proof of Lemma \ref{lem:integration on hyperplanes of Q^+}]
Denote by $\varphi_n(v)=\left(\frac{n}{2\pi}\right)^{\frac{1}{2}}e^{-\frac{nD\left(v,E_{vv'}\right)^2}{2}}$, where $D\left(v,A\right)$ is the distance of $v$ from the set $A$.
\par Using the standard change of variables $(v,v_*,\sigma)\to (v',v'_*,\sigma)$ we find that
\begin{equation}\nonumber
\int_{\R^d}\varphi_n(v)Q^{+}(f)(v)\:dv \leq C_{\Phi}b_\infty\left(1+2\norm{f}_{L^\infty_v}\right)\int_{\R^d\times \R^d \times \mathbb{S}^{d-1}}\varphi_n'\abs{v-v_*}^\gamma ff_*\:dvdv_* d\sigma .
\end{equation}
We have that
$$\int_{\mathbb{S}^{d-1}}\varphi_n\left(v'\right)\:d\sigma = \frac{2^{d-1}}{\abs{v-v_*}^{d-1}}\int_{\mathbb{S}_{vv_*}} \varphi_n(x)\:ds(x)$$
where $ds$ is the uniform measure on $\mathbb{S}_{vv_*}$ which is the sphere of radius $\abs{v-v_*}/2$ centred at $(v+v_*)/2$. It is easy to show (see Lemma \ref{lem:concentration of delta on a sphere}) that for any $a\in\R^d$ and $r>0$ we have
$$\sup_{n}\frac{1}{r^{d-2}}\int_{\mathbb{S}_r(a)}\varphi_n(x)ds(x) \leq \abs{\mathbb{S}^{d-2}}.$$
and as such
\begin{equation}\label{eq: varphi_n vs Q^+}
\int_{\R^d}\varphi_n(v)Q^{+}(f)(v)\:dv \leq C_{\Phi}\abs{\mathbb{S}^{d-2}}b_\infty\left(1+2\norm{f}_{L^\infty_v}\right)\int_{\R^d\times \R^d } \frac{ff_*}{\abs{v-v_*}^{1-\gamma}}\:dvdv_* d\sigma .
\end{equation}

Using the fact that $\varphi_n$ converge to the delta function of $E_{vv'}$ we conclude that

\begin{equation*}
\begin{split}
\int_{E_{vv'}}Q^{+}(f)(v'_*)\:dE(v'_*) &= \lim_{n\rightarrow +\infty} \int_{\R^d}\varphi_n(v)Q^{+}(f)(v)\:dv
\\ &\leq \abs{\mathbb{S}^{d-2}} C_{\Phi}b_\infty\left(1+2\norm{f}_{L^\infty_v}\right)\int_{\R^d\times \R^d} \frac{ff_*}{\abs{v-v_*}^{1-\gamma}}\:dvdv_*
\\&\leq \abs{\mathbb{S}^{d-2}} C_{\Phi}b_\infty\left(1+2\norm{f}_{L^\infty_v}\right)\norm{f}_{L^1_v} \sup\limits_{v_1 \in \R^d}\int_{\R^d} \frac{f}{\abs{v-v_1}^{1-\gamma}}\:dv,
\end{split}
\end{equation*}
which is the desired result.
\end{proof}

\bigskip
\begin{lemma}\label{lem:weight propagation preperation I}
Let $a\in \R^d$ and define 
\begin{equation}\nonumber
\psi_a(v)=\begin{cases} 0 & \abs{v}<\abs{a} \\
1 & \abs{v}\geq \abs{a}.
\end{cases}
\end{equation}
If $f\in L^{1}_{s,v}\cap L^\infty_v$ when $s\geq \frac{d}{d-1}$ then for almost every hyperplane $E_{vv'}$
\begin{equation*}
\int_{E_{vv'}}\psi_a(v_*')Q^{+}(f)(v_*')\:dE(v_*') \leq C_\Phi C_{d,\gamma}b_\infty
\left(\norm{f}_{L^1_{s,v}}+\norm{f}_{L^\infty_v}\right)^3 \left(1+\abs{a}\right)^{-s+\gamma-1}
\end{equation*}
where $C_{d,\gamma}>0$ is a constant depending only on $d$ and $\gamma$.
\end{lemma}
\bigskip

\begin{proof}[Proof of Lemma $\ref{lem:weight propagation preperation I}$]
The proof follows the same lines of the proof of Lemma \ref{lem:integration on hyperplanes of Q^+}. We define $\varphi_n$ to be the approximation of the delta function on the appropriate hyperplane. Then
\begin{equation*}
\begin{split}
\int_{\R^d}\varphi_n(v)\psi(v)Q^{+}(f)(v)\:dv \leq & C_\Phi b_\infty\left(1+2\norm{f}_{L^\infty_v}\right)
\\&\times\int_{\R^d\times \R^d \times \mathbb{S}^{d-1}}\varphi_n(v')\psi(v')f(v)f(v_*)\abs{v-v_*}^{\gamma}\:dvdv_{*}d\sigma
\end{split}
\end{equation*}
Since $\abs{v}\leq \abs{a}/2$ and $\abs{v_*}\leq \abs{a}/2$ implies $\psi(v')=0$ (as $\abs{v'}\leq \abs{v}+\abs{v_*}$) we conclude that the above is bounded by
\begin{equation*}
\begin{split}
&C_\Phi b_\infty\left(1+2\norm{f}_{L^\infty_v}\right)\int_{\left\{\abs{v}\geq \frac{\abs{a}}{2}\vee  \abs{v_*}\geq \frac{\abs{a}}{2}\right\}  \times \mathbb{S}^{d-1}}\varphi_n(v')f(v)f(v_*)\abs{v-v_*}^{\gamma}\:dvdv_{*}d\sigma 
\\&\leq C_\Phi \abs{\mathbb{S}^{d-2}} b_\infty\left(1+2\norm{f}_{L^\infty_v}\right)\int_{\left\{\abs{v}\geq \frac{\abs{a}}{2}\vee  \abs{v_*}\geq \frac{\abs{a}}{2}\right\}  }f(v)f(v_*)\abs{v-v_*}^{\gamma-1}\:dvdv_{*}d\sigma
\\&\leq C_\Phi \abs{\mathbb{S}^{d-2}} b_\infty\left(1+2\norm{f}_{L^\infty_v}\right)\left(\int_{\abs{v}>\frac{\abs{a}}{2}}f(v)\:dv\right) \left(\sup\limits_{\abs{v}>\frac{\abs{a}}{2}} \int_{\R^d}f(v_*)\abs{v-v_*}^{\gamma-1}\:dv_*\right)
\\&\leq C_\Phi C_{d,\gamma} b_\infty\left(1+2\norm{f}_{L^\infty_v}\right) \frac{\norm{f}_{L^1_{s,v}}}{\left(1+\abs{a}\right)^{s}}\frac{\norm{f}_{L^1_{s,v}}+\norm{f}_{L^\infty_v}}{\left(1+\abs{v}\right)^b}
\end{split}
\end{equation*}
for 
$$b=\min\left\{1-\gamma\:;\:s\left(1-\frac{1-\gamma}{d}\right)\right\}$$
where we have used Lemma \ref{lem:integration lemma part II}. The result follows from taking $n$ to infinity as $s\geq d/(d-1)$ implies 
\begin{equation*}
\max_{0\leq \gamma \leq 1}\frac{1-\gamma}{1-\frac{1-\gamma}{d}}\leq s.
\end{equation*}
\end{proof}
\bigskip

%%%%%%%%%%%%%%%%%%%%%%%%%%%%%%%%%%%%%%%%%%%%%%%%%%%%%%%%%%%%%%%%%%%%%%%%%%%%%%%%%%%%%%%%%%%%%%%%%%%%%%%%%%%%%%%%%%%%%%%%%%%%%%%%%%%%%%%%%%%%%%%%%%%%%%%%%%%%%%%%%%%
%%%%%%%%%%%%%%%%%%%%%%%%%%%%%%%%%%%%%%%%%%%%%%%%%%%%%%%%%%%%%%%%%%%%%%%%%%%%%%%%%%%%%%%%%%%%%%%%%%%%%%%%%%%%%%%%%%%%%%%%%%%%%%%%%%%%%%%%%%%%%%%%%%%%%%%%%%%%%%%%%%%
%%%%%%%%%%%%%%%%%%%%%%%%%%%%%%%%%%%%%%%%%%%%%%%%%%%%%%%%%%%%%%%%%%%%%%%%%%%%%%%%%%%%%%%%%%%%%%%%%%%%%%%%%%%%%%%%%%%%%%%%%%%%%%%%%%%%%%%%%%%%%%%%%%%%%%%%%%%%%%%%%%%

\subsection{A priori properties of solutions of $\eqref{BN}$}\label{subsec:apriorisolutions}

The first step towards the proof of Theorem $\ref{theo:apriori}$ is to obtain some \textit{a priori} estimates on $f$ when $f$ is a bounded solution of the Boltzmann-Nordheim equation.
\par We first derive an estimation of the growth of the moments of $f$ when $f_0$ has moments higher than $2$.

\bigskip
\begin{prop}\label{prop:prop of moments}
Assume that $f$ is a solution to the Boltzmann-Nordheim equation with initial conditions $f_0\in L^1_{s,v}$ for $s>2$. Then, for any $T<T_0$ we have that
\begin{equation*}
\norm{f(t,\cdot)}_{L^1_{s,v}}\leq e^{2C_\Phi C_s b_\infty\left(1+2\sup\limits_{t\in(0,T]}\norm{f}_{L^\infty_v}\right) \norm{f_0}_{L^1_{2,v}}t}\norm{f_0}_{L^1_{s,v}}.
\end{equation*}
\end{prop}
\bigskip

\begin{proof}[Proof of Proposition $\ref{prop:prop of moments}$]
We have that
\begin{equation}\nonumber
\begin{split}
&\frac{d}{dt}\int_{\R^d}\left(1+\abs{v}^s\right) f(v,t)\:dv 
\\&\quad= \frac{C_\Phi}{2}\int_{\R^d\times \R^d \times \mathbb{S}^{d-1}}q(f)(v,v_*)\left(\abs{v'}^s+\abs{v'_*}^s-\abs{v}^s-\abs{v_*}^s\right)\:dvdv_* d\sigma 
\\&\quad\leq C_\Phi C_s b_\infty\left(1+2\sup\limits_{t\in(0,T]}\norm{f}_{L^\infty_v}\right)\int_{\R^d\times \R^d}\abs{v}^{s-1}\abs{v_*}\left(\abs{v}^\gamma+\abs{v_*}^\gamma\right)
f(v)f(v_*)\:dvdv_*
\\&\quad\leq C_\Phi C_s b_\infty\left(1+2\sup\limits_{t\in(0,T]}\norm{f}_{L^\infty_v}\right)\int_{\R^d\times \R^d}\left(\abs{v}^{s}\abs{v_*}+\abs{v}^{s-1}\abs{v_*}^2\right)f(v)f(v_*)\:dvdv_*
\\&\quad\leq 2C_\Phi C_s b_\infty\left(1+2\sup\limits_{t\in(0,T]}\norm{f}_{L^\infty_v}\right) \norm{f_0}_{L^1_{2,v}}\int_{\R^d}\left(1+\abs{v}^s\right)f(v)\:dv
\end{split}
\end{equation}
where we have used the known inequality
\begin{equation}\nonumber
\abs{v'}^s+\abs{v'_*}^s-\abs{v}^s-\abs{v_*}^s \leq C_s \abs{v}^{s-1}\abs{v_*} 
\end{equation}
for $s>2$ and some $C_s$ depending only on $s$, the fact that $\gamma \leq 1$ and the inequality
\begin{equation}\nonumber
\abs{v}^\alpha \leq 1+\abs{v}^{\alpha+1}
\end{equation}
for any $\alpha \geq 0$. The result follows.
\end{proof}
\bigskip

The next stage in our investigation is to show that under the conditions of Theorem \ref{theo:apriori} one can actually bound the integral of $f$ over $E_{vv'}$ \emph{uniformly in time}, which will play an important role in the proof of the mentioned theorem, and more.

\bigskip
\begin{prop}\label{prop:integration on hyperplane of f}
Let $f$ be a solution to the Boltzmann-Nordheim equation that satisfies the conditions of Theorem $\ref{theo:apriori}$ and let $0\leq T < T_0$. Then there exists $C_E >0$ and $C_0 \in \R^*$ such that for any given $v\in \R^d$ we have that almost everywhere in the direction of $v-v'$ and for all $t\in [0,T]$
\begin{equation*}
\begin{split}
&\int_{E_{vv'}}f'_*(t)\:dE(v_*') \leq  C_E \: e^{-C_0t}\norm{f_0}_{L^\infty_{s,v}}
\\ &\quad+ C_E \frac{1-e^{-C_0T}}{C_0}\norm{f_0}_{L^1_v}\left(1+2\sup\limits_{\tau\in[0,T]}\norm{f(\tau,\cdot)}_{L^\infty_v}\right)\left(\norm{f_0}_{L^1_v}+\sup\limits_{\tau \in [0,T]}\norm{f(\tau,\cdot)}_{L^\infty_v}\right)
\end{split}
\end{equation*}
where the constant $C_E$ only depends on $d$, $s$ and the collision kernel, and $C_0$ depends also on $f_0$ and satisfies
$$Q^-(f)(v) \geq C_0.$$
\end{prop}
\bigskip

\begin{remark}
From Lemma $\ref{lem:Q^- control}$ we know that we can choose $C_0=C_{\Phi}l_b (C_\gamma\norm{f_0}_{L^1_{2,v}}-\norm{f_0}_{L^1_v})$ but the theorem can be stated more generally, as presented. Notice that the choice above can satisfy $C_0<0$, which will imply an exponential growth in the bound.
\end{remark}
\bigskip

\begin{proof}[Proof of Proposition \ref{prop:integration on hyperplane of f}]
Define $\varphi_n$ as in Lemma \ref{lem:integration on hyperplanes of Q^+}. Since $f$ is a solution to the Boltzmann-Nordheim equation and that $Q^-(f)(v) \geq C_0$ we find that
\begin{equation}\label{diffeqintegralEvv'f}
\frac{d}{dt}\int_{\R^d}\varphi_n(v)f(t,v)\:dv \leq -C_0\int_{\R^d}f(t,v)\varphi_n(v)\:dv +\int_{\R^d}\varphi_n(v)Q^{+}(f)(v)\:dv.
\end{equation}

\bigskip
Using $\eqref{eq: varphi_n vs Q^+}$ we conclude that
\begin{equation}\label{importantineqEvv*}
\begin{split}
&\int_{\R^d}\varphi_n(v)Q^{+}(f(t,\cdot))(v)\:dv 
\\&\leq \abs{\mathbb{S}^{d-2}} C_{\Phi}b_\infty\left(1+2\sup\limits_{\tau\in [0,T]}\norm{f(\tau,\cdot)}_{L^\infty_v}\right) \int_{\R^d\times \R^d }\abs{v-v_*}^{\gamma-1} f(t,v)f(t,v_*)\:dvdv_*
\\&\leq \abs{\mathbb{S}^{d-2}} C_{\Phi}b_\infty\left(1+2\sup\limits_{\tau\in [0,T]}\norm{f(\tau,\cdot)}_{L^\infty_v}\right)\norm{f_0}_{L^1_v} \sup\limits_{\tau \in [0,T],\:v_1 \in \R^d}\int_{\R^d}\frac{f(\tau,v_*)}{\abs{v_1-v_*}^{1-\gamma}} \:dv_*,
\end{split}
\end{equation}
where we used that $f$ is mass preserving. We notice that for $\gamma  > 1-d$ and $v\in\R^d$
$$\int_{\R^d}\frac{f(v_*)}{\abs{v-v_*}^{1-\gamma}} \:dv_*  \leq \norm{f}_{L^\infty_v}\int_{\abs{x}<1}\frac{dx}{\abs{x}^{1-\gamma}}+\norm{f}_{L^1_v}$$
implying
\begin{equation*}
\begin{split}
&\int_{\R^d}\varphi_n(v)Q^{+}(f(t,\cdot))(v)\:dv 
\\&\leq C_{d,\gamma} C_{\Phi}b_\infty\left(1+2\sup\limits_{\tau\in [0,T]}\norm{f(\tau,\cdot)}_{L^\infty_v}\right)\norm{f_0}_{L^1_v} \left(\norm{f_0}_{L^1_v}+\sup\limits_{\tau \in [0,T]}\norm{f(\tau,\cdot)}_{L^\infty_v}\right),
\end{split}
\end{equation*}
for an appropriate $C_{d,\gamma}$.
\bigskip
The resulting differential inequality from $\eqref{diffeqintegralEvv'f}$ is
$$\frac{d}{dt}\int_{\R^d}\varphi_n(v)f(t,v)\:dv \leq - C_0 \int_{\R^d}\varphi_n(v)f(t,v)\:dv +C_T$$
with an appropriate $C_T$, which implies by a Gr\"onwall lemma that
\begin{equation}\label{diffeqfinal}
\int_{\R^d}\varphi_n(v)f(t,v)\:dv \leq \left(\int_{\R^d}\varphi_n(v)f_0\:dv\right)e^{-C_0 t}+ \frac{C_T}{C_0}\left[1-e^{-C_0 t}\right].
\end{equation}
Since
$$\lim_{n\rightarrow\infty}\int_{\R^d}\varphi_n(v)f_0 \:dv = \int_{E_{vv'}}f_0(v'_*)\:dE(v'_*) \leq \norm{f_0}_{L^\infty_{s,v}}\int_{E_{vv'}}\frac{dE(v'_*)}{1+\abs{v_*'}^s} = C_{d,s}\norm{f_0}_{L^\infty_{s,v}}$$
as $s>d-1$, we take the limit as $n$ goes to infinity in $\eqref{diffeqfinal}$ which yields
\begin{equation}\label{eq:exact bound on Q+} 
\int_{E_{vv'}}f(t,v'_*)\:dE(v_*') \leq C_{d,s}\norm{f_0}_{L^\infty_{s,v}}e^{-C_0t}+\frac{C_T}{C_0}\left[1-e^{-C_0 t}\right]
\end{equation}
which is the desired result.
\end{proof}
\bigskip
%\begin{remark}\label{rem:propagationhyperplane}
%As will be shown, the proofs of local existence, as well as uniqueness, only require the uniform in time and in $v,v'$ integrability of the solution on the family of hyperplanes $E_{vv'}$. A careful look at the proof of Proposition $\ref{prop:integration on hyperplane of f}$ shows that for this to be true it is sufficient for $f_0$ to be bounded and to satisfy this uniform integrability. The only change would be to replace $\norm{f_0}_{L^\infty_{s,v}}$ by 
%$$\sup\limits_{v,v'\in\R^{d}} \int_{E_{vv'}}f_0(v'_*)\:dE(v'_*).$$
%\end{remark}
%\bigskip

Lastly, before proving Theorem \ref{theo:apriori}, we give one more \textit{a priori} type of estimates on the family of hyperplanes $E_{vv'}$.

\bigskip
\begin{prop}\label{prop:weight propagation preperation II}
Let $f$ be a solution to the Boltzmann-Nordheim equation that satisfies the conditions of Theorem $\ref{theo:apriori}$ and let $0\leq T < T_0$. For any $a\in \R^d$ define 
\begin{equation}\nonumber
\psi_a(v)=\begin{cases} 0 & \abs{v}<\abs{a} \\
1 & \abs{v}\geq \abs{a}.
\end{cases}
\end{equation}
Then for almost every hyperplane $E_{vv'}$ and $t\in[0,T]$
\begin{equation*}
\int_{E_{vv'}}\psi_v(v_*')f(t,v_*')\:dE(v_*') \leq \int_{E_{vv'}}\psi_v(v_*')f_0(v_*')\:dE(v'_*) 
+ C_{T,\alpha}\left(1+\abs{v}\right)^{-\alpha}
\end{equation*}
with $\alpha=3$ if $s\leq d+2$ and $\alpha = s'+1$ for any $s'<s-d$ if $s > d+2$. The constant $C_{T,\alpha} >0$ depends only on $T$, $d$, the collision kernel, $\sup\limits_{t\in(0,T]}\norm{f(t,\cdot)}_{L^\infty_v}$, $\norm{f_0}_{L^1_{2,v}\cap L^\infty_{s,v}}$, $s$ and $s'$.
\end{prop}
\bigskip

\begin{proof}[Proof of Proposition $\ref{prop:weight propagation preperation II}$]
We start by noticing that if $s-s'>d$ then
\begin{equation}\nonumber
\int_{\R^d}\left(1+\abs{v}^{s'}\right)f_0(v)\:dv \leq C_{s,s'}\norm{f_0}_{L^\infty_{s,v}}\int_{\R^d}\frac{dv}{1+\abs{v}^{s-s'}} = C_{s,s'}\norm{f_0}_{L^\infty_{s,v}}
\end{equation}
Thus, if $s>d+2$ we can conclude that $f_0\in L^{1}_{s',v}$ for any $2<s' < s-d$, improving the initial assumption on $f_0$.

\bigskip
We continue as in Lemma $\ref{lem:weight propagation preperation I}$ and define $\varphi_n$ to be the approximation of the delta function on $E_{vv'}$. Denoting by 
\begin{equation}\nonumber
I_n(t)=\int_{\R^d}\varphi_n(v_*)\psi_v(v_*)f(t,v_*)\:dv_*
\end{equation}
we find that, using Lemma $\ref{lem:weight propagation preperation I}$, Proposition $\ref{prop:prop of moments}$ and denoting by $C_T$ the appropriate constant from the mentioned lemma and proposition,
\begin{equation}\nonumber
\begin{split}
\frac{d}{dt}I_n(t)\leq & \left(-C_{\Phi}l_b\left(1+\abs{v}^\gamma\right)\norm{f_0}_{L^1_v}+C_{\Phi}C_\gamma l_b\norm{f_0}_{L^1_{2,v}}\right) I_n(t) 
\\ & + C_T\left\{\begin{array}{l} \displaystyle{\left(1+\abs{v}\right)^{-s'+\gamma-1}\quad \mbox{if} \quad s>d+2 \quad\mbox{and}\quad s'<s-d}\vspace{2mm}\\\vspace{2mm} \displaystyle{\left(1+\abs{v}\right)^{\gamma-3} \quad\mbox{if} \quad s\leq d+2.} \end{array} \right.
\end{split}
\end{equation}
The above differential inequality implies (see Lemma $\ref{lem:decay differential inequality}$ in Appendix) that for any 
$$\abs{v}\geq \left(\frac{2C_\gamma \norm{f_0}_{L^1_{2,v}}}{\norm{f_0}_{L^1_v}}\right)^{\frac{1}{\gamma}}-1,$$
the following holds:
\begin{equation}\nonumber
I_n(t)\leq I_n(0) + C_T \left\{\begin{array}{l} \displaystyle{\left(1+\abs{v}\right)^{-s'-1}\quad \mbox{if} \quad s>d+2 \quad\mbox{and}\quad s'<s-d}\vspace{2mm}\\\vspace{2mm} \displaystyle{\left(1+\abs{v}\right)^{-3} \quad\mbox{if} \quad s\leq d+2.} \end{array} \right.
\end{equation}
Taking $n$ to infinity along with Proposition \ref{prop:integration on hyperplane of f} yields the desired result as when 
$$\abs{v}<\left(\frac{2C_\gamma \norm{f_0}_{L^1_{2,v}}}{\norm{f_0}_{L^1_v}}\right)^{\frac{1}{\gamma}}-1$$
the following holds:
\begin{equation}\nonumber
\int_{E_{vv'}}\psi_v(v_*')f(t,v_*')\:dE(v_*') \leq \left(\frac{2C_\gamma \norm{f_0}_{L^1_{2,v}}}{ \norm{f_0}_{L^1_{v}}}\right)^{\frac{\beta}{\gamma}}\frac{1}{\left(1+\abs{v}\right)^\beta}\int_{E_{vv'}}f(t,v_*')\:dE(v_*'). 
\end{equation}
\end{proof}
\bigskip

\begin{remark}\label{rem:improved formula}
We notice that since $f_0\in L^\infty_{s,v}$
\begin{equation}\nonumber
\begin{split}
\int_{E_{vv'}}\psi_v(v_*')f_0(v_*')\:dE(v_*') &\leq \norm{f_0}_{L^\infty_{s,v}}\int_{E_{vv'}}\frac{\psi_v(v_*')}{1+\abs{v'_*}^s}\:dE(v_*') 
\\ &\leq C_{s,s^{''},d}\left(1+\abs{v}\right)^{-(s-s^{''})}
\end{split}
\end{equation}
for any $d-1<s^{''}<s$. This implies that Proposition $\ref{prop:weight propagation preperation II}$ can be rewritten as
\begin{equation}\nonumber
\begin{gathered}
\int_{E_{vv'}}\psi_v(v_*')f(t,v_*')\:dE(v_*') \leq C_T\begin{cases} \left(1+\abs{v}\right)^{-\left(s-d+1-\epsilon\right)} & s>d+2
\\ \left(1+\abs{v}\right)^{-\min\left(3,s-d+1-\epsilon\right)} & s\leq d+2
\end{cases}
\end{gathered}
\end{equation}
where we have picked $s^{''}=d-1+\epsilon$ and $s'=s-d-\epsilon$ for an arbitrary $\epsilon$ small enough. As $s-d+1-\epsilon\leq 3-\epsilon$ when $s\leq d+2$ for any $\epsilon$ we conclude that
\begin{equation}\label{eq:weight propagation preperation II improved}
\begin{gathered}
\int_{E_{vv'}}\psi_v(v_*')f(t,v_*')\:dE(v_*') \leq C_T\left(1+\abs{v}\right)^{-\left(s-d+1-\epsilon\right)} 
\end{gathered}
\end{equation}
\end{remark}
\bigskip

%%%%%%%%%%%%%%%%%%%%%%%%%%%%%%%%%%%%%%%%%%%%%%%%%%%%%%%%%%%%%%%%%%%%%%%%%%%%%%%%%%%%%%%%%%%%%%%%%%%%%%%%%%%%%%%%%%%%%%%%%%%%%%%%%%%%%%%%%%%%%%%%%%%%%%%%%%%%%%%%%%%
%%%%%%%%%%%%%%%%%%%%%%%%%%%%%%%%%%%%%%%%%%%%%%%%%%%%%%%%%%%%%%%%%%%%%%%%%%%%%%%%%%%%%%%%%%%%%%%%%%%%%%%%%%%%%%%%%%%%%%%%%%%%%%%%%%%%%%%%%%%%%%%%%%%%%%%%%%%%%%%%%%%
%%%%%%%%%%%%%%%%%%%%%%%%%%%%%%%%%%%%%%%%%%%%%%%%%%%%%%%%%%%%%%%%%%%%%%%%%%%%%%%%%%%%%%%%%%%%%%%%%%%%%%%%%%%%%%%%%%%%%%%%%%%%%%%%%%%%%%%%%%%%%%%%%%%%%%%%%%%%%%%%%%%

\subsection{Gain of regularity at infinity}\label{subsec:proofapriori}
This subsection is entirely devoted to the proof of Theorem $\ref{theo:apriori}$. 
\begin{proof}[Proof of Theorem $\ref{theo:apriori}$]
We start by noticing that the function 
\begin{equation}\nonumber
f_{l,v}(v_*)=\left(1-\psi_{\frac{v}{\sqrt{2}}}\left(v_*\right)\right)f(v_*),
\end{equation} 
where $\psi_a$ was defined in Lemma $\ref{lem:weight propagation preperation I}$, satisfies 
\begin{equation}\nonumber
f_{l,v}\left(v'\right)f_{l,v}\left(v_*'\right)=0.
\end{equation}
Indeed, as 
\begin{equation}\nonumber
\abs{v'}^2+\abs{v_*'}^2=\abs{v}^2+\abs{v_*}^2 \geq \abs{v}^2  
\end{equation}
we find that $\abs{v'}\geq \abs{v}/\sqrt{2}$ or $\abs{v_*'}\geq \abs{v}/\sqrt{2}$. This implies that 
\begin{equation}\nonumber
Q^+(f_{l,v})(v)=0
\end{equation}
and thus, by setting $f_{h,v}=f-f_{l,v}$ we have that
\begin{equation}\nonumber
\begin{split}
Q^{+}(f)(v)&\leq C_{\Phi}b_\infty\left(1+2\sup\limits_{(0,T]}\norm{f}_{L^\infty_{v}}\right)\int_{\R^d\times\mathbb{S}^{d-1}}\abs{v-v_*}^\gamma f\left(v'\right)f\left(v'_*\right)\:dv_* d\sigma
\\ & =C_{\Phi}b_\infty\left(1+2\sup\limits_{(0,T]}\norm{f}_{L^\infty_{v}}\right)\left(Q^+_{B,\gamma}\left(f_{h,v},f_{h,v}\right)+2Q^+_{B,\gamma}\left(f_{l,v},f_{h,v}\right)\right)
\\ & \leq 3C_{\Phi}b_\infty\left(1+2\sup\limits_{(0,T]}\norm{f}_{L^\infty_{v}}\right)Q^+_{B,\gamma}\left(f,f_{h,v}\right)
\end{split}
\end{equation}
where 
\begin{equation}\nonumber
Q^{+}_{B,\gamma}(f,g)=\int_{\R^d\times\mathbb{S}^{d-1}}\abs{v-v_*}^\gamma f\left(v'\right)g\left(v'_*\right)\:dv_* d\sigma.
\end{equation}
and we have used the fact that $Q^{+}(f,g)$ is symmetric under exchanging $f$ and $g$.

\bigskip
Using Carleman's representation $\eqref{eq: Carleman rep}$ for $Q^{+}_{B,\gamma}$ along with Lemma $\ref{lem:integration lemma part II}$ and Remark $\ref{rem:improved formula}$ we find that
\begin{equation}\label{eq:Q^+_B estimation}
\begin{split}
Q^{+}_{B,\gamma}(f,f_{h,v})(v) &\leq \int_{\R^d}\frac{f(v')\:dv'}{\abs{v-v'}^{d-1-\gamma}}\int_{E_{vv'}}f_{h,v}\left(v_*'\right)dE(v_*') 
\\& \leq C_T\frac{\norm{f_0}_{L^1_{2,v}}+\sup\limits_{(0,T]}\norm{f}_{L^\infty_v}}{\norm{f_0}_{L^1_{2,v}}}\left(1+\abs{v}\right)^{-\delta},
\end{split}
\end{equation} 

where $C_T >0$ is a constructive constant depending only on $d$, $s$, $f_0$, the collision kernel and $T$ and where
\begin{equation}\nonumber
\delta=\min\left(s-\gamma-\epsilon_1,\xi\right)
\end{equation}
with $\xi=s-d+1-\epsilon_1+\frac{2(1+\gamma)}{d}$ and $\epsilon_1$ to be chosen later.

\bigskip
As $f$ solves the Boltzmann-Nordheim equation, we find that it must satisfy the following inequality:
\begin{equation}\label{eq:diff inequality for f}
\begin{gathered}
\partial_t f \leq 3C_\Phi b_\infty C_T\left(1+2\sup\limits_{(0,T]}\norm{f}_{L^\infty_{v}}\right)\frac{\norm{f_0}_{L^1_{2,v}}+\sup\limits_{(0,T]}\norm{f}_{L^\infty_v}}{\norm{f_0}_{L^1_{2,v}}}\left(1+\abs{v}\right)^{-\delta}
\\ -\left(C_{\Phi}l_b\left(1+\abs{v}^\gamma\right)\norm{f}_{L^1_v}-C_{\Phi}C_\gamma l_b\norm{f}_{L^1_{2,v}}\right)f
\end{gathered}
\end{equation} 
where we have used Lemma $\ref{lem:Q^- control}$ and $\eqref{eq:Q^+_B estimation}$.
\par Solving $\eqref{eq:diff inequality for f}$ (see Lemma $\ref{lem:decay differential inequality}$ in the Appendix) with abusive notation for $C_T$, we find that for any $\tilde{\delta}\leq \delta$
\begin{equation}\label{eq:bootstrap weight equation}
\begin{split}
\norm{f(t,\cdot)}_{L^{\infty}_{\gamma+\tilde{\delta},v}} \leq & \norm{f_0}_{L^{\infty}_{\gamma+\tilde{\delta},v}} +C_T
% \frac{\left(\norm{f_0}_{L^1_{2,v}}+\sup\limits_{(0,T]}\norm{f}_{L^\infty_v}\right)^2}{l_b^2\norm{f_0}_{L^1_{2,v}}} 
%\\&+\left(C_T \frac{\left(\norm{f_0}_{L^1_{2,v}}+\sup\limits_{(0,T]}\norm{f}_{L^\infty_v}\right)^2}{l_b^2\norm{f_0}_{L^1_{2,v}}}\right)^{1+\frac{\delta}{\gamma}}\sup\limits_{\tau\in(0,T]}\norm{f(\tau,\cdot)}_{L^\infty_v}.
\end{split}
\end{equation}

\bigskip
Let $s' < \bar{s}$ be given and denote by $\epsilon=\bar{s}-s'$. We shall show that the $L^\infty_{s',v}-$norm of $f$ can be bounded uniformly in time by a constant depending only on the initial data, dimension and collision kernel.
\par If $\delta \geq s-\gamma-\epsilon$ the result follows from $\eqref{eq:bootstrap weight equation}$. Else, the same equation implies that $f(t,\cdot)\in L^{\infty}_{\xi+\gamma}$ uniformly in $t\in(0,T]$. Repeating the same arguments leading to $\eqref{eq:bootstrap weight equation}$ but using Lemma $\ref{lem:integration lemma part II}$ with an $L^\infty$ weight of $s_2=\xi+\gamma$ instead of $s_2=0$ yields an improved version of $\eqref{eq:bootstrap weight equation}$ where $\sup\limits_{\tau\in(0,T]}\norm{f(\tau,\cdot)}_{L^{\infty}_{v}}$ is replaced with $\sup\limits_{\tau\in(0,T]}\norm{f(\tau,\cdot)}_{L^{\infty}_{\xi+\gamma,v}}$, and $\delta$ is replaced with
\begin{equation}\nonumber
\delta_1=\min\left(s-\gamma-\epsilon_1,\xi+\frac{d-1-\gamma}{d}(\xi+\gamma)\right).
\end{equation}
We continue by induction. Defining
\begin{equation}\nonumber
\delta_n=\min\left(s-\gamma-\epsilon_1,\xi+(\xi+\gamma)\sum_{j=1}^n\left(\frac{d-1-\gamma}{d}\right)^{j}\right).
\end{equation}
we assume that for any $\tilde{\delta}\leq \delta_n$
\begin{equation}\label{eq:bootstrap weight equation I}
\begin{gathered}
\norm{f(t,\cdot)}_{L^{\infty}_{\gamma+\tilde{\delta},v}} \leq C_T
\end{gathered}
\end{equation}
where $C_{T}$ depends only on $C_\Phi$, $b_\infty$, $l_b$, $T$ ,$\sup\limits_{t\in(0,T]}\norm{f(t,\cdot)}_{L^\infty_v}$, $\norm{f_0}_{L^\infty_{s,v}}$ ,$\norm{f_0}_{L^1_{2,v}}$, $\gamma$, $s$, $d$ and $\epsilon_1$. 
\par If $\delta_n=s-\gamma-\epsilon$ the proof is complete, else we can reiterate the proof to find that (\ref{eq:bootstrap weight equation I}) is valid for $\delta\leq \delta_{n+1}$.

\bigskip
Since
\begin{equation}\nonumber
\begin{split}
\xi+\left(\xi+\gamma\right)\sum_{j=1}^{\infty}\left(\frac{d-1-\gamma}{d}\right)^j &=\frac{d}{1+\gamma}(\xi+\gamma)-\gamma
\\&=\frac{d}{1+\gamma}\left(s-d+1+\gamma+\frac{2(1+\gamma)}{d}-\epsilon_1\right)-\gamma
\end{split}
\end{equation}
we conclude that we can bootstrap our $L^\infty$ weight up to 
$$\frac{d}{1+\gamma}\left(s-d+1+\gamma+\frac{2(1+\gamma)}{d}\right)-\epsilon \geq \bar{s}-\epsilon$$
in finitely many steps with an appropriate choice of $\epsilon_1$. This completes the proof.
\end{proof}
\bigskip

\section{Creation of moments of all order}\label{sec:moments}

This section is dedicated to proving the immediate creation of moments of all order to the Boltzmann-Nordheim equation, as long as they are in $L^{\infty}_{\mbox{\scriptsize{loc}}}\left([0,T_0),L^1_{2,v}\cap L^\infty_v\right)$. This will play an important role in the proof of the uniqueness of the solutions, as when one deals with the difference of two solutions one cannot assume any fixed sign and usual control on the gain and loss terms fails. Higher moments of the solutions will be required to give a satisfactory result, due to the kinetic kernel $\abs{v-v_*}^\gamma$.
\par The instantaneous generation of moments of all order is a well known and important result for the Boltzmann equation (see \cite{MisWen}). As for finite times, assuming no blow ups in the solution, the Boltzmann-Nordheim's gain and loss terms control, and are controlled, by the appropriate gain and loss terms of the Boltzmann equation, one can expect that a similar result would be valid for the bosonic gas evolution.
\par We would like to emphasize at this point that our proofs follow the arguments used in \cite{MisWen} with the key difference of a newly extended Povzner-type inequality, from which the rest follows. The reader familiar with the work of Mischler and Wennberg may just skim through the statements and skip to the next section of the paper.

\bigskip
The study of the generation of higher moments will be done in three steps:
\par The first subsection is dedicated to a refinement of a Povzner-type inequality \cite{MisWen,Pov} which captures the geometry of the collisions in the Boltzmann kernel. Such inequalities control the evolution of convex and concave functions under the effect of a collision, which is what we are looking for in the case of moments.
\par In the second subsection we will prove the appearance of moments for solutions to Boltzmann-Nordheim equation for bosons in $L^1_{2,v}\cap L^\infty_v$.
\par We conclude by quantifying the rate of explosion of the $(2+\gamma)^{th}$ moment as the time goes to $0$. This estimate will be of great importance in the proof of the uniqueness.
\bigskip

%%%%%%%%%%%%%%%%%%%%%%%%%%%%%%%%%%%%%%%%%%%%%%%%%%%%%%%%%%%%%%%%%%%%%%%%%%%%%%%%%%%%%%%%%%%%%%%%%%%%%%%%%%%%%%%%%%%%%%%%%%%%%%%%%%%%%%%%%%%%%%%%%%%%%%%%%%%%%%%%%%%
%%%%%%%%%%%%%%%%%%%%%%%%%%%%%%%%%%%%%%%%%%%%%%%%%%%%%%%%%%%%%%%%%%%%%%%%%%%%%%%%%%%%%%%%%%%%%%%%%%%%%%%%%%%%%%%%%%%%%%%%%%%%%%%%%%%%%%%%%%%%%%%%%%%%%%%%%%%%%%%%%%%
%%%%%%%%%%%%%%%%%%%%%%%%%%%%%%%%%%%%%%%%%%%%%%%%%%%%%%%%%%%%%%%%%%%%%%%%%%%%%%%%%%%%%%%%%%%%%%%%%%%%%%%%%%%%%%%%%%%%%%%%%%%%%%%%%%%%%%%%%%%%%%%%%%%%%%%%%%%%%%%%%%%

\subsection{An extended version of a Povzner-type inequality}\label{subsec:povzner}

The main result of this subsection is the following Povzner-type inequality for the Boltzmann-Nordheim equation.

\begin{lemma}\label{lem:povzner}
Let $b(\theta)$ be a positive bounded function and let $F\in L^\infty(\R^d\times\R^d\times\mathbb{S}^{d-1})$ be such that $F \geq a >0$.
\\Given a function $\psi$ we define.
$$K_\psi(v,v_*) = \int_{\mathbb{S}^{d-1}}F(v,v_*,\sigma)b(\theta)\left(\psi(\abs{v'_*}^2)+\psi(\abs{v'}^2)-\psi(\abs{v_*}^2)-\psi(\abs{v}^2)\right)\:d\sigma.$$
Then, denoting by $\chi(v,v_*) = 1 - \mathbf{1}_{\{\abs{v}/2<\abs{v_*}<2\abs{v}\}}$, we find the following decomposition for $K$:
$$K_\psi(v,v_*) = G_\psi(v,v_*)-H_\psi(v,v_*),$$
where $G$ and $K$ satisfy the following properties:
\begin{enumerate}
\item[(i)] If $\psi(x) = x^{1+\alpha}$ with $\alpha >0$ then
$$\abs{G(v,v_*)} \leq C_G \alpha \left(\abs{v}\abs{v_*}\right)^{1+\alpha}$$
and
$$H(v,v_*) \geq C_H\alpha\left(\abs{v}^{2+2\alpha}+\abs{v_*}^{2+2\alpha}\right)\chi(v,v_*).$$
\item[(ii)]  If $\psi(x) = x^{1+\alpha}$ with $-1<\alpha<0$ then
$$\abs{G(v,v_*)} \leq C_G \abs{\alpha} \left(\abs{v}\abs{v_*}\right)^{1+\alpha}$$
and
$$-H(v,v_*) \geq C_H \abs{\alpha}\left(\abs{v}^{2+2\alpha}+\abs{v_*}^{2+2\alpha}\right)\chi(v,v_*).$$

\item[(iii)] If $\psi$ is a positive convex function that can be written as $\psi(x) = x\phi(x)$ for a concave function $\phi$ that increases to infinity and satisfies that for any $\eps>0$ and $\alpha\in (0,1)$
$$\left(\phi(x)-\phi(\alpha x)\right)x^\eps \underset{x\to\infty}{\longrightarrow} \infty$$
 Then, for any $\eps >0$,
$$\abs{G(v,v_*)} \leq C_G \abs{v}\abs{v_*}\left(1+\phi\left(\abs{v}^2\right)\right)\left(1+ \phi\left(\abs{v_*}^2\right)\right)$$
and
$$H(v,v_*) \geq C_H \left(\abs{v}^{2-\eps}+\abs{v_*}^{2-\eps}\right)\chi(v,v_*).$$
In addition, there is a constant $C>0$ such that $\phi'(x) \leq C/(1+x)$ implies $G(v,v_*) \leq C_G\abs{v}\abs{v_*}$.
\end{enumerate}
The constants $C_G$ and $C_H$ are positive and depend only on $\alpha$, $\psi$, $\eps$, $b$, $a$ and $\norm{F}_{L^\infty_{v,v_*,\sigma}}$.
\end{lemma}
\bigskip

\begin{remark}\label{remark:povzner}
The operator $H_\psi$ in the above lemma can be chosen to be monotonous in $\psi$ in the following sense: if $\psi=\psi_1-\psi_2 \geq 0$ is convex then $H_{\psi_1}-H_{\psi_2} \geq 0$. This property will prove itself extremely useful later on in the paper.
\end{remark}
\bigskip

\begin{proof}[Proof of Lemma $\ref{lem:povzner}$]
The proof follows similar arguments to the one presented in \cite{MisWen} where $F=1$ and $d=3$. Much like in the work of Mischler and Wennberg, we decompose $\abs{v'}^2$ and $\abs{v'_*}^2$ to a convex combination of $\abs{v}^2$ and $\abs{v_*}^2$ and a remainder term, and use convexity/concavity properties of $\psi$ and $\phi$.

\bigskip
We start by recalling the definition of $v',v'_*$ and $\cos \theta$:
$$\left\{ \begin{array}{rl}&\displaystyle{v' = \frac{v+v_*}{2} +  \frac{|v-v_*|}{2}\sigma} \vspace{2mm} \\ \vspace{2mm} &\displaystyle{v' _*= \frac{v+v_*}{2}  -  \frac{|v-v_*|}{2}\sigma} \end{array}\right., \: \quad \quad \mbox{cos}\:\theta = \left\langle \frac{v-v_*}{\abs{v-v_*}},\sigma\right\rangle .$$
One can see that
\begin{equation*}
\begin{split}
\abs{v'}^2 =& \abs{v}^2\left[\frac{1}{2} + \frac{1}{2}\left\langle \frac{v-v_*}{\abs{v-v_*}} , \sigma\right\rangle\right] +\abs{v_*}^2\left[\frac{1}{2} - \frac{1}{2}\left\langle \frac{v-v_*}{\abs{v-v_*}} , \sigma\right\rangle\right]
\\& +\left[\frac{\abs{v-v_*}}{2}\langle v+v_*,\sigma\rangle - \frac{1}{2}\left\langle \frac{v-v_*}{\abs{v-v_*}},\sigma\right\rangle \left(\abs{v}^2-\abs{v_*}^2\right)\right].
\end{split}
\end{equation*}
$$=\beta(\sigma)\abs{v}^2+(1-\beta(\sigma))\abs{v_*}^2+Z(\sigma)=Y(\sigma)+Z(\sigma),$$
where 
\begin{eqnarray}
\beta(\sigma) &=& \frac{1}{2} + \frac{1}{2}\left\langle \frac{v-v_*}{\abs{v-v_*}} , \sigma\right\rangle \quad \in [0,1], \label{betasigma}
\\ Y(\sigma) &=& \beta(\sigma) \abs{v}^2+ (1-\beta(\sigma))\abs{v_*}^2, \label{Ysigma}
\\ Z(\sigma) &=& \frac{\abs{v-v_*}}{2}\langle v+v_*,\sigma\rangle - \frac{1}{2}\left\langle \frac{v-v_*}{\abs{v-v_*}},\sigma\right\rangle \left(\abs{v}^2-\abs{v_*}^2\right) \label{Zsigma} 
\end{eqnarray}
Similarly, one has
$$\abs{v'_*}^2 = Y(-\sigma) + Z(-\sigma).$$
As $Z$ is an odd function in $\sigma$, we can split the integration over $\mathbb{S}^{d-1}$ to the domains where $Z$ is positive and negative. By changing $\sigma$ to $-\sigma$ and adding and subtracting the term $\psi(Y(\sigma))+\psi(Y(-\sigma))$, as well as using the fact that $\beta(-\sigma)+\beta(\sigma)=1$ we conclude that 
\begin{equation}\label{decompositionKpsi}
\begin{split}
K_\psi =&\quad \int_{\sigma:\:Z(\sigma)\geq 0}\left[b(\theta)F(\sigma)+b(\pi-\theta)F(-\sigma)\right] \left[\psi(Y+Z) -\psi(Y)\right]\:d\sigma
\\&+\int_{\sigma:\:Z(\sigma)\geq 0}\left[b(\theta)F(\sigma)+b(\pi-\theta)F(-\sigma)\right] \left[\psi(Y(-\sigma)-Z(\sigma)) -\psi(Y(-\sigma))\right]\:d\sigma
\\&-\int_{\mathbb{S}^{d-1}}\left[b(\theta)F(\sigma)+b(\pi-\theta)F(-\sigma)\right] \left[\beta\psi(\abs{v}^2)+(1-\beta)\psi(\abs{v_*}^2)-\psi(Y)\right]\:d\sigma,
\end{split}
\end{equation}
We define 
\begin{equation}\label{tildeHpsi}
\tilde{H}_\psi=\int_{\mathbb{S}^{d-1}}\left[b(\theta)F(\sigma)+b(\pi-\theta)F(-\sigma)\right]  \left[\beta\psi(\abs{v}^2)+(1-\beta)\psi(\abs{v_*}^2)-\psi(Y)\right]\:d\sigma
\end{equation}
and notice that due to the definition of $Y(\sigma)$ and the convexity or concavity of $\psi$, $\tilde{H}_{\psi}$ always has a definite sign. As such
\begin{equation}
\tilde{H}_{\psi} \geq a\int_{\mathbb{S}^{d-1}}\left[b(\theta)+b(\pi-\theta)\right] \left[\beta\psi(\abs{v}^2)+(1-\beta)\psi(\abs{v_*}^2)-\psi(Y)\right]\:d\sigma,
\end{equation}
when $\psi$ is convex and 
\begin{equation}
-\tilde{H}_{\psi} \geq \norm{F}_{L^\infty_{v,v_*,\sigma}}\int_{\mathbb{S}^{d-1}}\left[b(\theta)+b(\pi-\theta)\right] \left[\beta\psi(\abs{v}^2)+(1-\beta)\psi(\abs{v_*}^2)-\psi(Y)\right]\:d\sigma,
\end{equation}
when $\psi$ is concave. At this point the proof of $(i)$ and $(ii)$ for $H_{\psi}$ follows the arguments presented in \cite{MisWen}.

\bigskip
We now turn our attention to the remaining two integrals in \eqref{decompositionKpsi}. Due to the positivity of $b$ and $F$, and the monotonicity of $\psi$ both integrals will be dealt similarly and we restrict our attention to the first. One sees that
\begin{equation}\label{ineqGpsi}
\begin{split}
&\abs{\int_{\sigma:\:Z(\sigma)\geq 0}\left[b(\theta)F(\sigma)+b(\pi-\theta)F(-\sigma)\right] \left[\psi(Y+Z) -\psi(Y)\right]\:d\sigma} 
\\&\quad\quad\quad\leq 2b_\infty \norm{F}_{L^\infty_{v,v_*,\sigma}}\int_{\sigma:\:Z(\sigma)\geq 0}\left[\psi(Y+Z) -\psi(Y)\right]\:d\sigma.
\end{split}
\end{equation}
The rest of the proof will rely on a careful investigation of the integrand. To do so, we start by noticing that
\begin{equation}\label{controlYZsigma}
\begin{split}
2\sqrt{\beta}\sqrt{1-\beta}\abs{v}\abs{v_*} & \leq Y(\sigma) \leq \abs{v}^2+\abs{v_*}^2 \\
\abs{Z(\sigma)} & \leq 4 \abs{v}\abs{v_*}.
\end{split}
\end{equation}
\textbf{Case $(i)$:} In that case we have for all $\sigma$ on $\mathbb{S}^{d-1}$ such that $Z(\sigma)\geq 0$
$$\psi(Y+Z) -\psi(Y)=(1+\alpha)Z\left(C(\sigma)\right)^\alpha\leq (1+\alpha)Z\left(Y+Z\right)^\alpha .$$
As $Y+Z = \abs{v'}^2 \leq \abs{v}^2+\abs{v_*}^2$ we find that
\begin{eqnarray*}
\psi(Y+Z) -\psi(Y) &\leq& 4 \abs{v}\abs{v_*}\left(\abs{v}^2+\abs{v_*}^2\right)^\alpha
\\&\leq& (1+\alpha)\left\{\begin{array}{l}\displaystyle{4^{1+2\alpha}\left(\abs{v}\abs{v_*}\right)^{1+\alpha}\quad\mbox{if}\quad \frac{\abs{v}}{2} \leq \abs{v_*}\leq 2\abs{v}}\vspace{2mm}\\\vspace{2mm}\displaystyle{\frac{4}{\eps^{\frac{1}{\alpha}}}\left(\abs{v}^2+\abs{v_*}^2\right)^{1+\alpha} + 4\eps\left(\abs{v}\abs{v_*}\right)^{1+\alpha} \quad\mbox{otherwise.}}\end{array}\right.
\end{eqnarray*}
where we have used H\"older inequality in the second term. As $\epsilon$ is arbitrary, one can choose it such that
$$\frac{8(1+\alpha) b_\infty \norm{F}_{L^\infty_{v ,v_*. \sigma}}}{\eps^{\frac{1}{\alpha}}} \leq \frac{\tilde{C}_H}{2},$$
where $\tilde{C}_H$ is the constant associated to $\tilde{H}_\psi$. Defining
$$H_\psi=\tilde{H}_\psi+\frac{8(1+\alpha ) b_\infty \norm{F}_{L^\infty_{v ,v_*. \sigma}}}{\eps^{\frac{1}{\alpha}}}\left(\abs{v}^2+\abs{v_*}^2\right)^{1+\alpha} \chi(v,v_*)$$
and $G_\psi$ to be what remains, the proof is completed for this case.

\bigskip
\textbf{Case $(ii)$:} In that case we have for all $\sigma$ on $\mathbb{S}^{d-1}$ such that $Z(\sigma)\geq 0$
$$\psi(Y+Z) -\psi(Y)=(1+\alpha)Z\left(C(\sigma)\right)^\alpha\leq (1+\alpha)ZY^\alpha .$$
Using $\eqref{controlYZsigma}$ we find that
\begin{equation}\label{case2}
\psi(Y+Z) -\psi(Y) \leq C \left(\abs{v}\abs{v_*}\right)^{1+\alpha}\frac{1}{\left[\beta(\sigma)(1-\beta(\sigma))\right]^{\alpha/2}}.
\end{equation}
Since
$$\int_{\mathbb{S}^{d-1}}\frac{d\sigma}{\left[\beta(\sigma)(1-\beta(\sigma))\right]^{\alpha/2}}=C_d \int_{0}^{\frac{\pi}{2}} \frac{\sin^{d-2}\theta\cos^{d-2}\theta}{\left(\cos \theta \sin \theta \right)^{\alpha}}d\theta<\infty$$
This yields the desired result with the choice $H_\psi=\tilde{H}_\psi$ and $G_\psi$ the remaining terms.

\bigskip
\textbf{Case $(iii)$:} This case will be slightly more complicated and we will deal with the first two integrations in \eqref{decompositionKpsi} separately.\\
We start with the second integral. As $Z\geq 0$ in the domain of integration and $Y\geq 0$ always, we find that
$$\abs{\int_{\sigma:\:Z(\sigma)\geq 0}\left[b(\theta)F(\sigma)+b(\pi-\theta)F(-\sigma)\right] \left[\psi(Y(-\sigma)-Z(\sigma)) -\psi(Y(-\sigma))\right]\:d\sigma}$$
$$\leq 2b_\infty \norm{F}_{L^\infty_{v,v_*,\sigma}}\int_{\sigma:\:Z(\sigma)\geq 0} Z(\sigma)\psi'(Y(-\sigma))\:d\sigma$$
where we have used the fact that $\psi$ is convex. As $\psi'(x)=\phi(x)+x\phi'(x)$ and 
$$\phi(x)-\phi(0) \geq x \phi'(x)$$
when $x>0$, due to the concavity of $\phi$, we have that
\begin{equation}\label{eq:pov Y-Z}
\begin{gathered}
\abs{\int_{\sigma:\:Z(\sigma)\geq 0}\left[b(\theta)F(\sigma)+b(\pi-\theta)F(-\sigma)\right] \left[\psi(Y(-\sigma)-Z(\sigma)) -\psi(Y(-\sigma))\right]\:d\sigma}\\
\leq 16b_\infty \norm{F}_{L^\infty_{v,v_*,\sigma}} \abs{v}\abs{v_*}\left(\int_{\sigma:\:Z(\sigma)\geq 0}\phi(Y(-\sigma))\:d\sigma\right)
\end{gathered}
\end{equation}
where we have used $\eqref{controlYZsigma}$ and the positivity of $\phi$.\\
to deal with the first integral in $\eqref{decompositionKpsi}$ we notice that for $Z\geq 0$
$$\abs{\psi(Y(\sigma)+Z(\sigma))-\psi(Y(\sigma))}\leq Y(\sigma)Z(\sigma)\phi'(Y(\sigma))+Z(\sigma)\phi(Y(\sigma)+Z(\sigma))$$
where we have used to concavity of $\phi$. Like before we can conclude that
\begin{equation}\label{eq:pov Y+Z}
\begin{gathered}
\abs{\int_{\sigma:\:Z(\sigma)\geq 0}\left[b(\theta)F(\sigma)+b(\pi-\theta)F(-\sigma)\right] \left[\psi(Y(\sigma)+Z(\sigma)) -\psi(Y(\sigma))\right]\:d\sigma}\\
\leq 8b_\infty \norm{F}_{L^\infty_{v,v_*,\sigma}} \abs{v}\abs{v_*}\left(\int_{\sigma:\:Z(\sigma)\geq 0}\left(\phi(Y(\sigma))+\phi(Y(\sigma)+Z(\sigma))\right)\:d\sigma\right).
\end{gathered}
\end{equation}
Adding \eqref{eq:pov Y-Z} and \eqref{eq:pov Y+Z} and using the positivity and concavity of $\phi$ we find that by choosing $H_\psi=\tilde{H}_{\psi}$ we have that
\begin{equation*}
\begin{split}
&G_\psi(v,v_*)  \\
&\quad \leq 16b_\infty \norm{F}_{L^\infty_{v,v_*,\sigma}}  \abs{v}\abs{v_*}\left(\phi\left(\int_{\mathbb{S}^{d-1}}Y(\sigma)d\sigma\right)+\phi\left(\int_{\mathbb{S}^{d-1}}\left(Y(\sigma)+Z(\sigma)\right)d\sigma\right)\right)\\
&\quad =32 b_\infty \norm{F}_{L^\infty_{v,v_*,\sigma}}   \abs{v}\abs{v_*}\phi\left(\frac{\abs{v}^2+\abs{v_*}^2}{2}\right)\\
&\quad \leq 32 b_\infty \norm{F}_{L^\infty_{v,v_*,\sigma}}   \abs{v}\abs{v_*}\max\left(\phi(\abs{v}^2),\phi(\abs{v_*}^2)\right)
\end{split}
\end{equation*}
which completes the estimation for $G_\psi$ in the general case. 
\par Property $(iii)$ for $H_{\psi}$ is proved along the same lines of the proof of Mischler and Wennberg in \cite{MisWen}, as well as the second part of the case.

\end{proof}

\bigskip

%%%%%%%%%%%%%%%%%%%%%%%%%%%%%%%%%%%%%%%%%%%%%%%%%%%%%%%%%%%%%%%%%%%%%%%%%%%%%%%%%%%%%%%%%%%%%%%%%%%%%%%%%%%%%%%%%%%%%%%%%%%%%%%%%%%%%%%%%%%%%%%%%%%%%%%%%%%%%%%%%%%
%%%%%%%%%%%%%%%%%%%%%%%%%%%%%%%%%%%%%%%%%%%%%%%%%%%%%%%%%%%%%%%%%%%%%%%%%%%%%%%%%%%%%%%%%%%%%%%%%%%%%%%%%%%%%%%%%%%%%%%%%%%%%%%%%%%%%%%%%%%%%%%%%%%%%%%%%%%%%%%%%%%
%%%%%%%%%%%%%%%%%%%%%%%%%%%%%%%%%%%%%%%%%%%%%%%%%%%%%%%%%%%%%%%%%%%%%%%%%%%%%%%%%%%%%%%%%%%%%%%%%%%%%%%%%%%%%%%%%%%%%%%%%%%%%%%%%%%%%%%%%%%%%%%%%%%%%%%%%%%%%%%%%%%

\subsection{\textit{A priori} estimate on the moments of a solution}\label{subsec:moments}

The immediate appearance of moments of any order is characterized by the following proposition.

\bigskip
\begin{prop}\label{prop:moments}
Let $f$ be a non-negative solution of $\eqref{BN}$ in $L^{\infty}_{\mbox{\scriptsize{loc}}}\left([0,T_0),L^1_{2,v}\cap L^\infty_v\right)$, with initial data $f_0$, satisfying the conservation of mass and energy.
\\ If $\gamma>0$ then for all for all $\alpha>0$ and for all $0<T<T_0$, $$\int_{\R^d} \abs{v}^{\alpha}f(t,v)\:dv \in L^{\infty}_{\mbox{\scriptsize{loc}}}\left([T,T_0)\right).$$
\end{prop}
\bigskip

The proof of that proposition is done by induction and requires two lemmas. The first lemma proves a certain control of the $L^1_{2+\gamma/2,v}$-norm and will be the base case for the induction, while the second lemma will prove an inductive bound on the moments.
\par In what follows we will rely heavily on the following technical lemma, proved in the appendix of \cite{MisWen}:

\bigskip
\begin{lemma}\label{lem:technical MisWen lemma}
Let $f_0\in L^1_{2,v}$. Then, there exists a positive convex function $\psi$ defined on $\R^+$ such that $\psi(x)=x\phi(x)$ with $\phi$ a concave function that increases to infinity and satisfies that for any $\eps>0$ and $\alpha\in(0,1)$ 
$$\left(\phi(x)-\phi(\alpha x)\right)x^\eps \underset{x\to\infty}{\longrightarrow} \infty,$$
and such that
$$\int_{\R^d}\psi\left(\abs{v}^2\right)f_0(v)\:dv <\infty.$$
\end{lemma}

In what follows we will denote by $\psi$ and $\phi$, the associated functions given by Lemma \ref{lem:technical MisWen lemma} for the initial data $f_0$.

\bigskip
\begin{lemma}\label{lem:moments1}
Let $f$ satisfy the conditions of Proposition \ref{prop:moments}. Then for any $T$ in $[0,T_0)$ there exist $c_T, C_T >0$ such that for all $0\leq t \leq T$,
\begin{eqnarray}
&&\int_{\R^d}f(t,v)\psi\left(\abs{v}^2\right) \:dv + c_T\int_0^t \int_{\R^d}f(\tau,v)\left[M_{2+\frac{\gamma}{2}}(\tau)+ \psi\left(\abs{v}^2\right)\right]\:dvd\tau \nonumber
\\&&\quad\quad\quad \leq \int_{\R^d}f_0(v)\psi\left(\abs{v}^2\right)\:dv + C_T t. \label{initialmoments}
\end{eqnarray}
\end{lemma}
\bigskip

\begin{proof}[Proof of Lemma $\ref{lem:moments1}$]
We fix $T$ in $[0,T_0)$ and we consider $0\leq t \leq T$.
\par As seen in \cite{MisWen}, one can construct an increasing sequence of convex functions, $\left(\psi_n\right)_{n\in\N}$, that converges pointwise to $\psi$ and satisfies that $\psi_{n+1}-\psi_n$ is convex. Moreover, there exists a sequence of polynomials of order $1$, $\left(p_n\right)_{n\in\N}$, such that $\psi_n - p_n$ is of compact support.
\par The properties of $\psi_n$ imply that for a given $F$ as in Lemma $\ref{lem:povzner}$ we have that the associated operators $H_{\psi_n}$, $G_{\psi_n}$ satisfy:
\begin{itemize}
\item $H_{\psi_n}$ is positive and increasing (due to Remark $\ref{remark:povzner}$).
\item $H_{\psi_n}$ converges pointwise to $H_{\psi}$ (this follows from the appropriate representation of $H$, see \cite{MisWen}).
\item $\abs{G_{\psi_n}(v,v_*)}\leq C_G \abs{v}\abs{v_*}$ for all $n$.
\end{itemize}

\bigskip
As $f$ preserves mass and energy and $p_n$ is of order $1$:
$$\int_{\R^d} \left[f(t,v) - f_0(v)\right]\psi_n\left(\abs{v}^2\right) \:dv = \int_{\R^d} \left[f(t,v) - f_0(v)\right]\left(\psi_n\left(\abs{v}^2\right)-p_n\left(\abs{v}^2\right)\right) \:dv.$$
Since $\psi_n - p_n$ is compactly supported and $f$ solves the Boltzmann-Nordheim equation, we use Lemma $\ref{lem:integralQ}$ to conclude
\begin{eqnarray*}
&&\int_{\R^d} \left[f(t,v) - f_0(v)\right]\psi_n\left(\abs{v}^2\right) \:dv
\\&&\quad\quad\quad = \frac{C_\Phi}{2}\int_0^t\int_{\R^d\times\R^d\times\mathbb{S}^{d-1}}q(f)(\tau,v,v_*)\left[\psi'_{n*} + \psi'_n - \psi_{n*} - \psi_n\right]\:dvdv_*d\tau,
\end{eqnarray*}
with
$$q(f)(\tau,v,v_*) = |v-v_*|^\gamma b\left(\mbox{cos}\:\theta\right)f(\tau)f_*(\tau)\left(1+f'(\tau)+f'_*(\tau)\right).$$

Using Lemma $\ref{lem:povzner}$ with $F=1+f'+f'_*$ we find that the above implies, using the decomposition stated in the lemma, that
\begin{eqnarray}
&&\int_{\R^d}f(t,v)\psi_n\left(\abs{v}^2\right) \:dv + \frac{C_\phi}{2}\int_0^t\int_{\R^d\times\R^d}f(\tau)f(\tau)_*\abs{v-v_*}^\gamma H_{\psi_n}\:dv_* dv d\tau \nonumber
\\&&= \int_{\R^d}f_0(v)\psi_n\left(\abs{v}^2\right) \:dv + \frac{C_\phi}{2}\int_0^t\int_{\R^d\times\R^d}f(\tau)f(\tau)_*\abs{v-v_*}^\gamma G_{\psi_n}\:dv_* dv d\tau. \nonumber
\end{eqnarray}

\bigskip
At this point the proofs follows much like in the work of Mischler and Wennberg. We concisely outline the steps for the sake of completion.
\par Using the uniform bound on $G_{\psi_n}$ and the properties of $H_{\psi_n}$ we find that by the monotone convergence theorem
\begin{eqnarray}
&&\int_{\R^d}f(t,v)\psi\left(\abs{v}^2\right) \:dv + \frac{C_\phi}{2}\int_0^t\int_{\R^d\times\R^d}f(\tau)f(\tau)_*\abs{v-v_*}^\gamma H_{\psi}\:dv_* dv d\tau \nonumber
\\&&= \int_{\R^d}f_0(v)\psi\left(\abs{v}^2\right) \:dv + \frac{C_\phi C_G}{2}\int_0^t\int_{\R^d\times\R^d}f(\tau)f(\tau)_*\abs{v-v_*}^\gamma \abs{v}\abs{v_*}\:dv_* dv d\tau. \nonumber
\end{eqnarray}

Using Lemma $\ref{lem:povzner}$ again for $H_{\psi}$ and picking $\epsilon=\frac{\gamma}{2}$ in the relevant case we have that 
\begin{equation}\nonumber
\begin{split}
\int_{\R^d\times\R^d}f(\tau)f(\tau)_*\abs{v-v_*}^\gamma H_{\psi}\:dv_* dv \geq & C\int_{\R^d \times \R^d}f(\tau)f(\tau)_*\abs{v}^{2+\frac{\gamma}{2}}dv_* dv
\\&-c \int_{\R^d\times \R^d}f(\tau)f(\tau)_*\left(\abs{v}\abs{v_*}\right)^{1+\frac{\gamma}{4}}dv_* dv.
\end{split}
\end{equation}
As
$$M_\beta(f)(\tau) \leq \norm{f(\tau)}_{L^1_{2,v}}=\norm{f_0}_{L^1_{2,v}}$$
for any $\beta \leq 2$ we conclude that due to the conservation of mass and energy we have that
\begin{equation}\nonumber
\int_{\R^d}f(t,v)\psi\left(\abs{v}^2\right) \:dv + \frac{c_T}{2}\norm{f_0}_{L^1_v}\int_0^t M_{2+\frac{\gamma}{2}}(\tau) d\tau 
\leq \int_{\R^d}f_0(v)\psi\left(\abs{v}^2\right) \:dv +C_T t.
\end{equation}
The above also implies that
\begin{equation}\nonumber
\int_{\R^d}f(t,v)\psi\left(\abs{v}^2\right) \:dv \leq  \int_{\R^d}f_0(v)\psi\left(\abs{v}^2\right) \:dv +C_T T,
\end{equation}
which is enough to complete the proof.

\end{proof}
\bigskip

Next we prove the lemma that governs the induction step. Again, the proof follows \cite{MisWen} closely, yet we include it for completion.
\bigskip
\begin{lemma}\label{lem:moments2}
Let $T$ be in $(0,T_0)$. For any $n\in \N$ there exists $T_n >0$ as small as we want such that
$$M_{2+(2n+1)\gamma/2}(T_n)<\infty.$$
Moreover, for any $t\in[T_n,T]$ there exists $C_T >0$ and $c_{T_n,T}>0$ such that 
\begin{equation}\label{inductionmoments}
M_{2+(2n+1)\gamma/2}(t) + c_T\int_{T_n}^t \left[M_{2+(2n+1)\gamma/2}(\tau) + M_{2+(2n+3)\gamma/2}(\tau)\right]\:d\tau \leq C_{T_n,T} (1+t),
\end{equation}
\end{lemma}
\bigskip

\begin{proof}[Proof of Lemma $\ref{lem:moments2}$]
We start by noticing that since $M_{2+\gamma/2}\in L^1_{loc}\left([0,T_0)\right)$, according to Lemma $\ref{lem:moments1}$ and the conservation of mass, we can find $t_0$, as small as we want, such that
$$M_{2+\gamma/2}(t_0)<\infty.$$
We repeat the proof of Lemma $\ref{lem:moments1}$ with the function $\psi(x)=x^{1+\frac{\gamma}{4}}$ on the interval $[t_0,T]$, as we can still find the same polynomial approximation and a uniform bound on the associated $H$ and $G$, to find that for almost any $t\in [t_0,T_0)$
\begin{equation}\nonumber
\begin{gathered}
\int_{\R^d\times\R^d} f(t)\abs{v}^{2+\gamma/2}dvdv_* + C_T \int_{t_0}^t f(\tau)f_*(\tau)\abs{v}^{2+3\gamma/2}dvdv_*d\tau \\
\leq c_t \int_{t_0}^t f(\tau)f_*(\tau)\abs{v}^{2+\gamma/2}\abs{v_*}^{\gamma}dvdv_*d\tau + \int_{\R^d\times\R^d}f(t_0)\abs{v}^{2+\gamma/2}dvdv_*.
\end{gathered}
\end{equation}
This completes the proof in the case $n=0$ using Lemma $\ref{lem:moments1}$ again. Notice that as the right hand side is a uniform bound in $t$ we can conclude that the inequality is valid for any $t$ in the appropriate interval.
\par We continue in that manner, using Lemma $\ref{lem:povzner}$ with $\psi(x) = x^{1+(2n+3)\gamma/4}$, assuming we have shown the result for $M_{2+(2n+1)\gamma/2}$, and conclude the proof.
\end{proof}
\bigskip

We now posses the tools to prove the main proposition of this section.

\begin{proof}[Proof of Proposition $\ref{prop:moments}$]
We start by noticing, that since $f$ conserves mass and energy $f$ is in $L^1_{2,v}$ for all $t\in[0,T_0)$ and therefore the Proposition is valid for all $\alpha\in[0,2]$.

\bigskip
Given $\alpha>2$ and $0<T<T_1<T_0$ we know by Lemma \ref{lem:moments2} that we can construct an increasing sequence $\left(T_n\right)_{n\in \N}$ such that $T_n <T$ for all $n$ and 
\begin{equation}\nonumber
M_{2+(2n+1)\gamma/2}(t) < C_{T_n,T}(1+T_1)
\end{equation}
when $t\in [T_n,T_1]\subset [T,T_1]$. This completes the proof.
\end{proof}
\bigskip

\begin{remark}
We would like to emphasize at this point that this result is slightly different from the one for the Boltzmann equation. Indeed, in the case when $T_0=+\infty$ in the Boltzmann equation the bounds on the moments on $[T,\infty)$ depend only on $T$,while for the Boltzmann-Nordheim equation in our settings we can only find local bounds on the moments since we require the boundedness of the solution $f$.
\end{remark}
\bigskip

%%%%%%%%%%%%%%%%%%%%%%%%%%%%%%%%%%%%%%%%%%%%%%%%%%%%%%%%%%%%%%%%%%%%%%%%%%%%%%%%%%%%%%%%%%%%%%%%%%%%%%%%%%%%%%%%%%%%%%%%%%%%%%%%%%%%%%%%%%%%%%%%%%%%%%%%%%%%%%%%%%%
%%%%%%%%%%%%%%%%%%%%%%%%%%%%%%%%%%%%%%%%%%%%%%%%%%%%%%%%%%%%%%%%%%%%%%%%%%%%%%%%%%%%%%%%%%%%%%%%%%%%%%%%%%%%%%%%%%%%%%%%%%%%%%%%%%%%%%%%%%%%%%%%%%%%%%%%%%%%%%%%%%%
%%%%%%%%%%%%%%%%%%%%%%%%%%%%%%%%%%%%%%%%%%%%%%%%%%%%%%%%%%%%%%%%%%%%%%%%%%%%%%%%%%%%%%%%%%%%%%%%%%%%%%%%%%%%%%%%%%%%%%%%%%%%%%%%%%%%%%%%%%%%%%%%%%%%%%%%%%%%%%%%%%%

\subsection{The rate of blow up of the $L^1_{2+\gamma,v}$-norm at $t=0$}\label{subsec:explosionM2gamma}

In this subsection we will investigate the rate by which the $(2+\gamma)^{th}$ moment blows up as $t$ approaches zero. This will play an important role in the proof of the uniqueness to the Boltzmann-Nordheim equation.

\bigskip
\begin{prop}\label{prop:explosionM2gamma}
Let $f$ be a non-negative solution of the Boltzmann-Nordheim equation in $L^{\infty}_{\mbox{\scriptsize{loc}}}\left([0,T_0),L^1_{2,v}\cap L^\infty_v\right)$ satisfying the conservation of mass and energy.
\\Then, given $T<T_0$, if $M_{2+\gamma}(t)$ is unbounded on $(0,T]$ there exists a constant $C_T >0$  such that
$$\forall t \in (0,T], \quad M_{2+\gamma}(t) \leq \frac{C_T}{t}.$$
$C_T$ depends only on $\gamma,d$, the collision kernel, $\sup_{t\in (0,T]}\norm{f}_{L^\infty_{v}}$ and the appropriate norms of $f_0$
\end{prop}
\bigskip

\begin{proof}[Proof of Proposition $\ref{prop:explosionM2gamma}$]
Let $0<t<T<T_0$. We start by mentioning that due to Proposition $\ref{prop:moments}$ we know that all the moments considered in what follows are defined and finite. Using Lemma $\ref{lem:integralQ}$ we find that
\begin{equation}\label{explosion1}
\frac{d}{dt}M_{2+\gamma}(t) = \frac{C_\phi}{2}\int_{\R^d\times\R^d}\abs{v-v_*}^\gamma f f_* K_{1+\gamma/2}(v,v_*)\:dv_* dv,
\end{equation}
where $K_{1+\gamma/2}(v,v_*)$ is given by Lemma $\ref{lem:povzner}$ with the choice $\psi(x)=x^{1+\gamma/2}$. From the same lemma we have
\begin{equation*}
\begin{split}
K_{1+\gamma/2}(v,v_*) \leq& C_{1,T}\abs{v}^{1+\gamma/2}\abs{v_*}^{1+\gamma/2} - C_{2,T}\left(\abs{v}^{2+\gamma}+\abs{v_*}^{2+\gamma}\right)
\\& + C_{3,T}\left(\abs{v}^{2+\gamma}+\abs{v_*}^{2+\gamma}\right)\mathbf{1}_{\{\abs{v}/2<\abs{v_*}<2\abs{v}\}}
\end{split}
\end{equation*}
for constants $C_{1,T}$, $C_{2,T}$, $C_{3,T}$ depending only on $\gamma,T,d,\norm{f}_{L^\infty_{[0,T],v}}$ and the appropriate norms of $f_0$.
\par On $\{\abs{v}/2<\abs{v_*}<2\abs{v}\}$
$$\abs{v}^{2+\gamma} + \abs{v_*}^{2+\gamma} \leq 2^{2+\gamma/2}\abs{v}^{1+\gamma/2}\abs{v_*}^{1+\gamma/2}.$$
Therefore, $\eqref{explosion1}$ yields
$$\frac{d}{dt}M_{2+\gamma}(t) \leq \int_{\R^d\times\R^d}\abs{v-v_*}^\gamma f f_*\left[\tilde{C}_{1,T}\abs{v}^{1+\gamma/2}\abs{v_*}^{1+\gamma/2} - C_{2,T}\abs{v}^{2+\gamma}\right]\:dv_* dv.$$

\bigskip
Since $f$ preserves the mass and energy, and $0\leq \gamma \leq 1$, we find that with abusing notations for relevant constants
\begin{equation}\label{explosion2}
\frac{d}{dt}M_{2+\gamma}(t) \leq C_{1,T} M_{1+3\gamma/2} - C_{2,T} M_{2+\gamma},
\end{equation}
where we have used the fact that $\abs{\abs{v}^\gamma-\abs{v_*}^\gamma} \leq \abs{v-v_*}^\gamma \leq \abs{v}^\gamma+\abs{v_*}^\gamma $. As, for any $\epsilon>0$ 
$$\abs{v}^{1+3\gamma/2}=\abs{v}^{1+\gamma/2}\abs{v}^\gamma \leq \eps \abs{v}^{2+\gamma} + \frac{1}{4\eps} \abs{v}^{2\gamma}\leq \eps\left(1+\abs{v}^{2+2\gamma}\right)+ \frac{1}{4\eps} \abs{v}^{2\gamma}.$$
we conclude that since $2\gamma \leq 2$ we can take $\eps$ to be small enough such that $\eqref{explosion2}$ becomes
$$\frac{d}{dt}M_{2+\gamma}(t) \leq c_{T} - C_{T} M_{2+2\gamma}(t).$$
where $c_{T}$, $C_{T}>0$ are independent of $t$ and depend only on the relevant known quantities.

\bigskip
Due to Holder's inequality we know that
$$M_{2+\gamma} \leq M_{2}^{1/2}M_{2+2\gamma}^{1/2}$$
implying that
$$\frac{d}{dt}M_{2+\gamma}(t) \leq c_T - C_T M_{2+\gamma}^2(t).$$
As $M_{2+\gamma}(t)$ is unbounded in $(0,T]$, we know that there exists $t_0\in(0,T]$ such that
$$M_{2+\gamma}(t_0)\geq \max\left(\sqrt{\frac{2c_T}{C_T}},M_{2+\gamma}(T)\right)$$.
We find that
$$\frac{d}{dt}M_{2+\gamma}(t_0) \leq \frac{C_T}{2}M^2_{2+\gamma}(t_0)-C_T M^2_{2+\gamma}(t_0) <0$$
implying that there exists a neighbourhood of $t_0$ where $M_{2+\gamma}(t)$ decreases. As this means that $M_{2+\gamma}(t) \geq \frac{2c_T}{C_T}$ to the left of $t_0$ we can repeat the above argument and conclude that $M_{2+\gamma}(t)$ decreases on $(0,t_0]$. Moreover, in this interval we have 
$$\frac{d}{dt}M_{2+\gamma} \leq -\frac{C_T}{2}M^2_{2+\gamma}.$$
The above inequality is equivalent to 
$$\frac{d}{dt}\left(\frac{1}{M_{2+\gamma}}\right)\geq \frac{C_T}{2}, $$
which implies, by integrating over $(0,t)$ and remembering that $M_{2+\gamma}$ is unbounded, that
$$\frac{1}{M_{2+\gamma}(t)} \geq \frac{C_T}{2}t$$
on $(0,t_0]$, from which the result follows.
\end{proof}
\bigskip

\section{Uniqueness of solution to the Boltzmann-Nordheim equation}\label{sec:uniqueness}

This section is dedicated to proving that if a solution to the Boltzmann-Nordheim equation exists, with appropriate conditions on the initial data, then it must be unique. The main theorem we will prove in this section is the following:

\bigskip
\begin{theorem}\label{theo:uniqueness}
Let $f_0$ be in $L^1_{2,v} \cap L_{s,v}^\infty$, where $d-1<s$. If $f$ and $g$ are two non-negative mass and energy preserving solutions of the Boltzmann-Nordheim equation with the same initial data $f_0$ that are in $L^{\infty}_{\mbox{\scriptsize{loc}}}\left([0,T_0),L^1_{2,v}\cap L^\infty_v\right)$ then $f=g$ on $[0,T_0)$.
\end{theorem}
\bigskip

The proof relies on precise estimates of the $L^1_v$, the $L^1_{2,v}$ and the $L^\infty_v$-norms of the difference of two solutions. As the difference of solutions may not have a fixed sign, these estimations require some delicacy due to a possible gain of a $\abs{v}^\gamma$ weight from the collision operator.
\par In what follows we will repeatedly denote by $C_T$ constants that depend on $d,s$, the collision kernel, $\norm{f_0}_{L^1_{2,v}\cap L^\infty_{s,v}}$, $\sup_{t\in (0,T]}\norm{f}_{L^\infty_v}$, $\sup_{t\in (0,T]}\norm{g}_{L^\infty_v}$ and $T$. Other instances will be clear form the context.\\ 
We would like to point out that if $f$ is a weak solution to the Boltzmann-Nordheim equation, i.e.
\begin{equation}\label{eq:weak}
f(t,v)=f_0(v)+\int_{0}^t Q\left(f(s,\cdot) \right)ds,
\end{equation}
with the required conservation and bounds, then similarly to Lemma \ref{lem:Q^- control}, and using Lemma \ref{lem:controlQ+infty} together with Proposition \ref{prop:integration on hyperplane of f} show us that for a fixed $v\in\R^d$ we have that $Q \left(f(s,v)\right) \in L^\infty_t([0,T])$. This implies that $f$ is actually absolutely continuous with respect to $t$ and as such we can differentiate \eqref{eq:weak} strongly with respect to $t$. \\
The above gives validation to the techniques used in the next few subsection.
\bigskip

%%%%%%%%%%%%%%%%%%%%%%%%%%%%%%%%%%%%%%%%%%%%%%%%%%%%%%%%%%%%%%%%%%%%%%%%%%%%%%%%%%%%%%%%%%%%%%%%%%%%%%%%%%%%%%%%%%%%%%%%%%%%%%%%%%%%%%%%%%%%%%%%%%%%%%%%%%%%%%%%%%%
%%%%%%%%%%%%%%%%%%%%%%%%%%%%%%%%%%%%%%%%%%%%%%%%%%%%%%%%%%%%%%%%%%%%%%%%%%%%%%%%%%%%%%%%%%%%%%%%%%%%%%%%%%%%%%%%%%%%%%%%%%%%%%%%%%%%%%%%%%%%%%%%%%%%%%%%%%%%%%%%%%%
%%%%%%%%%%%%%%%%%%%%%%%%%%%%%%%%%%%%%%%%%%%%%%%%%%%%%%%%%%%%%%%%%%%%%%%%%%%%%%%%%%%%%%%%%%%%%%%%%%%%%%%%%%%%%%%%%%%%%%%%%%%%%%%%%%%%%%%%%%%%%%%%%%%%%%%%%%%%%%%%%%%

\subsection{Evolution of $\norm{f-g}_{L^1_v}$}

The following algebraic identity will serve us many times in what follows:
\begin{equation}\label{arithmetic}
abc-def = \frac{1}{2} (a-d)(bc+ef) + \frac{a+d}{4}\left[(b-e)(c+f) + (c-f)(b+e)\right].
\end{equation}

\bigskip
\begin{lemma}\label{lem:L1v}
Let $0\leq T<T_0$. Then, there exists $C_T >0$  such that for all $t\in [0,T]$:
$$\frac{d}{dt}\norm{f-g}_{L^1_v} \leq C_T \left[\norm{f-g}_{L^1_{2,v}} + \norm{f-g}_{L^\infty_{v}}\right].$$
\end{lemma}
\bigskip

\begin{proof}[Proof of Lemma $\ref{lem:L1v}$]
Given $T\in [0,T_0)$ we have, due to Lemma $\ref{lem:integralQ}$:
\begin{equation}\label{L1ineq1}
\begin{split}
\frac{d}{dt}\norm{f-g}_{L^1_v} = \int_{\R^d} \mbox{sgn}(f-g)\partial_t \left(f-g\right)\:dv = \int_{\R^d} \mbox{sgn}(f-g)\left(Q(f)-Q(g)\right)\:dv&
\\\quad\quad\quad= \frac{C_\Phi}{2}\int_{\R^d\times\R^d\times\mathbb{S}^{d-1}} b\left(\mbox{cos}\:\theta\right)|v-v_*|^\gamma P(f,g)\left[\Psi'_* + \Psi' - \Psi_* - \Psi\right]\:d\sigma dv_*dv&,
\end{split}
\end{equation}
where $\Psi(t,v)= \mbox{sgn}(f-g)(t,v)$ and 
\begin{equation}\label{P}
P(f,g) = ff_*(1+f'+f'_*) - gg_*(1+g'+g'_*)
\end{equation}

\bigskip
It is simple to see that $\abs{\Psi'_* + \Psi' - \Psi_* - \Psi} \leq 4$ and using the algebraic identity $\eqref{arithmetic}$ we also note that
\begin{equation*}
\begin{split}
\abs{P(f,g)} \leq  C_T \Bigg(&\abs{f-g}(f_*+g_*) 
\\&\quad+(f+g)\abs{f_*-g_*} +(f+g)(f_*+g_*)\Big[\abs{f'-g'} + \abs{f'_*-g'_*}\Big]\Bigg) .
\end{split}
\end{equation*}

Using the above with $\eqref{L1ineq1}$, along with known symmetry properties, we find that
\begin{equation*}
\begin{split}
\frac{d}{dt}\norm{f-g}_{L^1_v} \leq& C_T\Bigg(\int_{\R^d\times\R^d} |v-v_*|^\gamma \abs{f-g}(f_*+g_*) \:dv_*dv 
\\ &+\int_{\R^d\times\R^d\times\mathbb{S}^{d-1}} b\left(\mbox{cos}\:\theta\right)|v-v_*|^\gamma (f+g)(f_*+g_*)\abs{f'-g'} \: dv_*dvd\sigma\Bigg).
\end{split}
\end{equation*}
As
$$\abs{v-v_*}^\gamma \leq \left(1+\abs{v}^2\right)\left(1+\abs{v_*}^2\right),$$
since $\gamma\in [0,1]$, and using the conservation of mass and energy, as well as the fact that $f$ and $g$ has the same initial condition $f_0$, we conclude that
\begin{equation*}
\begin{split}
\frac{d}{dt}\norm{f-g}_{L^1_v} \leq& C_T \left(\norm{f_0}_{L^1_{2,v}}\norm{f-g}_{L^1_{2,v}} 
+ \norm{f_0}_{L^1_{2,v}}^2 \norm{f-g}_{L^\infty_v}\right),
\end{split}
\end{equation*}
proving the desired result.
\end{proof}
\bigskip

%%%%%%%%%%%%%%%%%%%%%%%%%%%%%%%%%%%%%%%%%%%%%%%%%%%%%%%%%%%%%%%%%%%%%%%%%%%%%%%%%%%%%%%%%%%%%%%%%%%%%%%%%%%%%%%%%%%%%%%%%%%%%%%%%%%%%%%%%%%%%%%%%%%%%%%%%%%%%%%%%%%
%%%%%%%%%%%%%%%%%%%%%%%%%%%%%%%%%%%%%%%%%%%%%%%%%%%%%%%%%%%%%%%%%%%%%%%%%%%%%%%%%%%%%%%%%%%%%%%%%%%%%%%%%%%%%%%%%%%%%%%%%%%%%%%%%%%%%%%%%%%%%%%%%%%%%%%%%%%%%%%%%%%
%%%%%%%%%%%%%%%%%%%%%%%%%%%%%%%%%%%%%%%%%%%%%%%%%%%%%%%%%%%%%%%%%%%%%%%%%%%%%%%%%%%%%%%%%%%%%%%%%%%%%%%%%%%%%%%%%%%%%%%%%%%%%%%%%%%%%%%%%%%%%%%%%%%%%%%%%%%%%%%%%%%

\subsection{Evolution of $\norm{f-g}_{L^1_{2,v}}$}

The most problematic term to appear in our evolution equation is that of the $L^1_{2,v}-$norm. We have the following:

\bigskip
\begin{lemma}\label{lem:L12v}
Let $0\leq T<T_0$. Then, there exists $C_T >0$  such that for all $t\in [0,T]$:
$$\frac{d}{dt}\norm{f-g}_{L^1_{2,v}} \leq C_T \left[M_{2+\gamma}(t)\norm{f-g}_{L^1_v} + \norm{f-g}_{L^1_{2,v}} + \left(1+M_{2+\gamma}(t)\right)\norm{f-g}_{L^\infty_{v}}\right],$$
where $M_{2+\gamma}$ is the $(2+\gamma)^{\mbox{th}}$ moment of $f+g$.
\end{lemma}
\bigskip

\begin{proof}[Proof of Lemma $\ref{lem:L12v}$]
We proceed like the proof of Lemma $\ref{lem:L1v}$. For a given fixed $T\in [0,T_))$ we have:
\begin{equation}\label{L2ineq1}
\frac{d}{dt}\norm{f-g}_{L^1_{2,v}} = \frac{C_\Phi}{2}\int_{\R^d\times\R^d\times\mathbb{S}^{d-1}} b|v-v_*|^\gamma P(f,g)\left[\Psi'_* + \Psi' - \Psi_* - \Psi\right]\: dv_*dvd\sigma,
\end{equation}
with $\Psi(t,v)= \mbox{sgn}(f-g)(t,v)\left(1+\abs{v}^2\right)$ and $P(f,g)$ given by $\eqref{P}$.
\par Using the algebraic identity $\eqref{arithmetic}$ and known symmetry properties we obtain
\begin{equation}\label{L12decomposition}
\frac{d}{dt}\norm{f-g}_{L^1_v} = C_\Phi \left(\frac{1}{2}I_1 + \frac{1}{4}I_2 + \frac{1}{8}I_3 + \frac{1}{4}I_4\right)
\end{equation}
with
\begin{eqnarray*}
I_1 &=& \int_{\R^d\times\R^d\times\mathbb{S}^{d-1}} b|v-v_*|^\gamma \left[G(\Psi) - \Psi\right] (f-g)(f_*+g_*)\:d\sigma dv_*dv,
\\I_2 &=& \int_{\R^d\times\R^d\times\mathbb{S}^{d-1}} b|v-v_*|^\gamma \left[G(\Psi) - \Psi\right] (f-g)(f_*(f'+f'_*)+g_*(g'+g'_*))\:d\sigma dv_*dv,
\\I_3 &=& \int_{\R^d\times\R^d\times\mathbb{S}^{d-1}} b|v-v_*|^\gamma \left[G(\Psi) - \Psi\right] (f+g)(f_*-g_*)(f'+f'_*+g'+g'_*)\:d\sigma dv_*dv,
\\I_4 &=& \int_{\R^d\times\R^d\times\mathbb{S}^{d-1}} b|v-v_*|^\gamma \left[G(\Psi) - \Psi\right] (f+g)(f_*+g_*)(f'_* - g'_*)\:d\sigma dv_*dv,
\end{eqnarray*}
and where we defined $G(\Psi) = \Psi'_* + \Psi' - \Psi_*$. It is immediate to verify that
\begin{equation}\label{controlG}
\abs{G(\Psi)}\leq 3+\abs{v'}^2+\abs{v'_*}^2 +\abs{v_*}^2 = 2\left(1 + \abs{v_*}^2\right) + \left(1 + \abs{v}^2\right).
\end{equation}

\bigskip
Thanks to the latter bound on $G(\Psi)$ and the fact $\Psi \cdot (f-g) = \left(1+\abs{v}^2\right)\abs{f-g}$ we find that
\begin{eqnarray}
I_1 &\leq& 2l_b \int_{\R^d\times\R^d} \left(1 + \abs{v_*}^2\right)\left(\abs{v}^\gamma+\abs{v_*}^\gamma\right)\abs{f-g}(f_*+g_*)\:dv dv_* \nonumber
\\ &\leq& 4l_b \norm{f_0}_{L^1_{2,v}}\norm{f-g}_{L^1_{2,v}} + 2l_b M_{2+\gamma}\norm{f-g}_{L^1_{v}}.\label{I1}
\end{eqnarray}
where we have used similar estimation as in Lemma $\ref{lem:L1v}$.
\par The term $I_2$ is dealt similarly:
\begin{equation}\label{I2}
I_2 \leq C_T\norm{f_0}_{L^1_{2,v}}\norm{f-g}_{L^1_{2,v}} + C_T M_{2+\gamma}\norm{f-g}_{L^1_{v}}.
\end{equation}
When dealing with $I_3$ we make the symmetric change of $(v,v_*)\to(v_*,v)$ and obtain:
\begin{equation}\label{I3}
I_3 \leq C_T\norm{f_0}_{L^1_{2,v}}\norm{f-g}_{L^1_{2,v}} + C_T M_{2+\gamma}\norm{f-g}_{L^1_{v}}.
\end{equation}
\par Lastly, we find that using similar methods
\begin{eqnarray}
\abs{I_4} &\leq& 4l_b\left(\int_{\R^d\times\R^d} (1+\abs{v}^2)\left(\abs{v}^\gamma + \abs{v_*}^\gamma\right)(f+g)(f_*+g_*) \:dv_*dv\right)\norm{f-g}_{L^\infty_v} \nonumber
\\&\leq& 4 l_b  \left( 2\norm{f_0}_{L^1_{2,v}}^2 + 4\norm{f_0}_{L^1_{v}}M_{2+\gamma}\right)\norm{f-g}_{L^\infty_v}. \label{I4}
\end{eqnarray}

\bigskip
To conclude we just add $\eqref{I1}$, $\eqref{I2}$, $\eqref{I3}$ and $\eqref{I4}$ with appropriate coefficients.
\end{proof}
\bigskip

%%%%%%%%%%%%%%%%%%%%%%%%%%%%%%%%%%%%%%%%%%%%%%%%%%%%%%%%%%%%%%%%%%%%%%%%%%%%%%%%%%%%%%%%%%%%%%%%%%%%%%%%%%%%%%%%%%%%%%%%%%%%%%%%%%%%%%%%%%%%%%%%%%%%%%%%%%%%%%%%%%%
%%%%%%%%%%%%%%%%%%%%%%%%%%%%%%%%%%%%%%%%%%%%%%%%%%%%%%%%%%%%%%%%%%%%%%%%%%%%%%%%%%%%%%%%%%%%%%%%%%%%%%%%%%%%%%%%%%%%%%%%%%%%%%%%%%%%%%%%%%%%%%%%%%%%%%%%%%%%%%%%%%%
%%%%%%%%%%%%%%%%%%%%%%%%%%%%%%%%%%%%%%%%%%%%%%%%%%%%%%%%%%%%%%%%%%%%%%%%%%%%%%%%%%%%%%%%%%%%%%%%%%%%%%%%%%%%%%%%%%%%%%%%%%%%%%%%%%%%%%%%%%%%%%%%%%%%%%%%%%%%%%%%%%%

\subsection{Control of $\norm{f-g}_{L^\infty_{v}}$}

Lastly, we deal with the evolution of the $L^\infty - $norm.

\bigskip
\begin{lemma}\label{lem:Linfty}
Let $0\leq T<T_0$. Then, there exists $C_T >0$  such that for all $t\in [0,T]$:
$$\norm{f-g}_{L^\infty_{v}} \leq C_T \int_0^t \left[\norm{f-g}_{L^1_{2,v}}(u) + \norm{f-g}_{L^\infty_v}(u)\right] \:du.$$
\end{lemma}
\bigskip

\begin{proof}[Proof of Lemma $\ref{lem:Linfty}$]
Given $T\in [0,T_0)$ and $t\in [0,T]$, we have that since $f(0)=g(0)$:
\begin{equation*}
\begin{split}
\abs{f(t)-g(t)} &= \int_0^t  \mbox{sgn}(f-g)(s)\left(Q(f(s))-Q(g(s))\right)\!ds\\&=C_\Phi\int_0^t\int_{\R^d\times\mathbb{S}^{d-1}}b\left(\mbox{cos}\:\theta\right)\abs{v-v_*}^\gamma  \mbox{sgn}(f-g)P(f',g')\:d\sigma dv_*ds 
\\&\quad- C_\Phi\int_0^t \int_{\R^d\times\mathbb{S}^{d-1}}b\left(\mbox{cos}\:\theta\right)\abs{v-v_*}^\gamma  \mbox{sgn}(f-g)P(f,g)\:d\sigma dv_*ds
\\&=J_1+J_2.
\end{split}
\end{equation*}
where $P$ is given by $\eqref{P}$, and we have used the convention $f''=f$ and $g''=g$.

\bigskip
Using the algebraic identity $\eqref{arithmetic}$ and the definition of $P$ we find that:
\begin{equation*}
\begin{split}
\abs{P(f',g')} \leq & C_T\left[\abs{f'-g'}(f'_*+g'_*) + \abs{f'_*-g'_*}(f'+g')\right]
\\&+\frac{1}{4}\abs{f_*-g_*}(f'+g')(f'_*+g'_*) + \frac{1}{4}\abs{f-g}(f'+g')(f'_*+g'_*).
\end{split}
\end{equation*}
\par The change of variable $\sigma \to -\sigma$ sends $v'$ to $v'_*$ and vice versa. Thus we find that:
\begin{equation*}
\begin{split}
\abs{J_1} \leq &  C_T \int_0 ^t  \int_{\R^d\times\mathbb{S}^{d-1}}\tilde{b}\left(\mbox{cos}\:\theta\right)\abs{v-v_*}^\gamma \abs{f'-g'}(f'_*+g'_*)\:d\sigma dv_*ds
\\&+\frac{1}{4} \int_0 ^t \int_{\R^d\times\mathbb{S}^{d-1}}b\left(\mbox{cos}\:\theta\right)\abs{v-v_*}^\gamma (f'_*+g'_*)(f'+g')\abs{f_*-g_*}\:d\sigma dv_*ds
\\&+\frac{1}{4}\norm{f-g}_{L^\infty_v} \int_{\R^d\times\mathbb{S}^{d-1}}b\left(\mbox{cos}\:\theta\right)\abs{v-v_*}^\gamma (f'_*+g'_*)(f'+g')\:d\sigma dv_*ds,
\end{split}
\end{equation*}
where we defined $\tilde{b}(x) = b(x) + b(-x)$. The first term can be dealt with using the appropriate Carleman change of variables, leading to the Carleman representation $\eqref{eq: Carleman rep}$. Indeed, one can show that
\begin{equation}\nonumber
\begin{split}
&\int_{\R^d\times\mathbb{S}^{d-1}}\tilde{b}\left(\mbox{cos}\:\theta\right)\abs{v-v_*}^\gamma \abs{f'-g'}(f'_*+g'_*)\:d\sigma dv_* 
\\&\quad\quad\quad=\int_{\R^d}\frac{\abs{f'-g'}}{\abs{v-v'}}\int_{E_{vv'}}\frac{\tilde{b}(\cos \theta)\abs{v-v_*}^\gamma}{\abs{v'_*-v'}^{d-2-\gamma}}\left(f'_*+g'_*\right)\:dE(v'_*)
\\&\quad\quad\quad\leq 2 \norm{b}_{L^\infty}\norm{\int_{E_{vv'}}\left(f'_*+g'_*\right)\:dE(v'_*)}_{L^\infty_v}\norm{\int_{\R^d}\frac{\abs{f'-g'}}{\abs{v-v'}^{d-1-\gamma}}dv'}_{L^\infty_v} 
\\&\quad\quad\quad\leq C_T \left(\norm{f-g}_{L^\infty_v}+\norm{f-g}_{L^1_v}\right)
\end{split}
\end{equation}
due to Proposition $\ref{prop:integration on hyperplane of f}$ and the inequality
\begin{equation}\nonumber
\int_{\R^d}\frac{f(v)}{\abs{v-v'}^{\beta}dv} \leq C_\beta \norm{f}_{L^\infty_v}+\norm{f}_{L^1_v},
\end{equation}
when $\beta<d$. The same technique will work for the third term in $J_1$, yielding
\begin{equation}\nonumber
\begin{split}
&\int_{\R^d\times\mathbb{S}^{d-1}}b\left(\mbox{cos}\:\theta\right)\abs{v-v_*}^\gamma (f'_*+g'_*)(f'+g')\:d\sigma dv_* 
\\&\quad\quad\quad\leq \norm{b}_{L^\infty}\norm{\int_{E_{vv'}}\left(f'_*+g'_*\right)\:dE(v'_*)}_{L^\infty_v}\norm{\int_{\R^d}\frac{\abs{f'+g'}}{\abs{v-v'}^{d-1-\gamma}}dv'}_{L^\infty_v}  \leq C_T.
\end{split}
\end{equation}
We are only left with the middle term of $J_1$. Using the simple inequality
$$\abs{v-v_*}^\gamma = \abs{v'-v'_*}^\gamma \leq \left(1+\abs{v'}^\gamma\right)\left(1 + \abs{v'_*}^\gamma\right).$$
we find that
\begin{equation}\nonumber
\begin{gathered}
\int_{\R^d\times\mathbb{S}^{d-1}}b\left(\mbox{cos}\:\theta\right)\abs{v-v_*}^\gamma (f'_*+g'_*)(f'+g')\abs{f_*-g_*}\:d\sigma dv_*\\
\leq l_b \norm{f+g}^2_{L^{\infty}_{\gamma,v}}\norm{f-g}_{L^1_v} \leq C_T\norm{f-g}_{L^1_v},
\end{gathered}
\end{equation}
where we have used Theorem $\ref{theo:apriori}$ and the fact that $\gamma<d-1<s$.
\par Combining the above yields
\begin{equation}\label{eq:J1 est}
\abs{J_1} \leq C_T \int_0 ^t \left(\norm{f-g}_{L^\infty_v}+\norm{f-g}_{L^1_v}\right)\:ds.
\end{equation}

\bigskip
The term $J_2$ requires a more delicate treatment. Starting again with the algebraic identity $\eqref{arithmetic}$ we find that:
\begin{equation}\nonumber
\begin{gathered}
\abs{P(f,g)-\frac{1}{2}(f-g)\left(f_*\left(1+f'+f'_*\right)+g_*\left(1+g'+g'_*\right)\right)} \\
\leq C_T (f+g)\abs{f_*-g_*}+\frac{1}{4}(f+g)(f_*+g_*)\abs{f'-g'}+\frac{1}{4}(f+g)(f_*+g_*)\abs{f'_*-g'_*}
\end{gathered}
\end{equation}
Thus,
\begin{equation}\nonumber
\begin{gathered}
- \mbox{sgn}(f-g)P(f,g) \leq -\frac{1}{2}\abs{f-g}\left(f_*\left(1+f'+f'_*\right)+g_*\left(1+g'+g'_*\right)\right) \\
+ C_T (f+g)\abs{f_*-g_*}+\frac{1}{4}(f+g)(f_*+g_*)\abs{f'-g'}+\frac{1}{4}(f+g)(f_*+g_*)\abs{f'_*-g'_*}
\end{gathered}
\end{equation}
implying that
\begin{equation*}
\begin{split}
J_2 &\leq   C_T \int_0 ^t  \int_{\R^d\times\mathbb{S}^{d-1}}\tilde{b}\left(\mbox{cos}\:\theta\right)\abs{v-v_*}^\gamma \abs{f_*-g_*}(f+g)\:d\sigma dv_*ds
\\&\quad +\frac{1}{2} \int_0 ^t \norm{f-g}_{L^\infty_v}\int_{\R^d\times\mathbb{S}^{d-1}}b\left(\mbox{cos}\:\theta\right)\abs{v-v_*}^\gamma (f_*+g_*)(f+g)\:d\sigma dv_*ds
\\&\leq C_T\int_{0}^t \left(\norm{f+g}_{L^\infty_{v,\gamma}}\norm{f-g}_{L^1_{2,v}}+\norm{f+g}^2_{L^\infty_{v,\gamma}}\norm{f-g}_{L^\infty_{v}}\right)\:ds \\
\\&\leq C_T \int_0^t \left(\norm{f-g}_{L^1_{2,v}}+\norm{f-g}_{L^\infty_v}\right)\:ds.
\end{split}
\end{equation*}
Combining the estimations for $J_1$ and $J_2$ yields the desired result.
\end{proof} 
\bigskip

%%%%%%%%%%%%%%%%%%%%%%%%%%%%%%%%%%%%%%%%%%%%%%%%%%%%%%%%%%%%%%%%%%%%%%%%%%%%%%%%%%%%%%%%%%%%%%%%%%%%%%%%%%%%%%%%%%%%%%%%%%%%%%%%%%%%%%%%%%%%%%%%%%%%%%%%%%%%%%%%%%%
%%%%%%%%%%%%%%%%%%%%%%%%%%%%%%%%%%%%%%%%%%%%%%%%%%%%%%%%%%%%%%%%%%%%%%%%%%%%%%%%%%%%%%%%%%%%%%%%%%%%%%%%%%%%%%%%%%%%%%%%%%%%%%%%%%%%%%%%%%%%%%%%%%%%%%%%%%%%%%%%%%%
%%%%%%%%%%%%%%%%%%%%%%%%%%%%%%%%%%%%%%%%%%%%%%%%%%%%%%%%%%%%%%%%%%%%%%%%%%%%%%%%%%%%%%%%%%%%%%%%%%%%%%%%%%%%%%%%%%%%%%%%%%%%%%%%%%%%%%%%%%%%%%%%%%%%%%%%%%%%%%%%%%%

\subsection{Uniqueness of the Boltzmann-Nordheim equation}

We are finally ready to prove our main theorem for this section.

\begin{proof}[Proof of Theorem $\ref{theo:uniqueness}$]:
Combining Lemma $\ref{lem:L1v}$, $\ref{lem:L12v}$ and $\ref{lem:Linfty}$ we find that for any given $T\in [0,T_0)$ the following inequalities hold:
\begin{equation}\label{inequalitiesuniqueness}
\left\{\begin{array}{rl} &\displaystyle{\frac{d}{dt}\norm{f-g}_{L^1_v} \leq C_T \left[\norm{f-g}_{L^1_{2,v}} + \norm{f-g}_{L^\infty_{v}}\right]} \vspace{2mm} \\ \vspace{2mm} &\displaystyle{\frac{d}{dt}\norm{f-g}_{L^1_{2,v}} \leq C_T \left[M_{2+\gamma}(t)\norm{f-g}_{L^1_v} + \norm{f-g}_{L^1_{2,v}} + (1+M_{2+\gamma}(t))\norm{f-g}_{L^\infty_{v}}\right]} \vspace{2mm} \\ \vspace{2mm} &\displaystyle{\norm{f-g}_{L^\infty_{v}}\leq C_T \int_0^t \left[\norm{f-g}_{L^1_{2,v}}(u) + \norm{f-g}_{L^\infty_v}(u)\right] \:du.}\end{array}\right.,
\end{equation}
where $C_T$ can be chosen to be the same in all the inequalities.
\par As the $L^1_v$, $L^1_{2,v}$ and $L^\infty_{v}$-norms of $f$ and $g$ are bounded uniformly on $[0,T]$ we see from (\ref{inequalitiesuniqueness}) that
$$\norm{f-g}_{L^1_v} \leq C_T t,$$
$$\norm{f-g}_{L^\infty_v}\leq C_T t.$$
Moreover, due to Proposition $\ref{prop:explosionM2gamma}$ we know that the rate of blow up of $M_{2+\gamma}$ is at worst of order $1/t$. More precisely there exists a constant $C_1$ that may depend on $T,d,\gamma$, $\sup_{t\in[0,T]}\norm{f}_{L^\infty_v}$, $\sup_{t\in[0,T]}\norm{g}_{L^\infty_v}$, the appropriate norms of $f_0$, or the bound of the $(2+\gamma)^{th}$ moment if it is bounded, such that
\begin{equation}
M_{2+\gamma}(t) \leq \frac{C_1}{t}.
\end{equation}
This, together with the middle inequality of (\ref{inequalitiesuniqueness}) implies that
$$\frac{d}{dt}\norm{f-g}_{L^1_{2,v}}\leq C_T\left(C_TC_1 + 2\norm{f_0}_{L^1_{2,v}}+C_T(T+C_1)\right),$$
from which we conclude that
$$\norm{f-g}_{L^1_{2,v}}\leq \left(C^2_TC_1 + 2C_T\norm{f_0}_{L^1_{2,v}}+C^2_T(T+C_1)\right)t.$$
Iterating this process shows that there exists $C_{n,T}>0$ such that 
$$\max\left(\norm{f-g}_{L^1_v},\norm{f-g}_{L^1_{2,v}},\norm{f-g}_{L^\infty_v}\right) \leq C_{n.T}t^n,$$
though the dependency of $C_{n,T}$ on $n$ may be slightly complicated. We will continue following the spirit of Nagumo's fixed point theorem.

\bigskip 
Firstly, we notice that by defining $t_0=\min\{T,1/(2C_T)\}$, a simple estimation in the third inequality of (\ref{inequalitiesuniqueness}) shows that for any $t\in[0,t_0]$:
$$\sup_{t\in[0,t]}\norm{f-g}_{L^\infty_{v}} \leq 2C_Tt\sup_{t\in[0,t]}\norm{f-g}_{L^1_{2,v}}.$$
This, together with the second inequality of (\ref{inequalitiesuniqueness}) and the moment bounds implies that for any $t\in[0,t_0]$ we have
\begin{equation}\label{finalineq}
\frac{d}{dt}\norm{f-g}_{L^1_{2,v}} \leq \frac{K_1}{t}\norm{f-g}_{L^1_{2,v}} + K_2 \sup\limits_{[0,t]} \norm{f-g}_{L^1_{2,v}},
\end{equation}
where $K_1=C_TC_1$ and $K_2=C_T\left(1+2C_T T+2C_TC_1\right)$.\\

\bigskip
Let $n\in \N$ be such that $K_1\leq n$ and define $X(t) =\norm{f-g}_{L^1_{2,v}}/t^n$. As $X(t)\leq C_{n+2,T}t^2$ and $X(0)=0$ we conclude that $X(t)$ is differentiable at $t=0$ and as such, in $[0,t_0]$.
We have that for $t\in[0,t_0]$:
\begin{eqnarray*}
\frac{d}{dt}X(t) &=& \frac{1}{t^n}\left(\frac{d}{dt}\norm{f-g}_{L^1_{2,v}} -\frac{n}{t}\norm{f-g}_{L^1_{2,v}}\right)
\\&\leq& \frac{K_2}{t^n}\sup\limits_{[0,t]} \norm{f-g}_{L^1_{2,v}}\leq K_2\sup\limits_{[0,t]} X(u),
\end{eqnarray*}
which implies that $X(t) \leq K_2 \sup_{[0,t]}X(u) t$. Continuing by induction we conclude that for any $n\in \N$ and $t\in [0,t_0]$
$$X(t) \leq \frac{K_2^n t^n}{n!}\sup_{[0,t]}X(u).$$
Taking $n$ to infinity shows that $X(t)=0$ for all $t\in [0,t_0]$, proving that $f=g$ on that interval. If $t_0=T$ we are done, else we repeat the same arguments, starting from $t_0$ where the functions are equal, on the interval $[t_0,2t_0]$. Continuing inductively we conclude the uniqueness in $[0,T]$.
\end{proof}
\bigskip

We finally have all the tools to show Theorem \ref{theo:uniqeandproperties}
\begin{proof}[Proof of Theorem \ref{theo:uniqeandproperties}]
This follows immediately from Theorem \ref{theo:apriori}, Proposition \ref{prop:moments} and Theorem \ref{theo:uniqueness}.
\end{proof}

\section{Local existence of solutions}\label{sec:existence}

In this section we will develop the theory of existence of local in time solutions to the Boltzmann-Nordheim equation and prove Theorem \ref{theo:existence}. From this point onwards, we assume that $f_0$ is not identically $0$. 
%For the second part We notice that since $s>d+2+\gamma$ one has that $f_0$ has a finite moment of order $2+\gamma+\epsilon$, for some $\epsilon>0$, which we will use later on.
\par The method of proof we will employ to show the above theorem involves a time discretisation of equation $\eqref{BN}$ along with an approximation of the Boltzmann-Nordheim  collision operator $Q$, giving rise to a sequence of approximate solutions to the equation.
\bigskip

%%%%%%%%%%%%%%%%%%%%%%%%%%%%%%%%%%%%%%%%%%%%%%%%%%%%%%%%%%%%%%%%%%%%%%%%%%%%%%%%%%%%%%%%%
%%%%%%%%%%%%%%%%%%%%%%%%%%%%%%%%%%%%%%%%%%%%%%%%%%%%%%%%%%%%%%%%%%%%%%%%%%%%%%%%%%%%%%%%%
%%%%%%%%%%%%%%%%%%%%%%%%%%%%%%%%%%%%%%%%%%%%%%%%%%%%%%%%%%%%%%%%%%%%%%%%%%%%%%%%%%%%%%%%%

\subsection{Some properties of truncated operators}\label{subsec:truncatedoperators}

The idea of approximating the collision kernel in the case of hard potentials is a common one in the Boltzmann equation literature (see for instance \cite{Ark}\cite{Ark1} or \cite{MisWen}). For $n\in\N$, we consider the following truncated operators:

$$Q_n(f) = C_\Phi\int_{\R^d\times \mathbb{S}^{d-1}}\left(|v-v_*|\wedge n\right)^\gamma b(\theta)\left[f'f'_*(1+f+ f_*) - ff_*(1+f'+f'_*)\right]dv_*d\sigma.$$
where $x \wedge y = \min(x,y)$.
\par We associate the following natural decomposition to the truncated operators:
$$Q_n(f) = Q_n^+(f)-f Q_n^-(f) ,$$
with $Q^+$ and $Q^-$ defined as in $\eqref{Q+}-\eqref{Q-}$. We have the following:

\bigskip
\begin{lemma}\label{lem:controlQ-n}
For any $f\in L^1_{2,v}\cap L^\infty_v$ we have that:
\begin{itemize}
\item $ \norm{fQ^-_n(f)}_{L^1_{2,v}} \leq C_\Phi l_b n^\gamma \left(1+2\norm{f}_{L^\infty_v}\right)\norm{f}_{L^1_{2,v}}^2,$
\item $\norm{Q^-_n(f)}_{L^\infty_v} \leq C_\Phi l_b n^\gamma \left(1+2\norm{f}_{L^\infty_v}\right)\norm{f}_{L^1_{v}},$
\item if  $f \geq 0$, then for any $v\in\R^d$
$$Q_n^-(f)(v) \geq C_\Phi l_b \left(n^\gamma \wedge \left(1+\abs{v}^\gamma\right)\right)\norm{f}_{L^1_v} - C_\Phi C_\gamma l_b \norm{f}_{L^1_{2,v}},$$
where $C_\gamma>0$ is defined by $\eqref{Cgamma}$.
\end{itemize}
\end{lemma}
\bigskip

\begin{proof}[Proof of Lemma $\ref{lem:controlQ-n}$]
As
$$Q^-_n(f)(v) = C_\Phi \int_{\R^d\times \mathbb{S}^{d-1}}\left(n \wedge \abs{v-v_*}\right)^\gamma b(\mbox{cos} \:\theta) f_*\left[1+f'_*+f'\right]\:dv_*d\sigma.$$
The first two inequalities are easily obtained by bounding $f'_*+f'$ by $2\norm{f}_{L^\infty_v}$ and the collision kernel by $n^\gamma b(\mbox{cos} \:\theta)$.

\bigskip
To show the last inequality we use the non-negativity of $f$ and mimic the proof of Lemma $\ref{lem:Q^- control}$:
\begin{eqnarray*}
Q^-_n(f)(v) &\geq& C_\Phi \int_{\R^d\times \mathbb{S}^{d-1}}\left(n \wedge \abs{v-v_*}\right)^\gamma b(\mbox{cos} \:\theta) f_*\:dv_*d\sigma
\\ &\geq& C_\Phi l_b \left[\int_{\abs{v-v_*}\leq n}\abs{v-v_*}^\gamma  f_*\:dv_* + \int_{\abs{v-v_*}\geq n}n^\gamma  f_*\:dv_*\right]
\\&\geq& C_\Phi l_b \left[\int_{\abs{v-v_*}\leq n}\big((1+\abs{v}^\gamma)-(1+\abs{v_*}^\gamma)\big) f_*\:dv_* + \int_{\abs{v-v_*}\geq n}n^\gamma  f_*\:dv_*\right]
\\&\geq& C_\Phi l_b \left[\left(n^\gamma \wedge \left(1+\abs{v}^\gamma\right)\right)\norm{f}_{L^1_v} -C_\gamma\int_{\abs{v-v_*}\leq n} (1+\abs{v_*}^2)f_*\:dv_*\right],
\end{eqnarray*}
where $C_\gamma$ was defined in $\eqref{Cgamma}$. The proof is now complete.
\end{proof}
\bigskip

As we saw in Section $\ref{sec:apriori}$, the control of the integral of $Q^+$ over the hyperplanes $E_{vv'}$ is of great importance in the study of $L^\infty$-norm for the solutions to the Boltzmann-Nordheim equation. We thus strive to find a similar result for the $Q^+_n$ operators. 

\bigskip
\begin{lemma}\label{lem:controlQ+n}
Let $f$ be in $L^1_{2,v}\cap L^\infty_v$. Then:
\begin{itemize}
\item $ \norm{Q^+_n(f)}_{L^1_{2,v}} \leq 2C_\Phi l_b n^\gamma \left(1+2\norm{f}_{L^\infty_v}\right)\norm{f}_{L^1_{2,v}}^2,$
\item If $f \geq 0$ then for almost every $(v,v')$
$$\int_{E_{vv'}}Q^+_n(f)(v'_*)\:dE(v'_*) \leq C_{+E}\norm{f}_{L^1_v}\left(1+2\norm{f}_{L^\infty_v}\right)\left[\frac{\abs{\mathbb{S}^{d-1}}}{d+\gamma-1}\norm{f}_{L^\infty_v} + \norm{f}_{L^1_v}\right],$$
where $C_{+E}$ was defined in Lemma $\ref{lem:integration on hyperplanes of Q^+}$,
\item If there exists $E_f>0$ such that for almost every $(v,v')$
$$\int_{E_{vv'}}\abs{f'_*}\:dE(v'_*) \leq E_f$$
then
$$\norm{Q_n^+(f)}_{L^\infty_v} \leq C_{+}E_f\left(1+2\norm{f}_{L^\infty_v}\right)\left[\frac{\abs{\mathbb{S}^{d-1}}}{1+\gamma}\norm{f}_{L^\infty_v} + \norm{f}_{L^1_v}\right],$$
where $C_+$ was defined in Lemma $\ref{lem:controlQ+infty}$.
\end{itemize}
\end{lemma}
\bigskip

\begin{proof}[Proof of Lemma $\ref{lem:controlQ+n}$]
To prove the first inequality, we notice that the change of variable $(v',v'_*) \rightarrow (v,v_*)$ yields the following inequality:
$$\int_{\R^d}\left(1+\abs{v}^\gamma\right)Q^+_n(f)\:dv\leq 2\int_{\R^d}\left(1+\abs{v}^2\right)fQ^-_n(f)\:dv,$$
from which the result follows due to Lemma $\ref{lem:controlQ-n}$.

\bigskip
The last two inequalities follow respectively from the Lemma $\ref{lem:integration on hyperplanes of Q^+}$ and Lemma $\ref{lem:controlQ+infty}$, as the truncated kernel is bounded by the collision kernel, and the following inequality for $\alpha < d$:

\begin{eqnarray*}
\int_{\R^d} \frac{f(v)}{\abs{v-v_0}^\alpha}\:dv &\leq&  \frac{\abs{\mathbb{S}^{d-1}}}{d-\alpha}\norm{f}_{L^\infty_v} + \norm{f}_{L^1_v}.
\end{eqnarray*}
\end{proof}
\bigskip

%%%%%%%%%%%%%%%%%%%%%%%%%%%%%%%%%%%%%%%%%%%%%%%%%%%%%%%%%%%%%%%%%%%%%%%%%%%%%%%%%%%%%%%%%
%%%%%%%%%%%%%%%%%%%%%%%%%%%%%%%%%%%%%%%%%%%%%%%%%%%%%%%%%%%%%%%%%%%%%%%%%%%%%%%%%%%%%%%%%
%%%%%%%%%%%%%%%%%%%%%%%%%%%%%%%%%%%%%%%%%%%%%%%%%%%%%%%%%%%%%%%%%%%%%%%%%%%%%%%%%%%%%%%%%

\subsection{Construction of a sequence of approximate solutions to the truncated equation}\label{subsec:constructionseq}

In this subsection we will start our path towards showing local existence of solutions to the Boltzmann-Nordheim equation by finding solutions to the truncated Boltzmann-Nordheim equation
$$\partial_t f_n = Q_n(f_n)$$
on an interval $[0,T_0]$, when $n\in \N$ is fixed and $T_0$ is independent of $n$. We will do so by an explicit Euler scheme.\\
To simplify the writing of what follows, we denote the mass and the energy of $f_0$ respectively by $M_0$ and $M_2$ and  we introduce the following notations:

\begin{equation}\label{CL}
C_L = C_\Phi l_b M_0,
\end{equation}
\begin{equation}\label{Kinfty}
K_\infty = \frac{2\norm{f_0}_{L^\infty_v}}{\min\left(1,C_L\right)},
\end{equation}
\begin{equation}\label{Einfty}
E_\infty = \sup\limits_{(v,v') \in \R^d\times\R^d}\left(\int_{E_{vv'}}f_0(v'_*) \:dE(v'_*) \right) + C_{+E}M_0(1+2K_\infty)\left[\frac{\abs{\mathbb{S}^{d-1}}}{d+\gamma-1}K_\infty + M_0\right]
\end{equation}
and 
\begin{equation}\label{Cinfty}
C_\infty = C_\Phi C_\gamma l_b (M_0+M_2)K_\infty + C_+ E_\infty\left(1+2K_\infty\right)\left[\frac{\abs{\mathbb{S}^{d-1}}}{1+\gamma}K_\infty + M_0\right].
\end{equation}
\bigskip

We are now ready to define the time interval on which we will work:
\begin{equation}\label{T0}
T_0 = \min\left\{1 \:;\: \frac{K_\infty}{2C_\infty}\min\left(1,C_L\right)\right\}.
\end{equation}
%An important observation to make at this point is the independence of $T_0$ in $n$. A property we will strongly use later on.

\bigskip
For a fixed $n$ we consider the following explicit Euler scheme on $[0,T_0]$: for $j\in \N$ we define
\begin{equation}\label{inductionseq}
\left\{ \begin{array}{rl}&\displaystyle{f^{(0)}_{j,n}(v) = f_0(v)} \vspace{2mm} \\ \vspace{2mm} &\displaystyle{f^{(k+1)}_{j,n}(v) = f^{(k)}_{j,n}(v)\left(1-\Delta_j Q^-_n\left(f^{(k)}_n\right)\right) + \Delta_j Q^+_n\left(f^{(k)}_{j,n}\right)}, \mbox{for}\: k \in \left\{0,\dots,\left[\frac{T_0}{\Delta_j}\right]\right\}, \end{array}\right\}.
\end{equation}
where $\Delta_j$, the time step, is chosen as follows:
\begin{equation}\label{Deltan}
\Delta_j =\min\left\{ 1, \frac{1}{2 C_\Phi  l_b j n^{\gamma} M_0\left[1+2K_\infty\right]}\right\}.
\end{equation}

\bigskip
We notice the following properties of the sequence:
%\textcolor{red}{We first need to prove that the sequence $\left(f_{j,n}^{(k)}\right)_{k \in \left\{0,\dots,\left[\frac{T_0}{\Delta_j}\right]\right\}}$ satisfies the expected $L^\infty$ bounds. This is the purpose of the following lemma.}

\bigskip
\begin{prop}\label{prop:approx}
For all $k$ in $\{0,\dots,[T_0/\Delta_j]\}$, we have that $f^{(k)}_{j,n}$ satisfies:
\begin{enumerate}
\item[(i)] $f^{(k)}_{j,n} \geq 0$;
\item[(ii)] $\norm{f^{(k)}_{j,n}}_{L^1_{v}} = M_0$, $\norm{\abs{v}^2f^{(k)}_{j,n}}_{L^1_{v}} = M_2$ and $\:\int_{\R^d}vf^{(k)}_{j,n} \:dv = M_1$;
\item[(iii)]
$$f^{(k)}_{j,n}(v) \leq f_0(v) -  C_L\sum\limits_{l=0}^{k-1}\Delta_j\left(n^\gamma \wedge \left(1+\abs{v}^\gamma\right)\right)f^{(l)}_{j,n} + k\Delta_j C_\infty$$
and for almost every $(v,v')$
$$\int_{E_{vv'}} f^{(k)}_{j,n}(v'_*)\:dE(v'_*) \leq  \int_{E_{vv'}} f_0(v'_*)\:dE(v'_*) + k\Delta_j C_{+E}M_0(1+2K_\infty)\left[\frac{\abs{\mathbb{S}^{d-1}}}{d+\gamma-1}K_\infty + M_0\right] $$
\item[(iv)]
$$\sup\limits_{v\in\R^d}\left[ f_{j,n}^{(k)}(v) +C_L\Delta_j\sum\limits_{l=0}^{k-1}\left(n^\gamma \wedge \left(1+\abs{v}^\gamma\right)\right)f_{j,n}^{(l)}(v)\right] \leq K_\infty$$
and for almost every $(v,v')$,
$$\int_{E_{vv'}} f_{j,n}^{(k)}(v'_*) dE(v'_*) \leq E_\infty.$$
\end{enumerate}
\end{prop}
\bigskip

\begin{proof}[Proof of Proposition $\ref{prop:approx}$]
The proof of the proposition is done by induction. The case $k=0$ follows directly from our definitions of $K_\infty$ and $E_\infty$. We proceed to assume that the claim is valid for $k$ such that $k+1\leq \frac{T_0}{\Delta_j}$.
\par Combining Lemma $\ref{lem:controlQ-n}$ with $(ii)$ and $(iv)$ of Proposition $\ref{prop:approx}$ for $f_{j,n}^{(k)}$ we have that
$$\Delta_j \norm{Q^-_n\left(f^{(k)}_{j,n}\right)}_{L^\infty_v} \leq \Delta_j C_\Phi l_b n^\gamma M_0(1+2K_\infty) \leq \frac{1}{2}.$$
Thus, by definition of $f^{(k+1)}_{j,n}$:
$$f^{(k+1)}_{j,n}(v) \geq \frac{1}{2}f^{(k)}_{j,n}(v) + \Delta_j Q^+_n\left(f^{(k)}_{j,n}\right) \geq 0$$
as $f^{(k)}_{j,n} \geq 0$, proving $(i)$.

\bigskip
Furthermore, we have
$$
\int_{\R^d}\left(\begin{array}{c} 1 \\v \\ \abs{v}^2\end{array}\right)f^{(k+1)}_{j,n}(v)\:dv = \int_{\R^d}\left(\begin{array}{c} 1 \\v \\ \abs{v}^2\end{array}\right)f^{(k)}_{j,n}(v)\:dv + \Delta_j \int_{\R^d}\left(\begin{array}{c} 1 \\v \\ \abs{v}^2\end{array}\right) Q_n(f^{(k)}_{j,n})(v) \:dv.
$$
Since $Q_n$ satisfies the same integral properties of $Q$, we find that the last term is zero. This shows that as $f_{j,n}^{(k)}$ satisfies $(ii)$, so does $f^{(k+1)}_{j,n}$.

\bigskip
In order to prove $(iii)$ we will use the positivity of $f^{(k)}_{j,n}$ along with Lemma $\ref{lem:controlQ-n}$, and Lemma $\ref{lem:controlQ+n}$ together with property $(iv)$ for $f^{(k)}_{j,n}$. This shows that:
$$f^{(k+1)}_{j,n}(v) \leq f^{(k)}_{j,n}(v) -  C_L\Delta_j\left(n^\gamma \wedge \left(1+\abs{v}^\gamma\right)\right)f^{(k)}_{j,n} + \Delta_j C_\infty$$
proving the first part of $(iii)$. Since $Q^-_n(f^{(k)}_{j,n})$ is positive we also find for almost every $(v,v')$
\begin{equation*}
\begin{split}
&\int_{E_{vv'}}f^{(k+1)}_{j,n}(v'_*)\:dE(v'_*) \leq \int_{E_{vv'}}f^{(k)}_{j,n}\:dE(v_*) + \Delta_j  \int_{E_{vv'}}Q_n^+(f^{(k)}_{j,n})\:dE(v_*)
\\&\quad\quad\quad\leq  \int_{E_{vv'}}f^{(k)}_{j,n}\:dE(v_*)+ \Delta_j \left(C_{+E}M_0\left(1+2K_\infty\right)\left[\frac{\abs{\mathbb{S}^{d-1}}}{d+\gamma-1}K_\infty+M_0 \right]\right)
\end{split}
\end{equation*}
where we have used property $(ii)$ of Lemma $\ref{lem:controlQ+n}$, and properties $(ii)$ and $(iv)$ of $f_{j,n}^{(k)}$. Thus, the second part of $(iii)$ is valid by the same property for $f_{j,n}^{(k)}$.

\bigskip
The last property $(iv)$ is a direct consequence of $(iii)$ along with the fact that $(k+1)\Delta_j \leq T_0$, and the definition of $T_0$.
\end{proof}
\bigskip

As a discrete version of the Boltzmann-Nordheim equation, our apriori estimates in Section \ref{sec:apriori} led us to believe that we may be able to propagate moments and weighted $L^\infty$ norm in our sequence. This is indeed the case, as we will state shortly. However, it is important to notice that while the truncated kernel can be thought of as a an appropriate kernel with $\gamma=0$, in order to get bounds that are independent in $n$ we must use estimation that use the $\gamma$ given in the problem. This will lead to a drop in the power we can weight the function against.\\
The following Lemma is easy to prove using similar methods to the ones presented in Section \ref{sec:apriori}. We state it here and leave the proof to the Appendix. 
\begin{lemma}\label{lem:additional properties of the sequence}
Consider the sequence defined in \eqref{inductionseq}.
\begin{itemize}
\item[(i)] Let $s>2$, there exists $C_s >0$ (uniform constant defined in Lemma \ref{lemapp:prop of moments}) such that for any $j \geq j_0=2(1+M_2)C_s/M_0$ we have that
\begin{equation}\label{eq:preprop of moments}
\int_{\R^d}(1+\abs{v}^s) f_{j,n}^{(k)}(v)dv \leq (D_s k\Delta_j +1) \int_{\R^d}(1+\abs{v}^s) f_0(v)dv,
\end{equation}
where $D_s=4C_\Phi C_s l_b(1+2K_\infty)(1+M_2)$. 
%Moreover, if $\int_{\R^d}(1+\abs{v}^s) f_0(v)dv<\infty$ then
%\begin{equation}\label{eq:prop of moments}
%\mathcal{M}_s=\sup_{k,j,n}\int_{\R^d}(1+\abs{v}^s) f_{j,n}^{(k)}(v)dv <\infty.
%\end{equation}
\item[(ii)] If $f_0\in L^\infty_{s,v}$ when $s>d+2\gamma$ then for any $s^\prime<s-2\gamma$
$$W_{s^\prime}=\sup_{k,j\geq j_0,n}\norm{f_{j,n}^{(k)}}_{L^\infty_{s^\prime,v}} < \infty.$$
\end{itemize}
\end{lemma}

%%%%%%%%%%%%%%%%%%%%%%%%%%%%%%%%%%%%%%%%%%%%%%%%%%%%%%%%%%%%%%%%%%%%%%%%%%%%%%%%%%%%%%%%%
%%%%%%%%%%%%%%%%%%%%%%%%%%%%%%%%%%%%%%%%%%%%%%%%%%%%%%%%%%%%%%%%%%%%%%%%%%%%%%%%%%%%%%%%%
%%%%%%%%%%%%%%%%%%%%%%%%%%%%%%%%%%%%%%%%%%%%%%%%%%%%%%%%%%%%%%%%%%%%%%%%%%%%%%%%%%%%%%%%%

\subsection{Convergence towards a mass and momentum preserving solution of the truncated Boltzmann-Nordheim equation}\label{subsec:convergenceseq}

In the previous subsection we have constructed a family of functions $\left(f_{j,n}^{(k)}\right)_{k \in \left\{0,\dots,[T_0/\Delta_j]\right\}}$ in $L^1_{2,v}\cap L^\infty_{s',v}$, for $s'<s-2\gamma$, with the same mass and energy as the initial data $f_0$. Our next goal is to use this family in order to find a sequence of functions, $\left(f_{j,n}\right)_{j \in \N}$ in $L^1([0,T_0]\times\R^d)\cap L^{\infty}\left([0,T_0];L^{\infty}_{s',v}(\R^d) \right)$ that converges strongly to a solution of the truncated Boltzmann-Nordheim equation, while preserving the mass and energy of the initial data. The construction of such sequence is fairly straight forward - we view the sequence $\left(f_{j,n}^{(k)}\right)_{k \in \left\{0,[\dots,T_0/\Delta_j]+1\right\}}$ as a constant in time sequence of functions and construct a piecewise function using them. Indeed, we define for any $j\in \N$:
\begin{equation}\label{fn}
f_{j,n}(t,v)=f_{j,n}^{(k)}(v) \quad (t,v) \in [k\Delta_j,(k+1)\Delta_j)\times\R^d,
\end{equation}
where we replace of $([T_0/\Delta_j]+1) \Delta_j$ by $T_0$.

\bigskip
\begin{prop}\label{prop:masspreservsol}
Let $f_0\in L^1_{2,v}\cap L^{\infty}_{s,v}$ for $s>d+2\gamma$. Then, the sequence $\left(f_{j,n}\right)_{j\in \N}$ converges strongly in $L^1([0,T_0]\times\R^d)$ to a function $f_n$ that belongs to $L^1([0,T_0]\times\R^d)$ and $L^{\infty}\left([0,T_0];L^{\infty}_{s',v}(\R^d)\right)$. Moreover:
\begin{itemize}
\item[(i)] $f_n$ is a solution of the truncated Boltzmann-Nordheim equation $\eqref{BN}$ with $Q$ replaced by $Q_n$ and initial data $f_0$,
\item[(ii)] $f_n$ is positive and for all $t$ in $[0,T_0]$, $\norm{\psi (\cdot)f_n(t,\cdot)}_{L^1_v} = \norm{\psi f_0}_{L^1_v}$ for $\psi(v)=1,v,\abs{v}^2$. 
\item[(iii)] $f_n$ satisfies 
$$\sup_{t\leq T_0}\norm{f_n(t,\cdot)}_{L^\infty_v} \leq K_\infty \quad\mbox{and}\quad \sup_{t\leq T_0}\norm{f_n(t,\cdot)}_{L^\infty_{s',v}} \leq W_{s'}$$
for any $s'<s-2\gamma$, where $W_{s'}$ has been defined in Lemma \ref{lem:additional properties of the sequence}.
\end{itemize}
%$$\sup\limits_{[0,T_0]\times\R^d}\left(f(t,v) + \int_0^t\left(1+\abs{v}^\gamma\right)f(s,v)\:ds\right) \leq 2K_\infty,$$
%where $K_\infty >0$ is given by $\eqref{Kinfty}$.
\end{prop}
\bigskip

\begin{proof}[Proof of Proposition $\ref{prop:masspreservsol}$]
For simplicity in the proof we will drop the subscript $n$.
% Also, to simplify some expressions in what is to follow, we will assume that $C_L=1$.
We start by noticing that by its definition and Proposition \ref{prop:approx}, $\left\{f_j \right\}_{j\in\N}$ has the same mass, energy and momentum as $f_0$.\\
We will now show that $\left(f_j\right)_{j\in\N}$ is a Cauchy sequence in $L^1([0,T_0]\times\R^d)$. Indeed, by its definition we find that
$$f_j^{(k)}(v)-f_j^{(0)}(v)=\Delta_j \sum_{l=0}^{k-1}Q_n\left(f_j^{(l)}\right)(v).$$
This, combined with the definition of $f_{j}$, shows that if $t\in [k\Delta_j,(k+1)\Delta_j)$ we have that
\begin{equation}\label{eq:almost solution}
f_j(t,v)=f_0(v) + \int_{0}^{t}Q_n\left(f_j(s,v)\right)\:ds - \left(t-k\Delta_j\right)Q_n\left(f_j^{(k)}\right).
\end{equation}
For a given $j\leq l$ we see that
$$\norm{f_j(t,\cdot)-f_l(t,\cdot)}_{L^1_v} \leq\abs{ \int_{0}^{t}Q^{+}_n(f_j(s,\cdot))-Q^{+}_n(f_l(s,\cdot))dvds}+ E_j$$
where we have used Lemma \ref{lem:controlQ+n} and Proposition \ref{prop:approx} and the symmetry of the collision operators. We also denoted 
$$E_j = 4\left(1+2K_\infty\right)\left(C_{+}E_\infty\left[\frac{\abs{\mathbb{S}^{d-1}}}{1+\gamma}K_\infty+M_0\right]+C_\Phi l_b n^\gamma(M_0+M_2) \right)\Delta_j.$$
We conclude that, using the algebraic property $\eqref{arithmetic}$,
\begin{equation*}
\begin{split}
\norm{f_j(t,\cdot)-f_l(t,\cdot)}_{L^1_v} \leq& C_\Phi l_b n^\gamma(1+2K_\infty) \int_{0}^t \int_{\R^d\times \R^d } \left(f_j+f_l\right)\abs{f_{j,*}-f_{l,*} }dv_*dvds
\\&+2C_\Phi n^\gamma \int_0^t \int_{\R^d\times\R^d\times \mathbb{S}^{d-1}}b(\cos \theta) f_jf_{j,*}\abs{f_j^{\prime}-f_l^{\prime}}dvdv_*d\sigma + E_j
\end{split}
\end{equation*}
Next, using Lemma \ref{lem:additional properties of the sequence} we have that
$$f_i(v)f_i(v_*) \leq \frac{W^2_{s^\prime}}{(1+\abs{v}^{s^\prime})(1+\abs{v_*}^{s^\prime})} \leq \frac{W_{s^\prime}^2}{1+2^{\frac{2-s^\prime}{2}}\left(\abs{v}^2+\abs{v_*}^2 \right)^{\frac{s^\prime}{2}}},$$
for any $i\in \N$. Thus,
\begin{equation*}
\begin{split}
\int_0^t& \int_{\R^d\times\R^d\times \mathbb{S}^{d-1}}b(\cos \theta) f_jf_{j,*}\abs{f_j^{\prime}-f_l^{\prime}}dvdv_*d\sigma 
\\&\leq l_b W_{s'}^2 \int_{\R^d\times \R^d} \frac{\abs{f_j-f_l}}{1+2^{\frac{s^\prime-\gamma}{2}}\abs{v_*}^{s^\prime}}dvdv_*=l_b C_{s^\prime}W^2_{s^\prime}\int_{0}^t\norm{f_j(s,\cdot)-f_l(s,\cdot)}_{L^1_v}ds 
\end{split}
\end{equation*}
if $s^\prime>d$, where $C_{s^\prime}$ is a uniform constant. Since $s>d+2\gamma$ and we can pick any $s^\prime$ up to $s-2\gamma$, the above is valid for $s^\prime=d+\epsilon$ for $\epsilon>0$ small enough.
Thus,
$$\norm{f_j(t,\cdot)-f_l(t,\cdot)}_{L^1_v}  \leq C_n \int_{0}^t\norm{f_j(s,\cdot)-f_l(s,\cdot)}_{L^1_v}ds+E_j$$
where $C_n>0$ is independent of $j$, $l$ and $t$. From which we conclude that
$$\norm{f_j(t,\cdot)-f_l(t,\cdot)}_{L^1_v}  \leq E_j e^{C_nt} $$
As $E_j$ goes to zero as $j$ goes to infinity, the above shows that $\left\{ f_j\right\}_{j\in N}$ converges to a function $f$ both in $L^1_v$ for any fixed $t$ and in $L^{1}\left( [0,T_0]\times \R^d \right)$. Thanks to \\
Passing to an appropriate subsequence, which we still denote by $\left\{ f_j \right\}_{j\in \N}$ we can assume that $f_j$ converges pointwise to $f$ almost everywhere. As such, the preservation of mass and $(iii)$ follow immediately form the associated properties of the sequence. Moreover, thanks to $\eqref{eq:almost solution}$ and the strong convergence we just showed, we conclude that
$$f(t,v)=f_0(v)+\int_{0}^t Q_n(f(s,v))ds,$$
showing $(i)$.

\bigskip
We are only left with showing the conservation of momentum and energy. Using Fatou's lemma we find that
$$\int_{\R^d}\abs{v}^2 f(v)dv \leq \liminf_{j\rightarrow\infty}\int_{\R^d}\abs{v}^2 f_j(v)dv =\int_{\R^d}\abs{v}^2 f_0(v)dv. $$
If we show tightness of the sequence $\left\{ \abs{v}^2 f_j(v) \right\}_{j\in\N}$, i.e. that for any $\epsilon>0$ there exists $R_\epsilon>0$ such that
$$\sup_{j\in\N}\int_{\abs{v}>R_\epsilon}\abs{v}^2f_j(v)dv < \epsilon$$
then the converse will be valid and we would show the conservation of the energy. To prove this we recall Lemma $\ref{lem:technical MisWen lemma}$ for $f_0$ and denote the appropriate convex function by $\psi$. We claim that there exists $C>0$, depending only on the initial data, $\gamma,d$ and the collision kernel \emph{but not $j$} such that for all $j\in\N$, for all $k\in\{0,\dots,\left[T_0/\Delta_j\right]+1\}$,
\begin{equation}\label{eq:conservation mid step}
\int_{\R^d} f^{(k)}_{j}(v)\psi\left(\abs{v}^2\right)\:dv \leq \int_{\R^d}\psi\left(\abs{v}^2\right)f_0(v)\:dv +Ck\Delta_j \norm{f_0}^2_{L^1_{2,v}},
\end{equation}
This will imply that
\begin{equation}\label{eq:uniformtightness}
\int_{\R^d} f_j(t,v)\psi\left(\abs{v}^2\right)\:dv \leq \int_{\R^d}\psi\left(\abs{v}^2\right)f_0(v)\:dv + C\norm{f_0}^2_{L^1_{2,v}},
\end{equation}
from which the desired result follows as $\psi(x)=x\phi(x)$ for a concave, increasing to infinity function $\phi$.
\par  We prove $\eqref{eq:conservation mid step}$ by induction. The case $k=0$ is trivial and we proceed to assume that $(k+1)\Delta_j\leq T_0$ and that inequality $\eqref{eq:conservation mid step}$ is valid for $f_{j}^{(k)}$. Defining
$$M^{(k)}_j = \int_{\R^d} f^{(k)}_j(v)\psi\left(\abs{v}^2\right)\:dv$$
and using the definition of $f_j^{(k)}$ and Lemma $\ref{lem:integralQ}$  we find that
\begin{equation}\label{ineqtight1}
\begin{split}
M^{(k+1)}_j &= M^{(k)}_j + \Delta_j \int_{\R^d}\psi\left(\abs{v}^2\right)Q_n(f^{(k)}_j)(v)\:dv 
\\&= M^{(k)}_j+ \frac{C_\Phi \Delta_j}{2} \int_{\R^d\times\R^d}\left(n\wedge \abs{v-v_*}\right)^\gamma f^{(k)}_j(v)f^{(k)}_j(v_*)
\\&\quad\quad\times\left[\int_{\mathbb{S}^{d-1}}\left[1+f^{(k)}_j(v')+f^{(k)}_j(v'_*)\right]b(\mbox{cos}\:\theta)\left(\psi'_*+\psi'-\psi_*-\psi\right)d\sigma\right]dv_*dv 
\\&=  M^{(k)}_j
\\&\quad+ \frac{C_\Phi \Delta_j}{2} \int_{\R^d\times\R^d}\left(n\wedge \abs{v-v_*}\right)^\gamma f^{(k)}_j(v)f^{(k)}_j(v_*)\left[G(v,v_*)-H(v,v_*)\right]dv_*dv.
\end{split}
\end{equation}
for the appropriate $G$ and $H$ given by Lemma $\ref{lem:povzner}$. Moreover, Lemma $\ref{lem:povzner}$ implies that 
\begin{eqnarray*}
G(v,v_*) &\leq& C_G \abs{v}\abs{v_*},
\\ H(v,v_*) &\geq& 0,
\end{eqnarray*}
where $C_G$ depends only on the collision kernel, $\gamma,d$ and possibly the mass of $f^{(k)}_j$. As the latter is uniformly bounded for all $j$ and $k$ by $K_\infty$ we can assume that $C_G$ is a constant that is independent of $j$ and $k$. From the above we conclude that
\begin{eqnarray*}
M^{(k+1)}_j &\leq& M^{(k)}_j + \frac{C_\Phi \Delta_j}{2}C_G \int_{\R^d\times\R^d}\abs{v-v_*}^\gamma \abs{v}\abs{v_*}f^{(k)}_j(v)f^{(k)}_j(v_*) \:dv_*dv
\\&\leq& M^{(k)}_j + \frac{C_\Phi\Delta_j}{2}C_G \left[\int_{\R^d}\left(1+\abs{v}^\gamma\right)\abs{v}f^{(k)}_j(v)\:dv \right]^2
\\&\leq& M^{(k)}_j + C_\Phi\Delta_j C_G \norm{f_j^{(k)}}_{L^1_{2,v}}^2 
\\&\leq& \int_{\R^d}\psi\left(\abs{v}^2\right)f_0(v) +\Delta_j\left(Ck+C_\Phi C_G\right) \norm{f_0}^2_{L^1_{2,v}}
\end{eqnarray*}
showing the desired result for the choice $C=C_\Phi C_G$. A similar, yet simpler, proof (as we have bounded second moment) shows the conservation of momentum.
%We start by noticing that property $(ii)$ of Proposition $\ref{prop:approx}$ implies that $\left(f_j\right)_{j\in\N}$ is uniformly bounded in $j$ on $L^1([0,T_0]\times\R^d)$, as well as have a uniform in $j$ bound on its second moment. Moreover, as property $(iv)$ of Proposition $\ref{prop:approx}$ gives us a uniform in $j$ $L^\infty-$bound on $\left(f_j\right)_{j\in\N}$ we can conclude that the sequence is equi-integrable. Applying the Dunford-Pettis theorem we conclude that $\left(f_j\right)_{n \in \N}$ is weakly compact in $L^1([0,T_0]\times\R^d)$. Thus, there exists $f$ in $L^1([0,T_0]\times\R^d)$ and a subsequence of $\left(f_{j_i}\right)_{i\in\N}$, that we will still denote by $f_j$, which converges weakly to $f$ in $L^1([0,T_0]\times\R^d)$. This is the desired function. 
\end{proof}

\subsection{Existence of a solution to the Boltzmann-Nordheim equation}

Now that we have solutions to the truncated equation, we are ready to show the existence theorem.
\begin{proof}[Proof of Theorem \ref{theo:existence}]
If $\gamma=0$ then the truncated equation is actually the full equation. As such, Proposition \ref{prop:masspreservsol} shows $(i)$. From now on we will assume that $\gamma>0$ and $s>d+2+\gamma$. We notice that in that case there exists $\epsilon>0$ such that the $f_0\in L^1_{2+\gamma+\epsilon,v}$.\\
We denote by $\left\{ f_n \right\}_{n\in\N}$ the solutions to the truncated equation
$$\begin{cases}
\partial_t f_n(t)=Q_n(f_n) & t>0,v\in\R^d \\
f(0,v)=f_0(v) & v\in\R^d
\end{cases}$$
given by Proposition \ref{prop:masspreservsol}. We will show that the sequence is Cauchy. In what follows, unless specified otherwise, constants $C$ that appear will depend on $K_\infty,E_\infty,C_\infty,T_0$ and $f_0$ \emph{but not on $n$, $m$ and $t$}. Assuming that $n\geq m$ and following the same technique as the one in Lemma \ref{lem:L12v} we see that
\begin{equation*}
\begin{split}
\frac{d}{dt}\norm{f_n(t)-f_m(t)}_{L^1_{2,v}} &=\int_{\R^d}\text{sgn}(f_n(t)-f_m(t))(1+\abs{v}^2)\left(Q_n(f_n(t))-Q_n(f_m(t)) \right)dv
\\&+\int_{\R^d}\text{sgn}(f_n(t)-f_m(t))(1+\abs{v}^2)\left(Q_n(f_m(t))-Q_m(f_m(t)) \right)dv
\\&=I_1+I_2.
\end{split}
\end{equation*}
Exactly as in Lemma \ref{lem:L12v}
$$I_1 \leq C \left(1+M_{2+\gamma}(f_n+f_m) \right)\left(\norm{f_n-f_m}_{L^1_{2,v}}+\norm{f_n-f_m}_{L^\infty_{v}} \right)$$
Since $M_{2+\gamma}(f_0)<\infty$ we find that, due to Lemma \ref{lem:additional properties of the sequence}, the sequence $\left\{f^{(k)}_{j,n} \right\}_{j,n\in\N}$, and as such, our $f_n$, have a uniform bound, depending on $f_0$, on their moment of order $2+\gamma$. Thus,
$$I_1 \leq C_1\left(\norm{f_n-f_m}_{L^1_{2,v}}+\norm{f_n-f_m}_{L^\infty_{v}} \right)$$
\par For the second term we notice that
\begin{equation}\label{eq:existproof II.0}
\begin{split}
&\abs{Q_n(f_m)(v)-Q_m(f_m)(v)} 
\\&\quad\quad\quad\leq C_\Phi b_\infty (1+2K_\infty)\int_{\{\abs{v-v_*}\geq m\}\times \mathbb{S}^{d-1}}\abs{v-v_*}^\gamma \left(f_m^\prime f^\prime_{m,*}+f_m f_{m,*}\right) dv_* d\sigma
\end{split}
\end{equation}
and as such
\begin{equation*}
\begin{split}
\int_{\R^d\times\R^d\times\mathbb{S}^{d-1}}&(1+\abs{v}^2) \abs{Q_n(f_m)(v)-Q_m(f_m)(v)}dv
\\&\leq \frac{C}{n^\eps}\int_{\R^d\times \R^d \times \mathbb{S}^{d-1}}\left(1+\abs{v}^2+\abs{v_*}^2\right)\abs{v-v_*}^{\gamma+\epsilon} f_mf_{m,*}\:dvdv_* d\sigma
\\&\leq \frac{C}{n^\eps}(M_0+M_2+\mathcal{M}_{2+\gamma+\epsilon})^2
\end{split}
\end{equation*}
where $\mathcal{M}_{2+\gamma+\epsilon}$ is a uniform bound on the $2+\gamma+\epsilon$ moment of all $\left\{f_n\right\}_{n\in\N}$, depending only on $f_0$ and other parameters of the problem. We conclude that
$$I_2 \leq \frac{C_2}{m^\epsilon}.$$
Thus, 
\begin{equation}\label{eq:massenergypart}
\frac{d}{dt}\norm{f_n(t)-f_m(t)}_{L^1_{2,v}} \leq C_1\left(\norm{f_n-f_m}_{L^1_{2,v}}+\norm{f_n-f_m}_{L^\infty_{v}} \right)+\frac{C_2}{m^\epsilon}.
\end{equation}
\par Next, we turn our attention to the $L^\infty$ norm. In order to do that we notice that due to Proposition \ref{prop:masspreservsol} 
\begin{equation}\label{eq:existproof II.1}
\begin{split}
\abs{v-v_*}^\alpha f_m^\prime f^\prime_{m,*} &\leq C_{s^\prime,\alpha}W_{s^\prime}^2 \left( \frac{1}{\left(1+\abs{v^\prime}^{s^\prime-\alpha}\right)}\frac{1}{\left(1+\abs{v^\prime_*}^{s^\prime}\right)} 
+ \frac{1}{\left(1+\abs{v^\prime_*}^{s^\prime-\alpha}\right)}\frac{1}{\left(1+\abs{v^\prime}^{s^\prime}\right)} \right),
\\&\leq \frac{\tilde{D}_{s^\prime,\alpha}}{1+\abs{v_*}^{s^\prime-\alpha}}
\end{split}
\end{equation}
and the same holds replacing $(v',v'_*)$ by $(v,v_*)$. We therefore see that by choosing $\alpha=\gamma+\epsilon$, if $s^\prime-d>\gamma+\epsilon$ we have that 
\begin{equation}\label{eq:existproof III}
\abs{Q_n(f_m)(v)-Q_m(f_m)(v)} \leq \frac{D_{s^\prime,\alpha}}{m^{\epsilon}}.
\end{equation}
where $D_{s^\prime,\alpha}$ is a constant that depends only on the parameters on the problems. Due to Lemma \ref{lem:additional properties of the sequence} we know that we can choose $s^\prime$ as close as we want to $$s-2\gamma>d+2-\gamma \geq d+\gamma.$$
Using the above, and Lemma \ref{lem:Linfty} we see that
\begin{equation*}
\begin{split}
\norm{f_n(t)-f_m(t)}_{L^\infty_v}\leq& C \int_{0}^t \left(\norm{f_n-f_m}_{L^1_{2,v}}+\norm{f_n-f_m}_{L^\infty_{v}} \right)ds 
\\&+ \int_{0}^t \abs{Q_n(f_m)(v)-Q_m(f_m)(v)}ds
\\\leq&  C \int_{0}^t \left(\norm{f_n-f_m}_{L^1_{2,v}}+\norm{f_n-f_m}_{L^\infty_{v}} \right)ds+\frac{D_{s^\prime,\alpha}T_0}{m^{\epsilon}}.
\end{split}
\end{equation*}
Following in the steps of the proof of the uniqueness in Section \ref{sec:uniqueness} we choose $t_0$, depending only on $T_0$ and $D_0$ such that for all $t\leq t_0$
$$\sup_{s\in[0,t]}\norm{f_n-f_m}_{L^\infty_v} \leq C_3\int_{0}^t \norm{f_n-f_m}_{L^1_{2,v}}ds+\frac{2D_{s^\prime,\alpha}T_0}{m^{\epsilon}}.$$
Combining this with the integral version of $\eqref{eq:massenergypart}$ gives us that in $[0,t_0]$
\begin{equation*}
\begin{split}
\norm{f_n(t)-f_m(t)}_{L^1_{2,v}} \leq& \norm{f_n(0)-f_m(0)}_{L^1_{2,v}} 
\\&+ C \left(\int_{0}^t (1+t-s)\norm{f_n(s)-f_m(s)}_{L^1_{2,v}}ds + \frac{1}{m^\eps}\right).
\end{split}
\end{equation*}
As $f_n(0)=f_m(0)$ we find that the above is enough to show that $\left\{f_n \right\}_{n\in\N}$ is Cauchy in $L^{1}_{2,v}$ as well as $L^1_{2,t,v}$ for $t\in[0,t_0]$. As $t_0$ was independent of $n$, $m$ and the bound that we used are valid for all $t\in[0,T_0]$ we can use the fact that $\left\{f_n(t_0) \right\}_{n\in\N}$ is Cauchy and repeat the process. This shows that the sequence is Cauchy in all of $[0,T_0]$ and we denote by $f$ its limit in $L^{1}_{2,t,v}$. Using the strong convergence of $f_n$ to $f$, and the fact that $f_n$ solves the truncated equation, we conclude that for any such $\phi$
$$\int_{0}^t \int_{\R^d}\phi(t,v)\left(f(t,v)-f_0(v)-\int_0^t Q(f)(s,v)ds \right)=0$$
which shows that $f$ is indeed the desired solution.\\
Since the convergence of $f_n$ to $f$ is in $L^{1}_{2,t,v}$ we conclude the conservation of mass, momentum and energy.

\bigskip
Lastly, we notice that $T_0$, the time we have worked with from the sequence $\left\{f_{j,n}^{(k)} \right\}_{j,n}$, depends only on $f_0$ and parameters of the collision. Thus If $\norm{f(t,\cdot)}_{L^\infty_v}$ is bounded on $[0,T_0]$ we can use Theorem \ref{theo:existence} together with the conservation of mass, momentum and energy, to repeat our arguments and extend the time under which the solution exists. We conclude that we can 'push' our solution up to a time $T_{max}$ such that
$$\limsup\limits_{T \to T_{max}^-}\norm{f}_{L^\infty_{[0,T]\times\R^d}} = + \infty.$$
This completes the proof.
\end{proof}
\bigskip

We end this section with a few remarks.
\begin{remark}\hspace{1mm}
\begin{enumerate}
\item[(i)] As we have shown existence of a mass, momentum and energy conserving solution to the Boltzmann-Nordheim equation that is in $L^{\infty}_{\mbox{\scriptsize{loc}}}\left([0,T_0),L^1_{2,v} \cap L_v^\infty \right)$, the \textit{a priori} estimate given by Theorem $\ref{theo:apriori}$ actually improve our regularity of the solution and we learn that $f$ belongs to $L^{\infty}_{\mbox{\scriptsize{loc}}}\left([0,T_0),L^1_{2,v} \cap L_{s',v}^\infty \right)$ for all $s'<s$. 
\item[(ii)] Note that we have given an explicit way to find the solution, as all our sequences converge strongly.
\end{enumerate}
\end{remark}
\bigskip

%\input{globalexistence_final}

%% APPENDIX %%%%

\appendix

\section{Simple Computations}\label{simplecomputation}
We gather a few simple computations in this Appendix to make some of the proofs of the paper more coherent, without breaking the flow of the paper.
\bigskip

%%%%%%%%%%%%%%%%%%%%%%%%%%%%%%%%%%%%%%%%%%%%%%%%%%%%%%%%%%%%%%%%%%%%%%%%%%%%%%%%%%%%%%%%%%%%%%%%%%%%%%%%%%%%%%%%%%%%%%%%%%%%%%%%%%%%%%%%%%%%%%%%%%%%%%%%%%%%%%%%%%%
%%%%%%%%%%%%%%%%%%%%%%%%%%%%%%%%%%%%%%%%%%%%%%%%%%%%%%%%%%%%%%%%%%%%%%%%%%%%%%%%%%%%%%%%%%%%%%%%%%%%%%%%%%%%%%%%%%%%%%%%%%%%%%%%%%%%%%%%%%%%%%%%%%%%%%%%%%%%%%%%%%%
%%%%%%%%%%%%%%%%%%%%%%%%%%%%%%%%%%%%%%%%%%%%%%%%%%%%%%%%%%%%%%%%%%%%%%%%%%%%%%%%%%%%%%%%%%%%%%%%%%%%%%%%%%%%%%%%%%%%%%%%%%%%%%%%%%%%%%%%%%%%%%%%%%%%%%%%%%%%%%%%%%%

\subsection{Proof of Lemma $\ref{lem:integration lemma part II}$}

We start by noticing that 
\begin{equation*}
\begin{split}
\int_{\R^d}f(v_*)\abs{v-v_*}^{-\alpha}dv_* &=\int_{\abs{v-v_*}<1}f(v_*)\abs{v-v_*}^{-\alpha}\:dv_*+\int_{\abs{v-v_*}>1}f(v_*)\:dv_* 
\\&\leq \norm{f}_{L^\infty_v}\int_{\abs{x}<1}\abs{x}^{-\alpha}\:dx + \norm{f}_{L^1_v}=C_{d,\alpha}\norm{f}_{L^\infty_v}+ \norm{f}_{L^1_v}
\end{split}
\end{equation*}
implying that the required integral is uniformly bounded in all $v$ as $0\leq \alpha <d$. Thus, in order to prove the Lemma we can assume that $\abs{v}>1$. For such a $v$ consider the sets
\begin{equation}\nonumber
A=\left\{v_* \in \R^d; \quad \abs{v_*}\leq \frac{\abs{v}}{2}\right\},
\end{equation}
\begin{equation}\nonumber
B=\left\{v_* \in\R^d; \quad \abs{v-v_*}\leq \frac{\abs{v}^{\frac{s_2-s_1}{d}}}{2}\right\},
\end{equation}
and $C=\left(A\cup B\right)^c$. 

\bigskip
We have that
$$\int_{A}f(v_*)\abs{v-v_*}^{-\alpha}\:dv_* \leq 2^{\alpha}\abs{v}^{-\alpha}\int_{A}f(v_*)\:dv_* \leq 2^{\alpha}\abs{v}^{-\alpha} \norm{f}_{L^1_v}$$
as $\abs{v-v_*} \geq \abs{v}-\abs{v_*}\geq \abs{v}/2$.
\par Next, we notice that if $v_*\in B$ and $\abs{v}>1$ then since $s_2-s_1<d$ we have that
\begin{equation}\nonumber
\abs{v_*}\geq \abs{v}-\abs{v-v_*}\geq \abs{v}-\frac{\abs{v}^{\frac{s_2-s_1}{d}}}{2} =\abs{v}\left(1-\frac{1}{2\abs{v}^{\frac{d-(s_2-s_1)}{d}}}\right)\geq \frac{\abs{v}}{2}.
\end{equation}
Thus
\begin{equation*}
\begin{split}
\int_{B}f(v_*)\abs{v-v_*}^{-\alpha}\:dv_* &\leq \norm{f}_{L^\infty_{s_2,v}}\int_{B}\left(1+\abs{v_*}^{s_2}\right)^{-1}\abs{v-v_*}^{-\alpha}\:dv_* 
\\&\leq 2^{s_2}\norm{f}_{L^\infty_{s_2,v}} \abs{v}^{-s_2}\int_{\abs{x}\leq \frac{\abs{v}^{\frac{s_2-s_1}{d}}}{2}}\abs{x}^{-\alpha}\:dx
\\&=\frac{2^{s_2}\abs{\mathbb{S}^{d-1}}}{d-\alpha}\norm{f}_{L^\infty_{s_2,v}} \abs{v}^{-s_2+\frac{(d-\alpha)(s_2-s_1)}{d}}
\\&=C_{d,\alpha}\norm{f}_{L^\infty_{s_2,v}}\abs{v}^{-s_1-\frac{\alpha(s_2-s_1)}{d}}.
\end{split}
\end{equation*}

\par Lastly, when $v_*\in C$ we have that $\abs{v_*}\geq \frac{\abs{v}}{2}$ and $\abs{v-v_*} \geq \frac{\abs{v}^{\frac{s_2-s_1}{d}}}{2}$. Thus
\begin{equation*}
\begin{split}
\int_{C}f(v_*)\abs{v-v_*}^{-\alpha}\:dv_* &\leq 2^{\alpha}\abs{v}^{-\frac{\alpha(s_2-s_1)}{d}}\int_{C}f(v_*)\left(1+\abs{v_*}^{s_1}\right)\abs{v_*}^{-s_1}\:dv_*
\\&\leq 2^{s_1+\alpha}\abs{v}^{-s_1-\frac{\alpha(s_2-s_1)}{d}}\norm{f}_{L^1_{s_1,v}}.
\end{split}
\end{equation*}
Combining all of the above gives the desired result.
\bigskip

%%%%%%%%%%%%%%%%%%%%%%%%%%%%%%%%%%%%%%%%%%%%%%%%%%%%%%%%%%%%%%%%%%%%%%%%%%%%%%%%%%%%%%%%%%%%%%%%%%%%%%%%%%%%%%%%%%%%%%%%%%%%%%%%%%%%%%%%%%%%%%%%%%%%%%%%%%%%%%%%%%%
%%%%%%%%%%%%%%%%%%%%%%%%%%%%%%%%%%%%%%%%%%%%%%%%%%%%%%%%%%%%%%%%%%%%%%%%%%%%%%%%%%%%%%%%%%%%%%%%%%%%%%%%%%%%%%%%%%%%%%%%%%%%%%%%%%%%%%%%%%%%%%%%%%%%%%%%%%%%%%%%%%%
%%%%%%%%%%%%%%%%%%%%%%%%%%%%%%%%%%%%%%%%%%%%%%%%%%%%%%%%%%%%%%%%%%%%%%%%%%%%%%%%%%%%%%%%%%%%%%%%%%%%%%%%%%%%%%%%%%%%%%%%%%%%%%%%%%%%%%%%%%%%%%%%%%%%%%%%%%%%%%%%%%%

\subsection{Additional estimations}

\par Throughout this section we will denote by $\mathbb{S}^{d-1}_{r}(a)$ the sphere of radius $r$ and centre $a \in \R^d$.

\bigskip
\begin{lemma}\label{lem:integration of |v-v_1|^(-alpha) over the sphere}
For any $a\in \R^d$ and $r>0$ we have that if $0\leq \alpha \leq d-1$ then there exists $C_{d,\alpha}>0$ such that
\begin{equation}\label{eq:integration of |v-v_1|^(-alpha) over the sphere}
\int_{\mathbb{S}^{d-1}_{r}(a)}\abs{v-v_1}^{-\alpha}d\sigma(v_1) \leq C_{d,\alpha}r^{-\alpha}.
\end{equation}
\end{lemma}
\bigskip

\begin{proof}[Proof of Lemma $\ref{lem:integration of |v-v_1|^(-alpha) over the sphere}$]
We have that if $v_1\in \mathbb{S}^{d-1}_{r}(a) $ then
\begin{equation}\nonumber
\abs{v-v_1}^2=\abs{v-a}^2+r^2-2r\abs{v-a}\cos \theta = \left(\abs{v-a}-r\right)^2+2r\abs{v-a}\left(1-\cos\theta\right),
\end{equation}
where $\theta$ is the angle between the constant vector $v-a$ and the vector $v_1$. At this stage we'll look at two possibilities: $\abs{\abs{v-a}-r}>\frac{r}{2}$ and $\frac{r}{2}\leq \abs{v-a}\leq\frac{3r}{2}$.
\par In the first case we have that
\begin{equation}\nonumber
\abs{v-v_1} \geq \abs{\left(v-a\right)-r} \geq \frac{r}{2}
\end{equation}
implying that
\begin{equation}\nonumber
\int_{\mathbb{S}^{d-1}_{r}(a)}\abs{v-v_1}^{-\alpha}d\sigma(v_1)  \leq \left(\frac{r}{2}\right)^{-\alpha}\int_{\mathbb{S}^{d-1}}d\sigma(v_1)= 2^{\alpha}r^{-\alpha}.
\end{equation}

\bigskip
In the second case we have that 
\begin{equation}\nonumber
\abs{v-v_1} \geq \sqrt{2r\abs{v-a}\left(1-\cos\theta\right)} \geq \sqrt{2}r\sin\left(\frac{\theta}{2}\right)
\end{equation}
implying that
\begin{eqnarray*}
\int_{\mathbb{S}^{d-1}_{r}(a)}\abs{v-v_1}^{-\alpha}d\sigma(v_1)  &\leq& \left(\sqrt{2}r\right)^{-\alpha}C_d\int_{0}^{\pi}\frac{\sin^{d-2} (\theta)}{\sin^{\alpha}\left(\frac{\theta}{2}\right)}d\theta
\\&=&C_{d,\alpha}r^{-\alpha}\int_{0}^{\frac{\pi}{2}}\frac{\cos^{d-2}(\theta)\sin^{d-2}(\theta)}{\sin^{\alpha}(\theta)}d\theta.
\end{eqnarray*}
The last integration is finite if and only if $d-2-\alpha>-1$, which is valid in our case. The proof is thus complete.
\end{proof}
\bigskip

\begin{lemma}\label{lem:concentration of delta on a sphere}
Let $E$ be any hyperplane in $\R^d$ with $d\geq 3$ and let $a\in\R^d$ and $r>0$. Then
\begin{equation}\label{eq:concentration of delta on a sphere}
\sup_{n\in\mathbb{N}} \frac{1}{r^{d-2}}\int_{\mathbb{S}^{d-1}_r(a)}\varphi_n(x)ds(x) \leq \abs{\mathbb{S}^{d-2}}
\end{equation}
where $\varphi_n(x)=\left(\frac{n}{2\pi}\right)^{\frac{1}{2}}e^{-\frac{nD(x,E)^2}{2}}$ with $D(x,A)$ the distance of $x$ from the set $A$, $ds(x)$ is the appropriate surface measure.
\end{lemma}
\bigskip

\begin{proof}[Proof of Lemma $\ref{lem:concentration of delta on a sphere}$]
Due to translation, rotation and reflection with respect to $E$ we may assume that $E=\left\{x\in \R^d,\quad x_d=0 \right\}$ and that $a=\abs{a}\hat{e}_d$. In that case
\begin{equation}\nonumber
\varphi_n(x)=\sqrt{\frac{n}{2\pi}}e^{-\frac{nx_d^2}{2}}
\end{equation}
and on $\mathbb{S}^{d-1}_r(a)$ we find that 
\begin{equation}\nonumber
\varphi_n(a+r\omega)=\sqrt{\frac{n}{2\pi}}e^{-\frac{n\left(\abs{a}+r\cos\theta\right)^2}{2}}%=\sqrt{\frac{n}{2\pi}}e^{-\frac{n\abs{a}^2}{2}}e^{-n\abs{a}r\cos\theta}e^{-\frac{nr^2\cos^2\theta}{2}},
\end{equation}
where $\theta$ is the angle with respect to the $\hat{e}_d$ axis. Thus
\begin{equation}\label{eq:Cd=Sd-2}
\frac{1}{r^{d-2}}\int_{\mathbb{S}^{d-1}_r(a)}\varphi_n(x)ds(x) =\abs{\mathbb{S}^{d-2}}\frac{\sqrt{n}r}{\sqrt{2\pi}}\int_{0}^{\pi}e^{-\frac{n\left(\abs{a}+r\cos\theta\right)^2}{2}}\sin^{d-2}\theta d\theta
\end{equation}
Using the change of variables $x=\sqrt{n}r\cos \theta$ yields 
\begin{equation}\nonumber
\begin{gathered}
\frac{1}{r^{d-2}}\int_{\mathbb{S}^{d-1}_r(a)}\varphi_n(x)ds(x) =\abs{\mathbb{S}^{d-2}}\frac{1}{\sqrt{2\pi}}\int_{-\sqrt{n}r}^{\sqrt{n}r}e^{-\frac{\left(\sqrt{n}\abs{a}+x\right)^2}{2}}\left(1-\frac{x}{nr^2}\right)^{\frac{d-3}{2}}dx \\
\leq  \frac{\abs{\mathbb{S}^{d-2}}}{\sqrt{2\pi}}\int_{\R}e^{-\frac{\left(\sqrt{n}\abs{a}+x\right)^2}{2}}dx = \abs{\mathbb{S}^{d-2}},
\end{gathered}
\end{equation}
completing the proof.
\end{proof}
%\begin{remark}
%One might ask oneself how did $\abs{\mathbb{S}^{d-2}}$ appear in \eqref{eq:Cd=Sd-2}. This follows form the following simple observation:
%$$\int_{\mathbb{S}^{d-1}}f(x_d)ds(x)=\int_{[0,2\pi]^{d-2}}\int_{0}^{\pi} J(\theta_1,\dots,\theta_{d-2})f(\sin \theta)sin^{d-2}\theta d\theta_1\dots d\theta_{d-2} d\theta$$
%$$=C_d \int_{0}^\pi f(\sin \theta)\sin^{d-2}\theta d\theta_1\dots d\theta_{d-2}d\theta.$$
%with $C_d$ a geometric constant. Applying this to $f(x)=1$ yields
%$$\abs{\mathbb{S}^{d-1}}=C_d \int_{0}^{\pi} \sin^{d-2}\theta d\theta = C_d \left(2\int_{0}^{\frac{\pi}{2}}\sin^{2\left(\frac{d-1}{2}-1 \right)}\cos^{2\left(\frac{1}{2} \right)-1} \right)$$
%$$=C_dB\left(\frac{d-1}{2},\frac{1}{2} \right)=C_d \frac{\Gamma \left(\frac{d-1}{2} \right)\Gamma \left(\frac{1}{2} \right)}{\Gamma \left(\frac{d}{2} \right)}$$
%As $\abs{\mathbb{S}^{d-1}}=\frac{2\pi^{\frac{d}{2}}}{\Gamma \left(\frac{d}{2} \right)}$ and $\Gamma \left(\frac{1}{2} \right)=\sqrt{\pi}$ we conclude the desired result.
%\end{remark}
\bigskip
\begin{lemma}\label{lem:decay differential inequality}
Assume that $\psi$ satisfied 
\begin{equation}\nonumber
\psi' \leq -C_1(1+\abs{v})^\alpha \psi + C_2 \psi+C_3\left(1+\abs{v}\right)^{-\beta}
\end{equation}
when $C_1,C_2,C_3>0$. 
\\Then for any $0<t<T$ and $\abs{v} \geq \left(\frac{2C_2}{C_1}\right)^{\frac{1}{\alpha}}-1$ 
\begin{equation}\nonumber
\left(1+\abs{v}\right)^{\alpha+\beta}\psi(t) \leq \left(1+\abs{v}\right)^{\alpha+\beta}\psi(0)+\frac{2C_3}{C_1}. 
\end{equation}
\end{lemma}
\bigskip

\begin{proof}[Proof of Lemma $\ref{lem:decay differential inequality}$]
Defining $\phi(t)=e^{\left(C_1(1+\abs{v})^\alpha  - C_2\right)t }\psi(t)$, we find that
\begin{equation}\nonumber
\phi' \leq \frac{C_3}{\left(1+\abs{v}\right)^\beta}e^{\left(C_1(1+\abs{v})^\alpha  - C_2\right)t }
\end{equation}
Using the assumption on $v$, which is equivalent to
\begin{equation}\nonumber
C_1(1+\abs{v})^\alpha  - C_2>\frac{C_1}{2}(1+\abs{v})^\alpha >0
\end{equation}
we find that
\begin{equation}\nonumber
\begin{split}
\psi(t) &\leq e^{-\left(C_1(1+\abs{v})^\alpha  - C_2\right)t}\psi(0) + \frac{C_3}{\left(1+\abs{v}\right)^\beta \left(C_1(1+\abs{v})^\alpha  - C_2\right)}\left(1-e^{-\left(C_1(1+\abs{v})^\alpha  - C_2\right)}t\right)
\\&\leq \psi(0) + \frac{2C_3}{C_1\left(1+\abs{v}\right)^{\alpha+\beta}}.
\end{split}
\end{equation}
from which the result follows.
\end{proof}
\bigskip

%%%%%%%%%%%%%%%%%%%%%%%%%%%%%%%%%%%%%%%%%%%%%%%%%%%%%%%%%%%%%%%%%%%%%%%%%%%%%%%%%%%%%%%%%%%%%%%%%%%%%%%
%%%%%%%%%%%%%%%%%%%%%%%%%%%%%%%%%%%%%%%%%%%%%%%%%%%%%%%%%%%%%%%%%%%%%%%%%%%%%%%%%%%%%%%%%%%%%%%%%%%%%%%
%%%%%%%%%%%%%%%%%%%%%%%%%%%%%%%%%%%%%%%%%%%%%%%%%%%%%%%%%%%%%%%%%%%%%%%%%%%%%%%%%%%%%%%%%%%%%%%%%%%%%%%

\section{Propagation of Weighted $L^\infty$ Norms for the Truncated Operators}\label{appsec:weighted}
Here we will discuss the propagation of the weighted $L^\infty$ norms for the constructed sequence $\left\{f_{j,n} \right\}_{j\in\N}$.
\begin{lemma}\label{lemapp:prop of moments}
Consider the sequence defined in \eqref{inductionseq}. Let $s>2$ and let $C_s$ be a uniform constant such that 
$$\abs{v'}^s+\abs{v'_*}^s-\abs{v}^s-\abs{v_*}^s \leq C_s \abs{v}^{s-1}\abs{v_*} $$
Then for any $j \geq j_0= 2(1+M_2)C_s/M_0$ we have that
\begin{equation}\label{appeq:prop of moments}
\int_{\R^d}(1+\abs{v}^s) f_{j,n}^{(k)}(v)dv \leq (D_s k\Delta_j +1) \int_{\R^d}(1+\abs{v}^s) f_0(v)dv,
\end{equation}
where $D_s=4C_\Phi C_s l_b(1+2K_\infty)(1+M_2)$.
\end{lemma}
 In what follows we will drop the subscript ${j,n}$ from the proofs to simplify the notation
 \begin{proof}
 The proof, as usual, goes by induction.The step $k=0$ is immediate. Assuming the claim is valid for $k$ we have that
 $$\int_{\R^d}(1+\abs{v}^s) f^{(k+1)}(v)dv = \int_{\R^d}(1+\abs{v}^s) f^{(k)}(v)dv + \Delta_j \int_{\R^d}\abs{v}^s Q_n\left(f^{(k)}\right)(v)dv$$
 $$\leq \int_{\R^d}(1+\abs{v}^s) f^{(k)}(v)dv+C_\Phi C_s l_b(1+2K_\infty)\Delta_j\int_{\R^d\times \R^d}\left(\abs{v}^s\abs{v_*}+\abs{v}^{s-1}\abs{v_*}^2 \right)f^{(k)}(v)f^{(k)}(v_*)dvdv_*$$
 where we have used a Povzner inequality much like Proposition \ref{prop:prop of moments}. Thus,
  $$\int_{\R^d}(1+\abs{v}^s) f^{(k+1)}(v)dv \leq \left(1+ 2C_\Phi C_s l_b(1+2K_\infty)(1+M_2)\Delta_j\right)\int_{\R^d}(1+\abs{v}^s) f^{(k)}(v)dv $$
  $$\leq \left(1+ 2C_\Phi C_s l_b(1+2K_\infty)(1+M_2)\Delta_n\right)(D_s k\Delta_n +1) \int_{\R^d}(1+\abs{v}^s) f_0(v)dv.$$
  The proof follows form the choice of $D_s$.
 \end{proof}
Much like in Section \ref{sec:apriori} we denote by 
$$\psi_a(v)=\begin{cases}
0 &\abs{v}<\abs{a} \\
1 & \abs{v}\geq \abs{a}
\end{cases}$$
We have the following:
\begin{lemma}\label{lem:towards Linfty bounds}
Consider the sequence defined in \eqref{inductionseq}. We have that
\begin{itemize}
\item[(i)] if $r\geq 2$ is such that $\int_{\R^d} (1+\abs{v}^r)f_0(v)dv <\infty$ then for any $j\geq j_0$
$$\int_{E_{vv'}}\psi_v(v'_*)f_{j,n}^{(k)}(v_*')dE(v_*') \leq \int_{E_{vv'}}\psi_v(v'_*)f_0(v_*')dE(v_*') + B_\infty k \Delta_j(1+\abs{v})^{-r+\gamma-1},$$
where $B_\infty=C_\Phi C_{d,\gamma}b_\infty \left(1+D_r+K_\infty \right)^3$ and $C_{d,\gamma}$ is a uniform constant defined in Lemma \ref{lem:weight propagation preperation I}. Moreover, one can choose
$$r=\begin{cases}
2, & \mbox{if}\:s\leq d+2 \\
s^\prime, & \mbox{if}\: s>d+2\:\mbox{and}\:s^\prime<s-d
\end{cases}$$
\item[(ii)] If $\epsilon$ is small enough
$$\int_{E_{vv'}}\psi_v(v'_*)f_{j,n}^{(k)}(v_*')dE(v_*') \leq \left(C_{s,\epsilon}\norm{f_0}_{L^\infty_{s,v}}+ B_\infty k \Delta_j\right)(1+\abs{v})^{-(s-d+1-\epsilon-\gamma)},$$
where $C_{s,\epsilon}$ is a uniform constant that depends only on $s$ and $\epsilon$.
\item[(iii)] If $f_0\in L^\infty_{s,v}$ when $s>d+2\gamma$ then for any $s^\prime<s-2\gamma$
$$W_{s^\prime}=\sup_{k,j\geq j_0,n}\norm{f_{j,n}^{(k)}}_{L^\infty_{s^\prime,v}} < \infty.$$
\end{itemize}
\end{lemma}
\begin{proof}
All the proofs will follow by induction. The step $k=0$ is trivial.\\
$(i)$ Using Lemma \ref{lemapp:prop of moments}, Lemma \ref{lem:weight propagation preperation I}, the fact that $Q_n^- \geq 0$ and $Q_n^{+}\leq Q^+$ we have that
\begin{equation*}
\begin{split}
&\int_{E_{vv'}}\psi_v(v'_*)f^{(k+1)}(v_*')dE(v_*') \leq \int_{E_{vv'}}\psi'_{v,*}f^{(k)\prime}_*dE(v_*') +\Delta_j\int_{E_{vv'}}\psi'_{v,*}Q^+\left(f^{(k)\prime}_*\right)dE(v_*')
\\&\:\leq \int_{E_{vv'}}\psi'_{v,*}f'_{0,*}dE(v_*') +  \frac{B_\infty k\Delta_j}{(1+\abs{v})^{r-\gamma+1}}+C_\Phi C_{d,\gamma}b_\infty \left(1+D_r+K_\infty \right)^3\frac{\Delta_j}{(1+\abs{v})}^{r-\gamma+1}.
\end{split}
\end{equation*}
The choice of $r$ follows the remark at the beginning of the proof of Proposition \ref{prop:weight propagation preperation II}.

\bigskip
$(ii)$ follows much like Remark \ref{rem:improved formula}.
\par The proof of $(iii)$ follows the same method of the proof of Theorem \ref{theo:apriori}, with a few small changes to give a uniform bound on the weighted norm that will be independent of the truncation. As seen in the aforementioned proof, together with $(ii)$
$$Q_n^{+}(f^{(k)})(v)\leq Q^{+}(f^{(k)})(v) \leq C_{0,s,\epsilon}\left(1+\abs{v}\right)^{-\delta}$$
where $C_0$ depends only on $s$, $\epsilon$, the initial data and the collision parameters, and
$$\delta=\min \left\{s-2\gamma-\epsilon,s-d+1-\gamma-\epsilon+\frac{2(1+\gamma)}{d} \right\}.$$
As such, we find that
$$f^{(k+1)}(v) \leq f^{(k)}+Q^+(f^{(k)})(v) \leq  f^{(k)}+ C_{0,s,\epsilon}\left(1+\abs{v}\right)^{-\delta},$$
implying that one can prove inductively that there exists a constant $\widetilde{W}_1$ that depends only on the $s$, $\epsilon$, the initial data and the collision parameters such that
$$f^{(k)}(v) \leq \widetilde{W}_1 k\Delta_j \left(1+\abs{v}\right)^{-\delta}+ f_0(v).$$
This implies that, with the notations of $(iii)$, $W_\delta<\infty$. At this point we continue by induction and by using Lemma $\ref{lem:integration lemma part II}$ with $s_1=\delta$. Note that the process can not go beyond $s-2\gamma$. Denoting by $\xi=s-d+1-\gamma+\frac{2(1+\gamma)}{d}$, we see that the process can continue until
$$s^{\prime\prime}<\xi\sum_{j=0}^\infty \left(\frac{d-1-\gamma}{d} \right)^j=\frac{d}{1+\gamma}\xi. $$
Because $s>d+\gamma$ the above is bigger than $s-2\gamma$ which means that we will reach the desired result in finitely many steps, completing the proof.
\end{proof}

% Pour une biblio classe
\bibliographystyle{acm}
\bibliography{bibliography}

%% Pour une biblio manuelle
%\newpage
%\include{bibliography}

% On met les signatures
\bigskip
\signmb
\signae

\end{document}